\documentclass[a4paper,nofootinbib,superscriptaddress,notitlepage,longbibliography]{revtex4-1}
\pdfoutput=1

\usepackage{color}
\usepackage{amsmath,amsfonts,amssymb,amsthm}
\usepackage{graphicx}
\usepackage[english]{babel}
\usepackage{fullpage}
\usepackage[caption=false]{subfig}
\usepackage{enumerate}
\usepackage{relsize}
\usepackage{float}

\usepackage{tikz}
\usetikzlibrary{arrows,decorations.pathmorphing,backgrounds,positioning,fit}
\usepackage{epstopdf} 
\usepackage{bm}  
\usepackage{pgfplots}
\pgfplotsset{compat=1.11}
\usetikzlibrary{external} 

\usepackage[pdfpagelabels,pdftex,bookmarks,breaklinks]{hyperref}
\definecolor{darkblue}{RGB}{0,0,127} 
\definecolor{darkgreen}{RGB}{0,150,0}
\hypersetup{colorlinks, linkcolor=darkblue, citecolor=darkgreen, filecolor=red, urlcolor=darkblue}
\hypersetup{pdfpagemode=UseNone} 
\hypersetup{pdfstartview=FitH}

\graphicspath{{./Figures/}}

\newtheorem{theorem}{Theorem}
\newtheorem{result}[theorem]{Result}

\newtheorem{lemma}[theorem]{Lemma}
\newtheorem{corollary}[theorem]{Corollary}

\newtheorem{definition}[theorem]{Definition}

\renewenvironment{proof}[1][Proof]{\noindent\textbf{#1.} }{\ $\Box$}

\def\R{\mathbb{R}}
\def\C{\mathbb{C}}

\def\e{\mathrm{e}}

\newcommand{\Eref}[1]{Eq.~(\ref{#1})}
\newcommand{\Sref}[1]{Sec.~\ref{#1}}
\newcommand{\Fref}[1]{Fig.~\ref{#1}}

\newcommand{\Thref}[1]{Thm.~\ref{#1}}

\newcommand{\Lref}[1]{Lemma~\ref{#1}}
\newcommand{\Cref}[1]{Corollary~\ref{#1}}
\newcommand{\Dref}[1]{Definition~\ref{#1}}

\def\th{^{\rm th}}

\def\nd{^{\rm nd}}

\DeclareMathOperator{\Tr}{Tr}

\newcommand{\ket}[1]{|{#1}\rangle}
\newcommand{\expect}[1]{\left\langle{#1}\right\rangle}

\newcommand{\ketbra}[2]{|{#1}\rangle\!\langle{#2}|}

\newcommand{\mbf}[1]{\mathbf{#1}}
\newcommand{\abs}[1]{\left|{#1}\right|}

\newcommand{\sub}[1]{{\!\mathsmaller{#1}}}

\begin{document}

\title{Topological quantum error correction in the Kitaev honeycomb model}

\author{Yi-Chan Lee}
\affiliation{Centre for Engineered Quantum Systems, School of Physics, University of Sydney, Sydney, NSW 2006, Australia}
\affiliation{Centre for Quantum Software and Information,
Faculty of Engineering and Information Technology,
University of Technology Sydney, NSW 2007, Australia}
\affiliation{Physics Department, National Tsing-Hua University, Hsinchu 300, Taiwan}
\author{Courtney G. Brell}
\affiliation{Institut f\"{u}r Theoretische Physik, Leibniz Universit\"{a}t Hannover, Appelstra\ss{}e 2, 30167 Hannover, Germany}
\affiliation{Perimeter Institute for Theoretical Physics, Waterloo, ON, Canada}
\author{Steven T. Flammia}
\affiliation{Centre for Engineered Quantum Systems, School of Physics, University of Sydney, Sydney, NSW 2006, Australia}
\affiliation{Center for Theoretical Physics, Massachusetts Institute of Technology, Cambridge, MA, USA}

\date{August 31, 2017}

\begin{abstract}
The Kitaev honeycomb model is an approximate topological quantum error correcting code in the same phase as the toric code, but requiring only a 2-body Hamiltonian. 
As a frustrated spin model, it is well outside the commuting models of topological quantum codes that are typically studied, but its exact solubility makes it more amenable to analysis of effects arising in this noncommutative setting than a generic topologically ordered Hamiltonian. 
Here we study quantum error correction in the honeycomb model using both analytic and numerical techniques. 
We first prove explicit exponential bounds on the approximate degeneracy, local indistinguishability, and correctability of the code space. 
These bounds are tighter than can be achieved using known general properties of topological phases. 
Our proofs are specialized to the honeycomb model, but some of the methods may nonetheless be of broader interest. 
Following this, we numerically study noise caused by thermalization processes in the perturbative regime close to the toric code renormalization group fixed point. 
The appearance of non-topological excitations in this setting has no significant effect on the error correction properties of the honeycomb model in the regimes we study. 
Although the behavior of this model is found to be qualitatively similar to that of the standard toric code in most regimes, we find numerical evidence of an interesting effect in the low-temperature, finite-size regime where a preferred lattice direction emerges and anyon diffusion is geometrically constrained. We expect this effect to yield an improvement in the scaling of the lifetime with system size as compared to the standard toric code.
\end{abstract}

\maketitle


\section{Introduction}

Topologically ordered systems are a promising avenue to implement quantum information storage in a physical system~\cite{Kitaev2003, Dennis2002}. 
These systems have degenerate ground spaces in which information can be encoded such that it is robust to local perturbations~\cite{Bravyi2010, Bravyi2011stability, Michalakis2013}. 
These kinds of systems can also be used to perform topologically protected quantum computation~\cite{Nayak2008}. 
The prototypical example of a topological spin model is the toric code~\cite{Kitaev2003}. 
This model involves local interactions in 2D, and gives rise to abelian anyonic quasiparticles. 
The toric code has many desirable features, such as exact degeneracy of the Hamitonian eigenspaces for all system sizes, harmonic spectrum, flat dispersion relation, and ground states that are renormalization group (RG) fixed points~\cite{Levin2005,Aguado2008a}. 
This wealth of special structure makes the analysis~\cite{Dennis2002, Alicki2009, Viyuela2012, Freeman2014, Jouzdani2014} and simulation~\cite{Brown2016, darmawan2016tensor} of noise and error correction protocols tractable, at least for certain types of noise. 
However, the toric code is experimentally challenging to implement due to the fact that it involves 4-body interactions.
Additionally, the special structure that enables such detailed analysis might obscure some features of error correction and noise that may be present in more general models.

The Kitaev honeycomb model~\cite{Kitaev2006} is a topologically ordered system involving only 2-body interactions, and is considered an experimentally plausible model to implement in many physical systems~\cite{Duan2003, Micheli2006,Zhang2007,Aguado2008creation,Jiang2008, You2010}. 
For suitable parameter regimes, the honeycomb model is in the same quantum phase as the toric code in the sense that the two are adiabatically connected, and this guarantees that the ground space properties and quasiparticle excitations of the two models are closely related~\cite{Kitaev2006, Hastings2005, Hastings2007, Dusuel2008, Schmidt2008, Vidal2008}. 
In other parameter regimes, the honeycomb model is also a promising avenue for the realization of a non-abelian topological phase, the Ising anyons, which has been proposed as a candidate for topological quantum computing~\cite{Bravyi2006, Freedman2006}. 
Moreover, while the special properties such as the exact degeneracy, harmonic spectrum, and RG fixedness mentioned above are all lost in the honeycomb model, it is still exactly solvable~\cite{Kitaev2006, Feng2007, Chen2008, Kells2009, Pedrocchi2011, Schmoll2017}. 
For these reasons it is desirable to understand the performance of the honeycomb model as a topological quantum memory. 

Here we study quantum error correction in the honeycomb model using a combination of analytic and numerical methods. 
Our main results are summarized in the next subsection. 
When considering the honeycomb model rather than the toric code, the behavior as a quantum memory will be affected by several distinct effects. 
Firstly, though the ground spaces of the honeycomb model and the toric code where the quantum information is stored are related, they are not identical, which will lead to different responses to local operations or perturbations. 
Additionally, on finite system sizes, the code space of the honeycomb model is not precisely degenerate~\cite{Kells2008, Kells2009a}, which will result in dephasing effects between encoded states. Another significant source of distinction between the two systems is that we expect them to be subject to different noise channels in the lab. 
Topological quantum memories are typically assumed to be exposed to some kind of thermal noise source, and so the details of the noise processes will depend heavily on the structure of the spectrum of the Hamiltonian.

We expect the phenomena seen in the honeycomb model to be representative of more general topological phases away from the finely tuned models typically studied. 
While many properties of topological phases are universal, the details and performance of error correction protocols can be expected to vary as we move around a given topological phase. 
Existing studies of topological error correction~\cite{Dennis2002, Duclos-Cianci2010, Duclos-Cianci2010a,Wang2010a, Wang2010, Duclos-Cianci2013, Bravyi2011haah, Bombin2012, Wootton2012, Anwar2014, Watson2014, Hutter2014a, Bravyi2014, Wootton2015, Fowler2015, Andrist2015,Brell2013, Wootton2013, DuclosCianci2013, Hutter2014, Hutter2015continuous, Wootton2015b, Brell2017, herold2015fault, herold2015cellular, dauphinais2016fault, torlai2016neural, darmawan2016tensor} are almost exclusively focused on particular representatives of a given phase, usually the RG fixed point, where in 2D errors have a natural interpretation as precisely localized quasi-particles that can be analyzed using the tools of topological quantum field theory. 
Here we analyze some of the relevant phenomenology associated with being at a point in a topologically ordered phase that is \emph{not} an RG fixed point, and in particular the effect of these phenomena on quantum information storage protocols. 
The advantage of studying such phenomena within the honeycomb model is that it retains some tractability that will be lost in the general case, and so at least some additional insights can be gained without appealing to challenging simulations of general quantum many-body systems. 

\subsection*{Summary of results}

Our study of quantum error correction in the honeycomb model divides the novel phenomenology into two broad classes: coherent and incoherent effects. We treat these separately using analytic methods for the coherent errors in \Sref{s:coherent} and numerical methods for the incoherent errors in \Sref{s:incoherent}. 

Coherent effects are those related to the approximate degeneracy of the ground space and the approximate local indistinguishability of different code states. We will provide explicit exponential bounds on these approximations. 
The main theorems can be summarized as follows; see \Sref{s:coherent} for more precise and detailed formulations. 
Our first result is a bound on the energy splitting among the four nearly degenerate ground states on the honeycomb model with linear size $N$ (i.e.~$2 N^2$ spins total) with toroidal boundary conditions.

\begin{result}[Approximate degeneracy, \Thref{Th:qsdegeneracy}]
	The energy splitting $\Delta E$ between any pair of ground states in the four-fold ground space of the honeycomb model is bounded by
	\begin{align}
		\Delta E \leq c_1 N^2 \mathrm{e}^{-aN} \,,
	\end{align}
	where $N$ is the linear size of the lattice, and $c_1, a > 0$ are constants determined by the coupling coefficients of the honeycomb Hamiltonian.
\end{result}

Thus, the dephasing rate between superpositions of ground states is exponentially small in the system size. It is also worth noting that the constant $a$ appearing here and elsewhere is strictly positive in the entire gapped phase of the honeycomb model. Our next result bounds the difference in expectation values of operators supported on sufficiently small regions.

\begin{result}[Approximate local indistinguishability, \Thref{T:arbitrary_operator}]
	For an arbitrary operator $O_A$ located within any simply connected local region $A$ of size $|A|\leq c_2 N/\log N$, the difference in expectation value between any pair of ground states $\ket{\psi_1}$ and $\ket{\psi_2}$ of the honeycomb model is bounded as 
	\begin{align}
		\abs{\expect{O_\sub{A}}_{\psi_1}-\expect{O_\sub{A}}_{\psi_2}}\leq c_3\, \mathrm{e}^{-a N/4} \|O_\sub{A}\| \,,
	\end{align}
	for constants $c_2, c_3, a >0$ determined by the coupling coefficients.
\end{result}

Intuitively, this suggests that the honeycomb model ground space is a quantum code capable of correcting errors in these regions $A$. 
This is indeed the case, and our next result allows us to bound the error $\epsilon(N)$ of a recovery map applied to the code space where the error is measured in the Bures distance; see Definition~\ref{d:epscorrect}.

\begin{result}[Approximate correctability, \Thref{T:localcorrect}]
	For any noise channel supported in any simply connected region $A$ with size $|A|\leq c_4 N/\log N$, the code space is $\epsilon(N)$-correctable in the Bures distance with
	\begin{align}
		\epsilon(N) \le  c_5\, \e^{-aN/16} \,,
	\end{align}
	for constants $c_4, c_5, a >0$ determined by the coupling coefficients.
\end{result}

We provide explicit expressions for the constant $a$ in \Eref{eq:adef} and for $c_1, \ldots, c_5$ in the formal statements of these theorems, \Thref{Th:qsdegeneracy}, \Thref{T:arbitrary_operator}, and \Thref{T:localcorrect} respectively.

The second class of phenomena we study arises from incoherent effects related to the interaction of the system with a thermalizing environment. 
In \Sref{s:incoherent} we employ numerical tools to study these effects and to evaluate the performance of a standard error correction protocol in their presence, assuming ideal syndrome measurements. 
Our results can be summarized as follows:
\begin{itemize}
	\item In the regimes studied, the appearance of non-topological excitations does not qualitatively affect the error correction properties.
	\item Asymptotically, the memory lifetime scales exponentially in inverse temperature, as expected~\cite{Brown2016}.
	\item In the low-temperature finite-size regime, where the diffusion of a single pair of topological excitations is the dominant failure mode, we anticipate an improved scaling of the memory lifetime with system size as compared to the toric code.
\end{itemize}
This final effect is a result of the structure of the honeycomb Hamiltonian energetically suppressing some local errors more than others. In the low-temperature, finite-size, perturbative regime we probe, with strictly local couplings to the thermal bath, this leads the topological excitations to be constrained to 1D motion instead of the 2D motion of the standard toric code~\cite{Freeman2014, Brown2016}. 
Though this effect relies on the special structure of this regime, it may be exploitable in suitable experimental implementations of a topological quantum memory.

The remainder of the paper is organized as follows. 
We review the honeycomb model and establish notation, terminology, and some known results in \Sref{s:honeycombModel}. 
Then \Sref{s:coherent} reviews known results from the general theory of topological phases before proving the three Results listed above. 
In \Sref{s:incoherent} we discuss our numerical simulation model for quantum error correction and the results of our simulations. 
We conclude in \Sref{s:discussion}.

\section{The honeycomb model}\label{s:honeycombModel}

The Kitaev honeycomb model is defined with qubits residing on the vertices of a honeycomb lattice, as shown in \Fref{f:honeycomb}. In our treatment, the vertical edges will be preferred, and it will be convenient to consider the (square) lattice $\Lambda$ formed by placing a site at each vertical edge. These vertical edges will often be referred to as ``dimers'', motivated by the perturbative treatment of this model sketched below~\cite{Dusuel2008, Schmidt2008, Vidal2008}. At each site $\mbf{q}=(q_x,q_y)$ of $\Lambda$, there are two qubits, corresponding to the lower and upper ends of the dimer (denoted $\bullet$ and $\circ$ respectively). We also choose the boundary conditions of $\Lambda$ to be periodic, so that the edges are identified as shown in \Fref{f:honeycomb}, and denote its linear dimensions by $N$ (taken to be even for convenience). The Hamiltonian of the model is given by
	\begin{align}\label{E:KHM}
	H = - \sum_{\mbf{q}} (J_x K^x_\mbf{q}
						+ J_y K^y_\mbf{q}
						+ J_z K^z_\mbf{q}) 
	\end{align}
	for some coupling constants $J_\alpha>0\;$ ($\alpha\in\{x,y,z\}$), and with
	\begin{align}
		K^x_\mbf{q} &= \sigma^x_{\mbf{q},\circ}\sigma^x_{\mbf{q}+\mbf{n}_x,\bullet}	\,,\\
		K^y_\mbf{q} &= \sigma^y_{\mbf{q},\circ}\sigma^y_{\mbf{q}+\mbf{n}_y,\bullet}	\,,\\
		K^z_\mbf{q} &= \sigma^z_{\mbf{q},\circ}\sigma^z_{\mbf{q},\bullet}	\,,
	\end{align}
	where $\sigma^{\alpha}_i$ is the relevant Pauli operator at site $i$.
	We will call the $K^{\alpha}_{\mbf{q}}$ link \emph{generators} to distinguish them from the more general class of link \emph{operators} that they generate under multiplication.
	
	\begin{figure}
		\centering
		\includegraphics{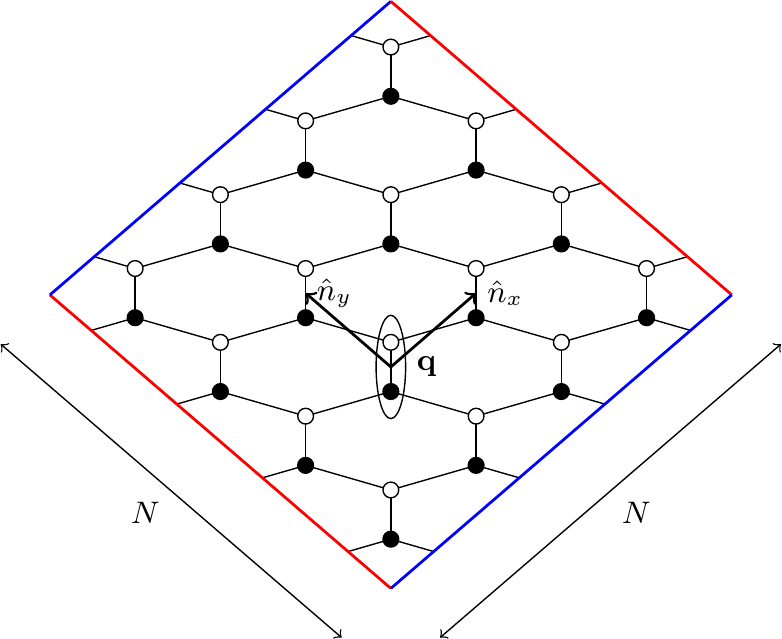}
		\caption{The vertical links of the honeycomb lattice form a square lattice with translation vectors as shown. We consider an $N\times N$ periodic square lattice such that blue and red boundaries are identified in the diagram. Each site of this square lattice contains two honeycomb vertices, and we denote the upper by $\circ$ and the lower by $\bullet$. Hence there are a total of $2 N^2$ spins in the lattice.}
		\label{f:honeycomb}
	\end{figure}	
	
	\subsection{Symmetries}\label{s:sym}
	
	Though the terms in this Hamiltonian do not commute with one another, constants of motion can be constructed as products of these terms. In particular, for each site $\mbf{q}$ of $\Lambda$ we can define an operator
\begin{align}
		W_{\mbf{q}} &= K^y_\mbf{q} K^z_{\mbf{q}+\mbf{n}_y} K^x_{\mbf{q}+\mbf{n}_y} K^y_{\mbf{q}+\mbf{n}_x} K^z_{\mbf{q}+\mbf{n}_x} K^x_\mbf{q}\\
					&= \sigma^z_{\mbf{q},\circ} \sigma^x_{\mbf{q}+\mbf{n}_y,\bullet} \sigma^y_{\mbf{q}+\mbf{n}_y,\circ} \sigma^z_{\mbf{q}+\mbf{n}_x+\mbf{n}_y,\bullet} \sigma^x_{\mbf{q}+\mbf{n}_x,\circ} \sigma^y_{\mbf{q}+\mbf{n}_x,\bullet}
\end{align}
such that $[H,W_{\mbf{q}}]=0$ for all $\mbf{q}$ and $[W_{\mbf{q}},W_{\mbf{q'}}]=0$ for all $\mbf{q},\mbf{q'}$. These $W_{\mathbf{q}}$ are naturally associated to the hexagonal plaquette of the honeycomb lattice above site $\mathbf{q}$. 
Two other privileged operators can also be defined as
	\begin{align}
		L_x &= \prod_{q_x} K^z_{q_x,0} K^x_{q_x,0} = -\prod_{q_x}\sigma^y_{(q_x,0),\bullet} \sigma^y_{(q_x,0),\circ}\,,\\
		L_y &= \prod_{q_y} K^z_{0,q_y} K^y_{0,q_y} = -\prod_{q_y}\sigma^x_{(0,q_y),\bullet} \sigma^x_{(0,q_y),\circ}\,.
	\end{align}

	Again, it can easily be verified that these operators commute with $H$, with each other, and with all $W_{\mbf{q}}$. The honeycomb model can be considered as a subsystem code with gauge generators identified with terms in the Hamiltonian (i.e.~the link generators), and in this setting the stabilizer group of the code is generated by the $W_{\mbf{q}}$, $L_x$, and $L_y$, while the logical algebra is trivial~\cite{Suchara2011}. However, this is not the interpretation we will use of this model as a quantum code.
		
	The honeycomb model can be in several phases for different regimes of the parameters $J_x, J_y, J_z$. 
In this work, we will consider the regime
\begin{align}
\label{eq:couplingconstants}
	J_z > J_x \geq J_y > 0 \quad \text{and} \quad J_z > J_x + J_y \,,
\end{align}
where the model is in the same phase as the toric code~\cite{Kitaev2006}. 
We will often refer to this as the ``abelian'' or ``gapped'' phase in order to distinguish it from the qualitatively different gapless non-abelian phase that can be found in a different parameter regime of the model. For choices of parameters that lead to a gap, the conditions (\ref{eq:couplingconstants}) can be imposed without loss of generality. 
In the abelian phase, there is a 4-dimensional ground space that is separated from other Hamiltonian eigenspaces by a constant gap independent of the system size. 
As $N$ diverges, this ground space will converge to an exact degeneracy of its energy eigenvalues. 
A basis for this space can be labeled by $\pm1$ eigenvalues of the $L_x$ and $L_y$ operators, and all such states will be $+1$ eigenstates of all $W_{\mbf{q}}$ operators. 
When considering the honeycomb model in this phase as a quantum memory, we will refer to the four-fold ground space as the code space $\mathcal{C}_{\mathrm{khm}}$, and as such the $L_x$ and $L_y$ operators will be logical operators for this code. 
Eigenstates of $L_x$ and $L_y$ within the code space will be called code basis states, and be labeled by $\ell = (l_x, l_y)$ for $l_x, l_y = \pm 1$ the eigenvalues of $L_x$ and $L_y$ respectively.
	
	\subsection{Exact ground states}\label{s:gs}
		
		While the honeycomb model was originally solved by Majorana fermionization~\cite{Kitaev2006}, it can also be solved in several other ways~\cite{Feng2007, Chen2008, Kells2009,Pedrocchi2011, Schmoll2017}. The essential ingredient for our analysis of coherent noise follows the treatment of Ref.~\cite{Kells2009}, which solves the honeycomb model without using either Majorana fermions or redundant degrees of freedom. 
The starting point of this analysis is to describe the model in terms of hard-core bosons and effective spins, followed by fermionization of the hard-core boson degrees of freedom. 
Following common terminology, we will refer to $-1$ eigenstates of a $W_\mbf{q}$ operator as having a vortex at site $\mbf{q}$, and $-1$ eigenstates of a $K^z_{\mbf{q}}$ operator as having a broken dimer at link $\mbf{q}$. 
By projecting into common eigenspaces of all $W_{\mbf{q}}$, $L_x$, and $L_y$, we can solve the model sector by sector. 
Associating Dirac fermion creation operators $c_\mbf{q}^{\dagger}$ to each site of $\Lambda$, we can represent the link generators in the vortex-free sector as
		\begin{align}	\label{E:fermionization}
			K^x_\mbf{q} &\rightarrow -(c_\mbf{q}^\dagger - c_\mbf{q})(c_{\mbf{q}+\mbf{n}_x}^\dagger + c_{\mbf{q} + \mbf{n}_x})	\,,	\notag \\
			K^y_\mbf{q} &\rightarrow -(c_\mbf{q}^\dagger - c_\mbf{q})(c_{\mbf{q}+\mbf{n}_y}^\dagger + c_{\mbf{q} + \mbf{n}_y})	\,,\\
			K^z_\mbf{q} &\rightarrow I - 2 c_\mbf{q}^\dagger c_\mbf{q}	\,.	\notag
		\end{align}
		Within this sector, a Fourier transform yields a quadratic effective Hamiltonian
		\begin{align}
			H_{\mathrm{vf}} = \sum_\mbf{k} \left[ \xi_\mbf{k}c_\mbf{k}^\dagger c_\mbf{k} + \frac{1}{2} (\Delta_\mbf{k} c_\mbf{k}^\dagger c_\mbf{-k}^\dagger + \Delta_\mbf{k}^* c_\mbf{-k} c_\mbf{k}) \right] - N^2 J_z	\,,
		\end{align}
		with $\mbf{k} = (k_x, k_y)$, and where
		\begin{align}
			c_\mbf{k} &= \frac{1}{\sqrt{N^2}}\sum_\mbf{q} c_\mbf{q} e^{-i\mbf{k}\cdot\mbf{q}}\,,	\label{E:fermionFourier}\\
			\xi_\mathbf{k} &= 2J_x\cos(k_x) + 2J_y\cos(k_y) + 2J_z\,,\label{Eq:gsxi}\\
			\Delta_\mathbf{k} &= i\left[ 2J_x\sin(k_x) + 2J_y\sin(k_y) \right]	\,.\label{Eq:gsDelta}
		\end{align}
		$H_\mathrm{vf}$ can then be diagonalized as 
		\begin{align}
			H_{\mathrm{vf}} = \sum_{\mbf{k}} E_{\mbf{k}}\left(\gamma_\mbf{k}^\dagger \gamma_{\mbf{k}} - \frac{1}{2} \right)
		\end{align}
		by the Bogoliubov-Valatin transformation
		\begin{align}	\label{E:BogoliubovTrans}
			\gamma_\mbf{k} = u_\mbf{k} c_\mbf{k} - v_\mbf{k} c_{-\mbf{k}}^\dagger	\,,
		\end{align}
		where
		\begin{align}	\label{E:Bogoliubov_u}
			u_\mbf{k} &= \sqrt{\frac{1}{2} \left(1 + \frac{\xi_\mbf{k}}{E_\mbf{k}} \right)}\,,\\
			v_\mbf{k} &= \frac{\abs{\Delta_\mbf{k}}}{\Delta_\mbf{k}} \sqrt{\frac{1}{2} \left(1 - \frac{\xi_\mbf{k}}{E_\mbf{k}} \right)}	\,. \label{E:Bogoliubov_v}
		\end{align}
		The ground states are fermionic vacuum states annihilated by $\gamma_\mbf{k}$ for all $\mbf{k}$, and have energy
		\begin{align}
			E_\ell = \sum_{\mathbf{k}} \left( -\frac{1}{2} E_{\mathbf{k}} \right)	\,,	\label{Eq:gsEnergy}
		\end{align}
		where $\ell$ denotes the corresponding eigenvalues of $L_x$ and $L_y$, and where
		\begin{align}
			E_\mathbf{k} &= \sqrt{\xi^2_\mathbf{k} + \abs{\Delta_\mathbf{k}}^2 }\,.\label{Eq:gsModeEn}
		\end{align}
		Since the code basis state for each $\ell$ is the ground state of a free-fermion Hamiltonian, we can easily calculate expectation values in these states.
		
		Under periodic boundary conditions, the energy dependence on $\ell$ can be considered to arise from the periodicity (anti-periodicity) of $c_\mbf{q}$ in \Eref{E:fermionFourier} for $l_\alpha=-1$ (+1 respectively). This can conveniently be absorbed into the allowed momenta to give~\cite{Kells2009}
		\begin{align}	\label{E:gsMomentum}
			k_\alpha &= 2\pi\frac{n_\alpha}{N} + \left(\frac{l_\alpha+1}{2}\right)\frac{\pi}{N}	
		\end{align}
		for $\alpha \in \{x, y\}$ and $0\leq n_{\alpha}<N$.
		
		Note that in the thermodynamic limit, $N \rightarrow \infty$,
		\begin{align}
			k_x \xrightarrow{N \rightarrow \infty} x \in [0, 2\pi)	\,,
		\end{align}
		and similarly for $k_y$. Since the $l_\alpha$-dependent term in \Eref{E:gsMomentum} vanishes as $N\rightarrow \infty$, the momenta $x$ and $y$ become independent of $\ell$ for all ground states. Thus the honeycomb model has an exact four-fold degeneracy in the thermodynamic limit. 
In this limit the energy density per site of $\Lambda$ (i.e.~per dimer) becomes 
		\begin{align}
			\frac{E_\ell}{N^2} \xrightarrow{N \rightarrow \infty} \frac{1}{4\pi^2} \int_0^{2\pi}\int_0^{2\pi} -\frac{1}{2} E(x, y) \mathrm{d}x \mathrm{d}y	\,.	\label{Eq:gsEnDensThermoLimit}
		\end{align}

	\subsection{In the perturbative regime}\label{s:perturb}
	
 		Though the honeycomb model can be exactly solved, for example as outlined above, in parts of our study it will be more convenient to use a perturbative approach to analyzing the model, following Ref.~\cite{Vidal2008}. In the regime where $J_z\gg J_x,J_y$, the honeycomb model can be considered as a perturbation away from the toric code. In this limit, the unperturbed Hamiltonian is taken as $H_0=-J_z\sum_{\mbf{q}}K^z_{\mbf{q}}$ and the perturbation 
		\begin{align}
			V=- \sum_{\mbf{q}} (J_x K^x_\mbf{q} + J_y K^y_\mbf{q})
		\end{align}
is such that $H=H_0+V$.
		
		Our analysis of this model makes use of the perturbative continuous unitary transformation (PCUT), which defines an effective Hamiltonian related to the original Hamiltonian by
		\begin{align}
			H_{\mathrm{eff}}=U^{\dagger}HU
		\end{align}
		such that $H_{\mathrm{eff}}$ commutes with $H_0$, and $U\to I$ as $\frac{J_x}{J_z}$ and $\frac{J_y}{J_z}$ vanish. In the thermodynamic limit, this can be calculated as~\cite{Vidal2008}
		\begin{align}\label{e:pcutH}
			H_{\mathrm{eff}} = E_0+\mu H_0-\sum_{\{\mbf{q}_1,\ldots,\mbf{q}_n\}}C_{\{\mbf{q}_1,\ldots,\mbf{q}_n\}}W_{\mbf{q}_1}\cdots W_{\mbf{q}_n}-\sum_{\{\mbf{q}'_1,\ldots,\mbf{q}'_n\}}D_{\{\mbf{q}'_1,\ldots,\mbf{q}'_n\}}R_{\{\mbf{q}'_1,\ldots,\mbf{q}'_n\}}
		\end{align}
		where $E_0$, $\mu$, $C$ and $D$ are constants that can be calculated perturbatively, and $R$ is a string operator that commutes with $K^z_\mbf{q}$ at each site except $\mbf{q}'_1$ and $\mbf{q}'_n$. The precise form for the constants and $R$ will not be crucial in our analysis; the key point is that each of the terms in \Eref{e:pcutH} is a Pauli operator. The interested reader can consult Ref.~\cite{Vidal2008} for details.
		
		Notice that \Eref{e:pcutH} commutes with $L_x$ and $L_y$, and so is exactly degenerate. For finite size systems, there are also corrections to Eq.~(\ref{e:pcutH}) that include loop operators $L_x$ and $L_y$~\cite{Kells2008, Kells2009a}. However, these appear only at order $N_x$ and $N_y$ respectively in $\frac{J_x}{J_z}$ and $\frac{J_y}{J_z}$. The effect of these corrections is to lift the degeneracy between different eigenstates of these operators. As we will show in \Sref{s:quasidegen}, these effects are exponentially small in the system size, and so we will typically neglect them in our treatment of $H_{\mathrm{eff}}$.
		
		The Hamiltonian $H_{\mathrm{eff}}$ is effective in the sense that observables acting on $H_{\mathrm{eff}}$ are related to those on the original system by the transformation $U$. This can also be computed perturbatively if desired, taking local operators to quasi-local operators. For the most part when considering the honeycomb model in the perturbative regime, we will work in the effective picture for simplicity.
		
		When restricted to the lowest energy eigenspace of $H_0$ (the no-broken-dimer sector) with corresponding projector $\Pi_0$, the effective Hamiltonian simplifies to
		\begin{align}\label{e:loweffectiveham}
			\Pi_0H_{\mathrm{eff}}\Pi_0 = E_0'-\sum_{\{\mbf{q}_1,\ldots,\mbf{q}_n\}}C_{\{\mbf{q}_1,\ldots,\mbf{q}_n\}}\Pi_0W_{\mbf{q}_1}\cdots W_{\mbf{q}_n}\Pi_0 \,.
		\end{align}		
		This can be further simplified by considering only the leading non-trivial order terms. These terms appear at fourth order, and yield the Hamiltonian
		\begin{align}
			\Pi_0H_{\mathrm{eff}}^{(4)}\Pi_0 = E_0'-\sum_{\mbf{q}}C_{\mbf{q}}\Pi_0W_{\mbf{q}}\Pi_0\,.
		\end{align}		
		Within the image of $\Pi_0$ (that is, the common $+1$ eigenspace of all $K^z_{\mbf{q}}$), we can consider a 2-dimensional effective spin at each site $\mbf{q}$ of the lattice with a computational basis given by eigenstates of the $Z_{\mbf{q},\bullet}$ operator. 
Using this identification, the Hamiltonian $\Pi_0H_{\mathrm{eff}}^{(4)}\Pi_0$ gives the toric code, with the $W_{\mbf{q}}$ playing the role of the stabilizer operators as was first noted by Kitaev~\cite{Kitaev2006}.
		
		For this reason, it will often be convenient to characterize excitations of the honeycomb model in the perturbative regime by eigenvalues of the $K^z_{\mbf{q}}$ and $W_{\mbf{q}}$ operators. Although a common eigenstate of these operators is not generally an eigenstate of either $H$ or $H_{\mathrm{eff}}$, the ground space of $H_{\mathrm{eff}}$ can be simply found as the $+1$ eigenspace of all $K^z_\mbf{q}$ and $W_\mbf{q}$ operators. 
It is thus convenient to consider the abelian stabilizer group generated by the $K^z_\mbf{q}$ and $W_\mbf{q}$, and consider $-1$ eigenvalues for any of these operators as localized quasiparticle excitations---broken dimers or vortices respectively. 
In the perturbative limit, vortices can be considered as toric code quasiparticles, while broken dimers are simply localized non-topological quasiparticles. 
		

\section{Coherent noise}\label{s:coherent}

	\subsection{Properties of topological phases}\label{s:asymptotics}
		
Systems within the same topological phase share common properties that are relevant for error correction. 
In the following subsections, we will make use of the particular algebraic structure of the honeycomb model to extract bounds on its behavior, but it is also possible to derive results based on general topological properties of the system. 
This will give some idea of the asymptotic behavior of the model, without giving at all tight bounds. 
In this subsection we will derive several bounds which motivate the more precise study succeeding it.

The strategy in this analysis will be to relate the properties of the honeycomb model to the known properties of the standard toric code, simply through their being in the same topological phase. 
A topological phase is defined as an equivalence class of (quasi-)local Hamiltonians connected by smooth, uniformly gapped paths~\cite{Chen2010}. 
In order to make this concrete, let us explicitly specify a path of Hamiltonians relating the honeycomb code space to the toric code ground space:
\begin{align}\label{e:hampath}
	H(s) &=-\sum_{\mbf{q}}\bigl(J_wW_{\mbf{q}}+sJ_xK^x_{\mbf{q}}+sJ_yK^y_{\mbf{q}}+J_zK^z_{\mbf{q}}\bigr)\,.
\end{align}
At $H(0)=-\sum_{\mbf{q}}\left(J_wW_{\mbf{q}}+J_zK^z_{\mbf{q}}\right)$ it is clear that we have an encoded version of the toric code. For our choice of boundary conditions and encoding, this has distance $d=N$~\cite{Bombin2007}. That is, it is a $[[2N^2,2,N]]$ code. It satisfies exact local indistinguishability and it has an exactly degenerate code space~\cite{Kitaev2003,Kitaev2006}. At $H(1)=-\sum_{\mbf{q}}\left(J_wW_{\mbf{q}}+J_xK^x_{\mbf{q}}+J_yK^y_{\mbf{q}}+J_zK^z_{\mbf{q}}\right)$, we have the honeycomb model Hamiltonian plus some additional $J_wW_{\mbf{q}}$ terms. Since these terms commute with the honeycomb Hamiltonian, and it is known that the ground space of the model is in the common $+1$ eigenspace of all $W_{\mbf{q}}$, it is easy to see that $H(1)$ has the same code space as the honeycomb Hamiltonian for the same $J_x, J_y, J_z$.

Within the vortex-free sector, the honeycomb model is known to have a gap between the code space and the next lowest eigenvalue given by $\delta E = 2(J_z-J_x-J_y)$ in the thermodynamic limit~\cite{Kitaev2006}. Since $H(s)$ has energy at least $4J_w$ for any state orthogonal to the vortex-free sector, by choosing $J_w = \frac{1}{4}\delta E$, we can guarantee that $H(s)$ also has a uniform gap above the code space of at least $\delta E$.

The key technical tool underlying the kinds of general topological arguments we will discuss here is the quasi-adiabatic continuation mapping between the code spaces of $H(0)$ and $H(1)$. There are several formulations of the quasi-adiabatic continuation~\cite{Hastings2004, Hastings2005, Osborne2007, Nachtergaele2007} that would yield slightly different bounds, but for concreteness we follow Ref.~\cite{Bachmann2012}. In order to be explicit about the locality properties of the quasi-adiabatic continuation, we note that all the Hamiltonians $H(s)$ satisfy a Lieb-Robinson bound~\cite{Lieb1972}, a version of which we state here for completeness.

\begin{lemma}[Lieb-Robinson bound~\cite{Lieb1972}]
Under time evolution for time $t$, there exist constants $C,\alpha,v>0$ such that observables $O_A$ and $O_B$ whose support lies on disjoint regions $A, B\subset \Lambda$ satisfy
\begin{align}
	\|[\e^{iH(s)t}O_A\e^{-iH(s)t},O_B]\|\leq C\|O_A\|\|O_B\|\min(|A|,|B|)\e^{-\alpha\left(d(A,B)-v|t|\right)}
\end{align}
for all $s$, where $d(A,B)$ is the distance between regions $A$ and $B$, and $\|\cdot\|$ is the operator norm.
\end{lemma}

The quasi-adiabatic continuation $U_{\mathrm{qac}}(s)$ maps the code space of $H(0)$ to that of $H(s)$. It is shown in Ref.~\cite{Bachmann2012} that a Lieb-Robinson bound implies that the effect of the quasi-adiabatic continuation on an operator whose support is contained within a region $A$ can be approximated well by an operator $V_{A^r}(s)$ localized within an $r$-ball around $A$ (we will suppress $s$ labels for clarity).

\begin{lemma}[Locality of quasi-adiabatic continuation~\cite{Hastings2004}, following~{\cite[Thm.~3.4]{Bachmann2012}}]\label{l:localityqac}
The action of the quasi-adiabatic continuation on an operator $O_A$ supported on $A$ can be approximated by $V_{A^r}$ to an accuracy
\begin{align}\label{e:approxqac}
	\|U_{\mathrm{qac}}O_AU_{\mathrm{qac}}^{\dagger}-V_{A^r}O_AV_{A^r}^{\dagger}\|\leq 2C'\|O_A\|\left(\kappa r\right)^{10}\e^{-\frac{2}{7}\frac{\kappa r}{\ln^2 \kappa r}}
\end{align}
for some constants $\kappa, C'$ (depending on coupling constants of the Hamiltonian), and for $\kappa r> c$ for some absolute constant $c$.
\end{lemma}

Given these lemmas, the next theorem provides a convenient way to summarize their implications for error correction. Following Ref.~\cite{Flammia2016}, define an $[[n,k,d,\delta,r]]$ code to consist of $k$ logical qubits encoded in $n$ physical qubits, where every region of size less than $d$ is $(\delta,r)$-correctable. A region $A$ is $(\delta,r)$-correctable if an erasure of qubits within $A$ can be corrected to error $\delta$ in the Bures distance, by acting on an $r$-ball enclosing $A$. 

\begin{theorem}[Smoothness of code parameters~{\cite[Thm.~16]{Flammia2016}}]\label{t:smoothcode}
Given a code space $\mathcal{C}_0$ with parameters $[[n,k,d]]$ for $d=\Omega(n^\gamma)$ for some $\gamma>0$, and a unitary $U$ whose action on local regions can be approximated by a local operator as in \Eref{e:approxqac}, $\mathcal{C}_1\equiv U\mathcal{C}_0U^{\dagger}$ is an $[[n,k,d-2r,\delta(r),r]]$ approximate code, with $\delta(r)\leq 2C'\left(\kappa r\right)^{10}\e^{-\frac{2}{7}\frac{\kappa r}{\ln^2 \kappa r}}$.
\end{theorem}

According to \Thref{t:smoothcode}, we immediately see that the honeycomb model is a $[[2N^2,2,N-2r, \delta(r), r]]$ approximate error correcting code, for any $r>\frac{c}{\kappa}$. We expect that this bound could be modified in several ways by using alternative formulations of the quasi-adiabatic continuation. In particular, a similar bound should hold for $r\sim 1$.

An intuitive corollary of this theorem is that the honeycomb model is approximately locally indistinguishable. The standard toric code (like other topologically ordered models at the RG fixed point) is well-known to satisfy local indistinguishability in the sense that all codestates look identical when restricted to any local region. In contrast, the states of the honeycomb code space will not exactly satisfy this condition. Thus instead of considering exact local indistinguishability, we consider its approximate version, which we define as follows:

\begin{definition}[Approximate local indistinguishability]\label{D:Apprxlocalindist}
Consider regions $A$ with linear size less than $N^*$, and let $O_\sub{A}$ be any Hermitian operator whose support is contained in $A$, i.e.~the linear size of $\mbox{supp}(O_\sub{A})$ is smaller than $N^*$. A subspace $\mathcal{C}$ is $\Delta$-approximately indistinguishable on length scale $N^*$ if, for all $A$ and all $O_\sub{A}$,
\begin{align}\label{E:Apprxlocalindist}
\sup_{\psi_1,\psi_2\in \mathcal{C}}\abs{\expect{O_\sub{A}}_{\psi_1}-\expect{O_\sub{A}}_{\psi_2}}\leq \Delta\cdot\|O_\sub{A}\|	\,,
\end{align}
A $0$-approximately indistinguishable space is said to be exactly locally indistinguishable.

A family of subspaces $\{\mathcal{C}_N\}_N$ of models of lattice size $N$ is said to be approximately locally indistinguishable if there exists $\beta>0$, $N_0<\infty$, and superpolynomially decaying function $\Delta$ such that for all $N>N_0$, $\mathcal{C}_N$ is $\Delta(N)$-approximately indistinguishable on length scale $N^*\propto N^\beta$. An approximately locally indistinguishable family of states with $\Delta(N)=0$ is said to be exactly locally indistinguishable.
\end{definition}

In discussing local indistinguishability, we will typically neglect the distinction between the code space of the honeycomb model on a fixed lattice, and the family of code spaces of the honeycomb model for all lattice sizes. Thus we will abuse the language by saying that the honeycomb code space is approximately locally indistinguishable, it being implicit that the more precise statement is that the family of code spaces of the honeycomb model for all lattice sizes is approximately locally indistinguishable.

Our definition of local indistinguishability differs slightly from related definitions in the literature, since it is typical to either consider only exact local indistinguishability~\cite{Bravyi2010,Bravyi2011stability}, or to consider a more restrictive definition of local topological quantum order that includes both local indistinguishability and local consistency~\cite{Michalakis2013,Cirac2013}.

In order to show that the honeycomb code space is approximately locally indistinguishable, it will be convenient to recall Thm.~7 of Ref.~\cite{Flammia2016}, which states that
\begin{theorem}[A correctable region avoids logical operators~{\cite[Thm.~7]{Flammia2016}}]\label{T:correctregionavoidlogop}
If region $A$ is $(\delta,r)$-correctable, then for any logical unitary $U_L$, there exists an operator $V_{\bar{A}}$ supported only on the complement of $A$, such that
\begin{align}
\|(U_L-V_{\bar{A}})\Pi\|&\leq 4\sqrt{\delta}\\
\|\Pi(U_L-V_{\bar{A}})\|&\leq 4\sqrt{\delta}
\end{align}
for $\Pi$ the projector to the code space.
\end{theorem}

This implies that such a $(\delta,r)$-correctable region $A$ is also $8\sqrt{\delta}$ indistinguishable, since there exists logical unitary $U_L$ that maps between any given pair of ground states $U_L |\psi_1\rangle = |\psi_2\rangle$ such that
\begin{align}
\left|\expect{O_A}_{\psi_1}-\expect{O_A}_{\psi_2}\right|&= \bigl|\bigl\langle\psi_2\bigl|  \bigl[U_L,O_A\bigr] \bigr|\psi_1\bigr\rangle\bigr|\\
&\leq \bigl\| \Pi \bigl[U_L,O_A\bigr]\, \Pi \bigr\|\\
&\leq 8\sqrt{\delta}\bigl\|O_A\bigr\| \,, \label{e:correctindistinguish}
\end{align}
where the last line follows from using the above theorem, the triangle inequality, and $[O_A,V_{\bar{A}}]=0$.

For the honeycomb model, regions $A$ of linear size $N^*$ contain at most $2(N^*)^2$ qubits. Thus, for $r<\frac{N}{2}-(N^*)^2$, all such regions satisfy $|A|< d = N-2r$ and are therefore $(\delta(r),r)$-correctable, with $\delta(r)$ as in \Thref{t:smoothcode}. By setting $N^*=\frac{1}{2}N^{\frac{1}{2}}$ and $r=\frac{N}{2}-(N^*)^2-1$, we see that $8\sqrt{\delta(r)}$ is superpolynomially decaying in $N$. 
Note that this only holds when $r$ is large enough that \Lref{l:localityqac} applies, but this will be true for all $N>N_0=\frac{4c}{\kappa}+4$. 
Then \Eref{e:correctindistinguish} immediately implies that the honeycomb code space is approximately locally indistinguishable.

Alternatively, an analogous result could have been derived directly by considering local operators of the toric code, for which exact local indistinguishability holds, evolved under the truncated quasi-adiabatic continuation $V_{A^r}(s)$, and using the bound \Eref{e:approxqac}.

A further implication is that the splitting of the code space of the honeycomb model is superpolynomially suppressed in system size. This can be seen by considering the expectation value of the Hamiltonian in an arbitrary ground state, noting that any local operator can only distinguish the code states by a superpolynomially small amount in system size, and the Hamiltonian is composed of a polynomial number of local operators. This gives that the energies of different code states can only differ by an superpolynomially small amount.

These general properties of approximate degeneracy, approximate local indistinguishability, and approximate code distance motivate the study of the honeycomb model for quantum information storage. By making use of specific properties of the honeycomb model, as opposed to general properties of the toric code phase, in the following sections we will be able to rigorously make much more precise statements about the behavior of this system.

	\subsection{Approximate degeneracy}\label{s:quasidegen}
	
In this section, we consider more carefully the approximate degeneracy of the honeycomb Hamiltonian. In the context of error correction, this corresponds to the resilience of information stored in the honeycomb code space against dephasing errors. An intrinsic dephasing effect exists in finite size systems even without the presence of other external noise, due to the inexact degeneracy of the code space. 

This  approximate degeneracy has been previously studied in Ref.~\cite{Kells2008, Kells2009a}, though their perturbative analysis applies only in the limit that $J_z\gg J_x, J_y$. In the previous section, we were only able to argue on general grounds that the splitting asymptotically decays superpolynomially. 
In contrast, here we rigorously and precisely quantify the splitting of ground state energy within the whole abelian phase. 

Recall that without loss of generality in the abelian phase, we take $J_z > J_x \geq J_y > 0$ and $J_z > J_x + J_y$ as in Eq.~(\ref{eq:couplingconstants}). 
Here and elsewhere it will be convenient to introduce a constant 
\begin{align}\label{eq:adef}
	a = \ln\biggl(\frac{J_z-J_y}{J_x}\biggr)
\end{align}
which is strictly positive everywhere in the gapped phase. 
The main result of this section is then the following:
\begin{theorem}(Energy splitting bound)\label{Th:qsdegeneracy}
	Let $\Delta E = \abs{E_{\ell} - E_{\ell'}}$ be the energy splitting between any pair of (pseudo-ground) states in the four-dimensional code space of the honeycomb model with coupling constants as in Eq.~(\ref{eq:couplingconstants}). Then 
	\begin{align}
		\frac{\Delta E}{J_z} \leq \frac{16 \sqrt{2} N^2}{\mathrm{e}^{aN}-1} = \mathcal{O}\bigl( N^2 \mathrm{e}^{-aN} \bigr)\,.
	\end{align}
	where $a > 0$ is given by $a = \ln\bigl(\tfrac{J_z-J_y}{J_x}\bigr)$.
\end{theorem}
This result shows that, in a superposition of code states, dephasing due to energy splitting decreases exponentially in the linear system size within the entire gapped phase. Moreover, this tendency vanishes exactly where the gap closes on the phase boundary $J_z = J_x + J_y$, where $a \rightarrow 0$.

\begin{proof}
Our strategy will be to look at the thermodynamic limit of the model and show that finite-size approximations to this limit converge to the asymptotic result quickly in the linear size of the system. Rather than bound the finite-size corrections perturbatively, we will use arguments from approximation theory that give exponential accuracy. 

The key ingredient of our proof is a bound on the error of the rectangular rule from numerical integration. Consider a function $f:\R\rightarrow\R$, and define the error for the $N$-point rectangular rule by
\begin{align}
	\epsilon_0(N)
	\equiv \frac{1}{2\pi} \int^{2\pi}_0 f(x) \mathrm{d}x - \frac{1}{N} \sum_{n=0}^{N-1} f\Bigl(\frac{2\pi n}{N}\Bigr).
\end{align}

Suppose that $f$ is also analytic and $2\pi$-periodic. Then there exists a strip in the complex plane $D_0=\R\times(-a_0,a_0)\in\C$ with $a_0>0$ such that $f$ can be extended to a holomorphic and $2\pi$-periodic bounded function $f:D_0\rightarrow\C$. We will choose the largest possible strip, so that $a_0$ is defined as the supremum of the half-width over of all such strips for $f$. Let $Q_0$ be an upper bound for $\abs{f}$ on $D_0$; in particular we choose $Q_0=\sup_{D_0}\abs{f}$. Then by a theorem of Davis~\cite{Davis1959} the error for the rectangular rule can be bounded by
\begin{align}
	\abs{\epsilon_0(N)}\leq \frac{2Q_0}{\mathrm{e}^{a_0 N}-1} \,.\label{E:erbnd1D}
\end{align} 
See \cite[Thm.~9.28]{Kress1998} for a simple proof.

For our purposes, the function $f$ of interest will be the energy density of any of the four honeycomb code basis states. In order to use Davis' theorem, we must first extend it to apply to a two-dimensional function $f(x,y)$. We define the error of two-dimensional integration
	\begin{align}
		\epsilon(N) 
    	\equiv 
	    \frac{1}{4\pi^2} \int^{2\pi}_0 \int^{2\pi}_0 f(x,y) \mathrm{d}x\mathrm{d}y
		- \frac{1}{N^2} \sum_{n_x=0}^{N-1} \sum_{n_y=0}^{N-1} 
		f\Bigl(
			\frac{2\pi n_x}{N},\frac{2\pi n_y}{N}
		\Bigr), 
	\end{align}
	which can be expressed in terms of the errors of one-dimensional integrations in $x$ and $y$:
	\begin{align}
		\epsilon_1(N;y)=\frac{1}{2\pi}\int^{2\pi}_0f(x,y)\mathrm{d}x - \frac{1}{N}\sum_{n_x=0}^{N-1}f\Bigl(\frac{2\pi n_x}{N},y\Bigr), \\
		\delta_1(N;x)=\frac{1}{2\pi}\int^{2\pi}_0f(x,y)\mathrm{d}y - \frac{1}{N}\sum_{n_y=0}^{N-1}f\Bigl(x, \frac{2\pi n_y}{N}\Bigr).
	\end{align}
	If for all real values $y$, $f(x, y)$ is analytic and $2\pi$-periodic in $x$, then for every value of $y$ there exists a strip of width $a(y) > 0$, $D_a = \R\times\bigl(-a(y),a(y)\bigr)\in\C$ and an upper bound $Q(y) = \sup_{x \in D_a} |f(x,y)|$ such that 
	\begin{align}
		\abs{\epsilon_1(N;y)} &\leq \frac{2Q(y)}{\mathrm{e}^{a(y) N}-1}.
	\end{align}
	We also introduce analogous constants $b(x)$ on a strip $D_b$ and $R(x) = \sup_{y \in D_b} |f(x,y)|$ for the case of the $y$ variable, and we obtain
	\begin{align}
			\abs{\delta_1(N;x)} &\leq \frac{2R(x)}{\mathrm{e}^{b(x) N}-1}.
	\end{align}
	According to the definition of $\epsilon_1(N;y)$ and $\delta_1(N;x)$ we know that
\begin{align*}
	\frac{1}{N^2} \sum_{n_x, n_y=0}^{N-1} f\biggl( \frac{2\pi n_x}{N},\frac{2\pi n_y}{N} \biggr) 
			= & \frac{1}{N} \sum_{n_y=0}^{N-1} \left[ \frac{1}{2\pi}\int_{0}^{2\pi} f\biggl( x,\frac{2\pi n_y}{N} \biggr) \mathrm{d}x - \epsilon_1 \biggl( N;\frac{2\pi n_y}{N} \biggr) \right] \\
			= &\frac{1}{2\pi}\int_{0}^{2\pi} \left[ \frac{1}{2\pi} \int_{0}^{2\pi} f(x,y) \mathrm{d}y - \delta_1 (N;x) \right] \mathrm{d}x - \frac{1}{N} \sum_{n_y=0}^{N-1} \epsilon_1 \biggl( N;\frac{2\pi n_y}{N} \biggr) \\
			= &\frac{1}{4\pi^2}\int_{0}^{2\pi} \int_{0}^{2\pi} f(x,y) \mathrm{d}x\mathrm{d}y - \frac{1}{2\pi}\int_{0}^{2\pi} \delta_1 (N;x) \mathrm{d}x - \frac{1}{N} \sum_{n_y=0}^{N-1} \epsilon_1 \biggl( N;\frac{2\pi n_y}{N} \biggr). 
\end{align*}
Using the triangle inequality then gives us an error bound for the two-dimensional rectangular rule, 
\begin{align}
	\abs{\epsilon(N)} &\leq \abs{\frac{1}{2\pi}\int_{0}^{2\pi} \delta_1 (N;x) \mathrm{d}x} + \abs{\frac{1}{N} \sum_{n_y=0}^{N-1} \epsilon_1 \biggl( N;\frac{2\pi n_y}{N} \biggr)}	\notag\\
	&\leq \max_{x} \abs{\delta_1 (N;x)}  +  \max_{y} \abs{\epsilon_1 (N;y)}.
\end{align}
	Here the optimization is over real values of the periodic variables $x$ or $y$, so  a maximum (as opposed to supremum) is appropriate due to compactness. Using the bound from the one-dimensional case, we obtain 
	\begin{align}
		\abs{\epsilon(N)} &\leq \max_{y} \frac{2Q(y)}{\mathrm{e}^{a(y)N}-1} +  \max_{x} \frac{2R(x)}{\mathrm{e}^{b(x)N}-1} \notag \\
		& \le \frac{2Q}{\mathrm{e}^{a N}-1}  + \frac{2R}{\mathrm{e}^{bN}-1}\,,
	\end{align}
where in the last line we have defined
	\begin{align}
		Q = \min_{y} Q(y)\,, \quad a = \min_y a(y)\,, \quad R = \min_{x} R(x) \,, \quad b = \min_x b(x) \,.
	\end{align}
The ``$\min$'' in the definitions of $Q$ and $R$ is not a typo, and the stronger inequality using a minimum instead of a naive maximum follows from the results of Lemma~\ref{L:expConstants} below.

	Now let us consider the energy density per dimer of the system in the code basis state $\ell$ of \Eref{Eq:gsEnergy}, which can be considered as an approximate evaluation of \Eref{Eq:gsEnDensThermoLimit} as
	\begin{align}
		\frac{1}{4\pi^2} \int^{2\pi}_0 \int^{2\pi}_0 \left( -\frac{1}{2}E(x, y) \right) \mathrm{d}x \mathrm{d}y = \frac{E_\ell}{N^2} + \epsilon_\ell(N) \,,
	\end{align}
	where $\epsilon_\ell(N)$ is the error term for approximating the two-dimensional integral by a sum for ground state $\ell$. Note that the continuum integral is independent of $\ell$. Then from the triangle inequality we have
	\begin{align}
		\frac{\Delta E}{N^2} \leq \abs{\epsilon_\ell(N)} 
	 	 + \abs{\epsilon_{\ell'}(N)}\,.
	\end{align}

Recall that the mode energy $E_\mbf{k}$ in Eq. (\ref{Eq:gsModeEn}) is composed of trigonometric functions, is $2\pi$-periodic and is analytic in the gapped phase, so the mode energy satisfies the conditions of Davis' theorem for both $\ell$ and $\ell'$.
Since $\epsilon_{\ell}(N)$ and $\epsilon_{\ell'}(N)$ are the errors for the same integral, they share the same bound, and
	\begin{align}
		\frac{\Delta E}{N^2} &\leq 2\left( \frac{2Q}{\mathrm{e}^{a N}-1}  + \frac{2R}{\mathrm{e}^{bN}-1}\right) \\
		&\leq 4\left(  \frac{Q+R}{\mathrm{e}^{\min\{a,b\}N}-1}  \right),
	\end{align}
which follows immediately from the two-dimensional error bound we derived above.

For the coupling constants as in \Eref{eq:couplingconstants}, we will prove in \Lref{L:expConstants} below that $Q + R \le 4 \sqrt{2} J_z$ and $b \ge a \ge \ln\bigl(\frac{J_z - J_y}{J_x}\bigr)$. The theorem then follows immediately. 
\end{proof}

The following lemma provides the explicit values of the constants used in the proof of \Thref{Th:qsdegeneracy}.

\begin{lemma}(Bounds on constants $a, b, Q, R$) \label{L:expConstants}
For the honeycomb model with coupling constants as in \Eref{eq:couplingconstants}, for all $y$ let $a(y)$ be the supremum half-width over all strips $D_a = \R\times(-a(y),a(y))\in\C$ containing an analytic extension in $x$ of $E(x,y)$, let $a = \min_y a(y)$, and let $Q = \min_{y} Q(y) = \min_y \sup_{D_a} \abs{E(x,y)}$. Define $b$ and $R$ in analogy with $a$ and $Q$, but with $x \leftrightarrow y$. Then we have the following inequalities
\begin{align}
	\max_{y} \frac{Q(y)}{\mathrm{e}^{a(y)N}-1} \le \frac{Q}{\mathrm{e}^{a N}-1} \quad \text{and} \quad \max_{x} \frac{R(x)}{\mathrm{e}^{b(x)N}-1} \le \frac{R}{\mathrm{e}^{bN}-1} \,,
\end{align}
as well as
\begin{align}
	Q + R \le 4 \sqrt{2} J_z \quad \text{and} \quad b \ge a \ge \ln\biggl(\frac{J_z - J_y}{J_x}\biggr)\,.
\end{align}

\end{lemma}
\begin{proof}
	To compute the explicit bounds, we have to extend the variables $x$ and $y$ into the complex plane separately. We will focus on the case that $x=x_\sub{R}+ix_\sub{I}$ and $y\in \mathbb{R}$, and since $x$ and $y$ are symmetric in the function we consider, the bounds in two directions have the same form with $x$ and $y$ exchanged. At all times we restrict our discussion to the gapped phase with coupling constants as in \Eref{eq:couplingconstants}.
	
	The exponent $a$ in the error bound is determined by the width of the strip where the energy function $E=E(x,y)$ is analytic in the complex plane, so the boundary of the region in $\C$ where $E$ is analytic will give us the width of this strip. Since the square root of a complex number is analytic except at its branch cut, the non-analytic region of $E$ satisfies
	\begin{align}
		\mbox{Re}\bigl\{ E^2 \bigr\} \leq 0\qquad\qquad\mbox{and}\qquad\qquad \mbox{Im}\bigl\{ E^2 \bigr\} = 0 \,.   \label{Eq:ReImInquality}
	\end{align}
	By doing the complex extension explicitly with the energy function from \Eref{Eq:gsModeEn}, we have the following expressions
	\begin{align}
		\mbox{Re}\biggl\{ \frac{E^2}{4} \biggr\} &= J_x^2 + C^2 + 2J_x\cosh(x_\sub{I}) \bigl[ B\cos(x_\sub{R}) + A\sin(x_\sub{R}) \bigr],\\
		\mbox{Im}\biggl\{ \frac{E^2}{4} \biggr\} &= 2J_x\sinh(x_\sub{I}) \bigl[ A\cos(x_\sub{R}) - B\sin(x_\sub{R}) \bigr],
	\end{align}
	where $A = J_y\sin(y)$, $B = J_y\cos(y) +　J_z$ and $C = \sqrt{A^2 +　B^2}$\,.
	This expression simplifies if we define unit vectors
	\begin{align}
		&\hat{z}       = \bigl( \cos(x_\sub{R}), \sin(x_\sub{R}) \bigr)\,,\\
		&\hat{z}^\perp = \bigl( -\sin(x_\sub{R}), \cos(x_\sub{R}) \bigr)\,,\\
		&\hat{c}       = \biggl( \frac{B}{C}, \frac{A}{C} \biggr)\,,
	\end{align}
	such that 
	\begin{align}
		\mbox{Re}\biggl\{ \frac{E^2}{4} \biggr\} &= J_x^2 + C^2 + 2J_x C \cosh(x_\sub{I}) \bigl( \hat{c}\cdot\hat{z} \bigr)\,,\\
		\mbox{Im}\biggl\{ \frac{E^2}{4} \biggr\} &= 2J_x C \sinh(x_\sub{I}) \bigl( \hat{c}\cdot\hat{z}^\perp \bigr)\,.
	\end{align}
	From \Eref{Eq:ReImInquality} we find that for non-analytic $E$, either $\sinh(x_\sub{I}) = 0$ or $\hat{c}\cdot\hat{z}^\perp=0$, since $C>0$ and $J_x > 0$. The first case gives $x_\sub{I} = 0$ which is trivial, and the second reveals that $\hat{c}\cdot\hat{z} = \pm 1$.  Hence \Eref{Eq:ReImInquality} gives
	\begin{align}
	\begin{array}{ll}
		\cosh(x_\sub{I}) \leq -\dfrac{J_x^2 + C^2}{2 J_x C}, &\text{ for } \hat{c}\cdot\hat{z}=+1	\,,\\[10pt]
		\cosh(x_\sub{I}) \geq  \dfrac{J_x^2 + C^2}{2 J_x C}, &\text{ for } \hat{c}\cdot\hat{z}=-1	\,.
	\end{array}
	\end{align}
	There is no solution for $x_\sub{I}$ in the first case since $-\frac{J_x^2 + C^2}{2 J_x C} < 0$. In the second case, we have
	\begin{align}
		x_\sub{I} \geq \cosh^{-1}\biggl( \frac{J_x^2 + C^2}{2 J_x C} \biggr)\,, 
	\end{align}
	and this inequality reveals the width of strip, $a(y)$. We can make one step further since $\cosh^{-1}(u) = \ln(u + \sqrt{u^2-1})$ for $u \geq 1$. The condition is always satisfied because $J_x, C > 0$ in the gapped phase. hence the width of the strip in the gapped phase (where $C^2 > J_x^2$) becomes
	\begin{align}
		a(y) = \ln \Biggl( \frac{C}{J_x} \Biggr) = \ln \Biggl( \frac{\sqrt{J_z^2 + J_y^2 + 2J_y J_z \cos(y)}}{J_x} \Biggr) \,,
	\end{align}
	and minimizing over the real variable $y$ we find
	\begin{align}
	a = \min_y a(y) = \ln \biggl( \frac{J_z - J_y}{J_x} \biggr)\,.
	\end{align}
The value for $b$ follows the same argument, and yields $b = \ln \bigl( \frac{J_z - J_x}{J_y} \bigr)$. The inequality $b \ge a$ follows from elementary algebra using the conditions in \Eref{eq:couplingconstants}. 

To compute $Q$, we apply the maximal modulus principle from complex analysis that states that the maximum modulus of a function $E$ which is analytic in an open subset $D$ of the complex plane lies on the boundary $\partial D$ of $\bar{D}$, the closure of $D$. To make the calculation simpler, we take the fourth power of the energy function to eliminate the square root, such that
\begin{align}
	\frac{\abs{E}^4}{16} =  \Bigl[ J_x^2 + C^2 + 2J_x C \cosh(x_\sub{I}) \bigl( \hat{c}\cdot\hat{z} \bigr)\Bigr]^2 + \Bigl[ 2J_x C \sinh(x_\sub{I}) \bigl( \hat{c}\cdot\hat{z}^\perp \bigr)\Bigr]^2.
\end{align}
	
The boundary of the strip is along the line $x_\sub{I} = \cosh^{-1}\biggl( \frac{J_x^2 + C^2}{2 J_x C} \biggr)$, and on the boundary the energy function takes the values
\begin{align}
	\frac{\abs{E}^4}{16} =  \Bigl[ J_x^2 + C^2 + \bigl(J_x^2 + C^2 \bigr) \bigl(\hat{c}\cdot\hat{z} \bigr)\Bigr]^2 + \Bigl[ \bigl(C^2 - J_x^2 \bigr) \bigl( \hat{c}\cdot\hat{z}^\perp \bigr)\Bigr]^2.
\end{align}
Since only $\hat{z}$ and $\hat{z}^\perp$ depend on $x_\sub{R}$, and they only appear in the term inside an inner product with $\hat{c}$, we can replace the $x_\sub{R}$-dependence by an angle $\theta$ such that $\hat{c}\cdot\hat{z} = \cos(\theta)$ and $\hat{c}\cdot\hat{z}^\perp = \pm\sin(\theta)$. Now we arrive at a simple form for the energy function on the boundary,
\begin{align}
	\frac{\abs{E}^4}{16} = \bigl(J_x^2 + C^2\bigr)^2 \bigl( 1+\cos(\theta) \bigr)^2 + \bigl(C^2 - J_x^2\bigr)^2\sin^2(\theta) \,,
\end{align}

The extreme values of this function are at points where $\tfrac{\partial}{\partial\theta}\abs{E}^4 = 0$, 
which is equivalent to
\begin{align}
	\left\{\begin{array}{l}
			\sin(\theta) = 0\\
			\cos(\theta) = - \frac{\bigl(C^2 + J_x^2\bigr)^2}{4C^2J_x^2}
		\end{array}\right. \,.
\end{align}
There is no solution for $\theta$ in the $\cos(\theta)$ equation, since for a gapped system (\Eref{eq:couplingconstants}) we have $\bigl(C^2 + J_x^2\bigr)^2 - 4C^2J_x^2 > 0$. Therefore the maximum of the energy function is given by
\begin{align}
	\sup_{x_\sub{R}} \abs{E}^4 = 64\bigl[J_x^2 + J_y^2 + J_z^2 + 2J_yJ_z\cos(y)\bigr]^2
\end{align}
and
\begin{align}
	Q(y) = \sup_{D_a} \abs{E} &= 2\sqrt{2}\bigl[J_x^2 + J_y^2 + J_z^2 + 2J_yJ_z\cos(y)\bigr]^{1/2}\,.
\end{align}
Minimizing this expression over the real variable $y$ yields
\begin{align}
	Q = \min_y Q(y) = 2\sqrt{2}\bigl[J_x^2 + (J_z - J_y)^2 \bigr]^{1/2}\,.
\end{align}
The derivation for the constant $R$ is the same, except exchanging $x$ and $y$. 

Having derived the expressions for $Q(y)$ and $a(y)$ (as well as their minimums), we now turn to the claimed upper bound 
\begin{align}\label{E:trickyinequality}
	\max_{y} \frac{Q(y)}{\mathrm{e}^{a(y)N}-1} \le \frac{Q}{\mathrm{e}^{a N}-1} \,.
\end{align}
First let
\begin{align}
	u(y) = \mathrm{e}^{a(y)} = \sqrt{J_y^2 + J_z^2 + 2J_yJ_z\cos(y)}/J_x \quad \text{and} \quad g(u) =  \frac{\sqrt{1+u^2}}{{u}^{N}-1}\,.
\end{align}
Then the maximization on the left hand side can be rewritten in this new variable as simply
\begin{align}
	 2\sqrt{2} J_x \max_{y} g\bigl(u(y)\bigr) \,.
\end{align}
The stationary points of $g\bigl(u(y)\bigr)$ are where $\frac{\mathrm{d} g}{\mathrm{d}u} \frac{\mathrm{d}u}{\mathrm{d}y} = 0$, and $g(u)$ is a strictly decreasing function of $u$ for $u >1$, so there are no stationary points where $\frac{\mathrm{d} g}{\mathrm{d}u}$ vanishes. Therefore the maximum must occur where $\frac{\mathrm{d}u}{\mathrm{d}y} = 0$, which simplifies to just $\sin(y) = 0$. Again because $g(u)$ is strictly decreasing the maximum occurs when $u$ is a small as possible, and it follows that $y=\pi$ is always the maximizer, justifying our choice to minimize in the definitions of $Q$ and $a$. This establishes the bound in Eq.~(\ref{E:trickyinequality}). The analogous inequality involving $R$ and $b$ follows the same argument. 

Finally, the upper bound on $Q+R$ then comes from maximizing the coupling constants $J_x, J_y$ in this sum over the region specified by Eq.~(\ref{eq:couplingconstants}), with the maximum occurring in the limit that $J_x, J_y \to 0$.
\end{proof}

	\subsection{Approximate local indistinguishability}\label{s:approxlocalind}
	
\label{s:approxlocalindistinguish}

In addition to the intrinsic dephasing generated by finite size effects, another source of coherent noise for quantum information stored in the honeycomb model will be from local perturbations or local operations, representing corrections to the model or experimental imperfections. To analyze the stability of this system against these sources of error, we study the local indistinguishability of the honeycomb code space~\cite{Kitaev2003, Bravyi2010, Bravyi2011stability, Michalakis2013}.

As argued in \Sref{s:asymptotics}, the honeycomb code space will not be exactly locally indistinguishable, but will satisfy the approximate local indistinguishability condition given in \Dref{D:Apprxlocalindist}. Although this has already been shown, the decay functions and constants in the argument are far from optimal. In this section, we will reprove the approximate local indistinguishability of the honeycomb model in a more precise form. The explicit theorem will be presented as \Thref{T:arbitrary_operator}, but before proving this we will introduce several relevant lemmas.

As in \Sref{s:quasidegen}, the key technical ingredient of the proof is a bound on the discrepancy between quantities on finite size lattices and quantities in the thermodynamic limit. Our argument begins by classifying Pauli operators into two sets, the centralizer of the set of $W_\mbf{q}$ and its complement. The complement is shown to satisfy exact local indistinguishability in \Lref{L:LocalIndAntiCommWq}, which leads us to focus on the set of local Pauli operators that commute with all $W_\mbf{q}$. This set is proved to be the group of link operators (those operators generated by link generators) in \Lref{L:linkGenSet}. Analogously to the proof of Theorem \ref{Th:qsdegeneracy}, we can express the expectation value of these link operators as discrete Fourier transforms, and approximate them by continuous Fourier transforms for large $N$. We bound the error of this approximation for string-like link operators in \Lref{L:single_string_case}, and general link operators in \Lref{L:many_string_case}. Finally in \Thref{T:arbitrary_operator} we express an arbitrary local operator in the Pauli basis to demonstrate its approximate local indistinguishability.

For the remainder of this section, we will consider only regions $A$ formed by a simply-connected set of honeycomb lattice plaquettes. We define the corresponding set of plaquette operators contained in $A$ as
\begin{align}	
	\mathcal{W}(A) = \{W_\mbf{q}\, :\, \mbox{supp}(W_\mbf{q})\subset A \}	\,.\label{E:WsetDef}
\end{align}
It is clear that the honeycomb codespace $\mathcal{C}_{\mathrm{khm}}$ possesses the following property:

\begin{lemma}[]\label{L:LocalIndAntiCommWq}
	For any Pauli $P$ with $P\notin \mathcal{Z}(\mathcal{W}(\Lambda))$ (where $\mathcal{Z}(\mathcal{S})$ denotes the centralizer of set $\mathcal{S}$, and recall that $\Lambda$ is the entire lattice), and all $\ket{\psi} \in \mathcal{C}_{\mathrm{khm}}$, 
	\begin{align}
		\expect{P}_{\psi} = 0\,.
	\end{align}
\end{lemma}
\begin{proof}
	It is known that $\mathcal{C}_{\mathrm{khm}}$ is in the vortex-free sector~\cite{Lieb1994}, and so any $P \notin \mathcal{Z}(\mathcal{W}(\Lambda))$ will map $\mathcal{C}_{\mathrm{khm}}$ into the orthogonal complement of $\mathcal{C}_{\mathrm{khm}}$. Thus $\expect{P}_\psi = 0$ for any $\ket{\psi} \in \mathcal{C}_{\mathrm{khm}}$.
\end{proof}

Given this observation, in order to prove indistinguishability we can restrict to considering Paulis $P\in \mathcal{Z}(\mathcal{W}(\Lambda))$ (since the Pauli operators form a basis).

	A vertex is said to be in the boundary of $A$ if $A$ does not contain all of its neighbors. A link is said to be across the boundary of $A$ if it contains one vertex within $A$ and another vertex outside of $A$. As well as $\mathcal{W}(A)$, the set of plaquette operators within $A$, we define the set of link generators within $A$, and the restrictions of link generators across the boundary of $A$ to $A$ itself as
	\begin{align}
		\mathcal{K}(A) &= \{K^{\alpha}_{\mbf{q}}\,:\,\mathrm{supp}(K^{\alpha}_{\mbf{q}})\in A\}	\,,\\
		\mathcal{B}(A) &= \{\sigma^{\alpha}_{\mathrm{supp}(K^{\alpha}_{\mbf{q}})\cap A}\,:\,|\mathrm{supp}(K^{\alpha}_{\mbf{q}})\cap A|=1\}	\,.
	\end{align}

	\begin{lemma}[Classification of Pauli operators]\label{L:linkGenSet}
	Every Pauli $P_\sub{A}\in \mathcal{Z}(\mathcal{W}(\Lambda))$ with support contained within region $A$ (formed of simply-connected plaquettes) is a link operator in $A$. In other words,
	\begin{align}
		\{P_\sub{A}\, :\,\mathrm{supp}(P_{\sub{A}})\in A, \; P_\sub{A}\in \mathcal{Z}(\mathcal{W}(\Lambda))\} = \langle \mathcal{K}(A) \rangle
	\end{align}
	 where $\langle \mathcal{S} \rangle$ denotes the group generated by set $\mathcal{S}$. Further, $\langle \mathcal{K}(A) \rangle$ contains at most $2^{\frac{3}{2}\abs{A}}$ linearly independent elements.

\end{lemma}

Note that, to avoid confusion between notation for group generation and expectation value, expectation values will always be labelled with a state subscript as in $\expect{O}_\psi$.

\begin{proof}
The number of vertices, edges, plaquettes, and boundary vertices in $A$ are given by $|A|$, $|\mathcal{K}|$, $|\mathcal{W}|$, and $|\mathcal{B}|$ respectively (in this proof we will suppress $A$ arguments of the sets $\mathcal{K}(A)$, $\mathcal{W}(A)$, and $\mathcal{B}(A)$ for clarity). Since $A$ is simply connected, we have from Euler's formula
	\begin{align}	\label{eq:eulerFormula}
		|A|-|\mathcal{K}|+|\mathcal{W}|=1\,.
	\end{align}
	All boundary vertices of $A$ are contained in two edges of $A$ while all other vertices of $A$ are contained in three edges of $A$. Since each edge contains two vertices, by adding up all the edges at each vertex, we have counted every edge exactly twice, giving $2|\mathcal{B}|+3(|A|-|\mathcal{B}|) = 2|\mathcal{K}|$. Rearranging, we find
	\begin{align}
		|\mathcal{B}| = 3|A|-2|\mathcal{K}|\,.
	\end{align}

	In general, given a set $\mathcal{S}$ of multiplicatively independent commuting Pauli stabilizers on $n$ qubits (such that $\expect{\mathcal{S}}\not\owns -I$), a minimal generator of the centralizer $\mathcal{Z}(\mathcal{S})$ has size $2n-|\mathcal{S}|$~\cite{Nielsen2000book}. We can use this fact to show that $\bar{\mathcal{K}} \equiv \mathcal{K}\cup \mathcal{B}$ generates $\mathcal{Z}_A(\mathcal{W})$ (the centralizer of $\mathcal{W}$ in the Pauli group on qubits in $A$) by a simple counting argument.
	
	There are $|\bar{\mathcal{K}}| = |\mathcal{K}|+|\mathcal{B}|$ operators in this generating set, but only $|\bar{\mathcal{K}}|-1$ of them are (multiplicatively) independent, since $\prod_{P\in\bar{\mathcal{K}}}P = I$. We know that $\mathcal{Z}_A(\mathcal{W})$ has a minimal generator of size $2|A|-|\mathcal{W}| = 3|A|-|\mathcal{K}|-1$. Since $\bar{\mathcal{K}}$ also contains $|\mathcal{K}|+|\mathcal{B}|-1 = 3|A|-|\mathcal{K}|-1$ independent generators, each of which lies in $\mathcal{Z}_A(\mathcal{W})$, we conclude that $\langle\bar{\mathcal{K}}\rangle = \mathcal{Z}_A(\mathcal{W})$ as claimed.
		
	It is also easy to see that each operator in $\mathcal{B}$ will anticommute with at least one $W$ stabilizer not contained in $A$, as will all of the operators generated by $\mathcal{B}$, except for the identity $I$ and $\prod_{P\in \mathcal{B}}P$. Since all elements of $\mathcal{K}$ are in $\mathcal{Z}(\mathcal{W}(\Lambda))$, and $\prod_{P\in \mathcal{B}}P = \prod_{P\in \mathcal{K}}P$ can be generated by the elements of $\mathcal{K}$, we conclude that $\mathcal{Z}_A(\mathcal{W}(\Lambda))=\langle \mathcal{K}\rangle$.
	
	Since each element of $\langle \mathcal{K}\rangle$ is self inverse and commutes with every other up to a phase, and since its minimal generator has size $|\mathcal{K}|$, the number of linearly independent elements in $\langle \mathcal{K}\rangle$ is simply $2^{|\mathcal{K}|}$. Noting that each edge involves 2 qubits, and every qubit is involved in either 2 or three edges within $A$, it is easy to see that $|A|\leq|\mathcal{K}|\leq\frac{3}{2}|A|$. This bounds the number of linearly independent elements of $\langle \mathcal{K}\rangle$ to be less than $2^{\frac{3}{2}\abs{A}}$. 
\end{proof}

The result of \Lref{L:linkGenSet} allows us to focus on link operators in $\expect{\mathcal{K}(A)}$. In order to analyze these operators, it will be convenient to make the following definitions:
\begin{definition}\label{D:endpoints} An \textbf{endpoint} of a link operator is a qubit with an odd number of link generators acting on it. A link operator is \textbf{string-like} if it has precisely two endpoints.
\end{definition}
In fact, we will now see that the expectation value of a link operator in a fixed honeycomb model code basis state depends (up to a phase) only on its endpoints.

Noting that any link operator has an even number of endpoints, we denote the $2M$ endpoints of a link operator $K\in\expect{\mathcal{K}(A)}$ as $\mbf{Q}_K = \left\{(\mbf{q}_m,s_m)\right\}_{m=1}^{2M}$, where $s_m \in \{+1,-1\}$ corresponding to $\bullet$ and $\circ$ sites, respectively.

Denote the code basis state $\ell=(l_x,l_y)$ (the state in the codespace that is also the $l_x$ eigenstate of $L_x$ and the $l_y$ eigenstate of $L_y$) by $\ket{\phi_\ell}$. In any such state, link generators $K^\alpha_\mbf{q}$ can be fermionized as in \Eref{E:fermionization}. Since $(c^\dagger_\mbf{q} \pm c_\mbf{q})^2 = \pm 1$ and $K^z_\mbf{q} = (c^\dagger_\mbf{q} + c_\mbf{q})(c^\dagger_\mbf{q} - c_\mbf{q})$, we can then calculate
\begin{align}	\label{E:fermionizedString}
	\expect{K}_{\phi_\ell} \sim \expect{\prod_{(\mbf{q},s)\in\mbf{Q}_K} \left(c^\dagger_{\mbf{q}} + s c_{\mbf{q}} \right)}_{\phi_\ell}	\,,
\end{align}
where $\sim$ reflects neglecting the possible $\pm 1$, $\pm i$ phases.
Note that this decomposition is contingent on the fact that $A$ is simply connected, since we can then without loss of generality assume that $A$ does not cross the periodic boundary of the lattice.
	
By the Fourier transformation in \Eref{E:fermionFourier} and the Bogoliubov transformation in \Eref{E:BogoliubovTrans}, we know that in code basis states,
\begin{align}
	c^\dagger_{\mbf{q}} = \frac{1}{N}\sum_{\mbf{k}} (u_{\mbf{k}}\gamma^\dagger_{\mbf{k}} + v_{\mbf{k}}^* \gamma_{-\mbf{k}}) e^{-i\mbf{k} \cdot \mbf{q}}	\,.
\end{align}
This gives for the expectation value of a link operator $K$:
\begin{align}	\label{E:stringExpect}
	\expect{K}_{\phi_\ell}	\sim \expect{\prod_{m=1}^{2M} \nu_m}_{\phi_\ell}	\,,
\end{align}
where 
\begin{align}
	\nu_m =  \frac{1}{N}\sum_{\mbf{k}_m} \left[
			 \left(u_{\mbf{k}_m}\gamma^\dagger_{\mbf{k}_m} + v_{\mbf{k}_m}^* \gamma_{-\mbf{k}_m}\right) e^{-i\mbf{k}_m \cdot \mbf{q}_m}
		 	+ s_m \left(u_{\mbf{k}_m}^* \gamma_{\mbf{k}_m} + v_{\mbf{k}_m}\gamma^\dagger_{-\mbf{k}_m}\right) e^{i\mbf{k}_m \cdot \mbf{q}_m} 
		\right]
\end{align}
Since code basis states are ground states of a free fermion Hamiltonian, we can apply Wick's theorem~\cite{Wick1950}, and we find
\begin{align} \label{E:wickthm}
	\expect{K}_{\phi_\ell} \sim \sum_{\mu} \mathrm{sgn}(\pi_\mu) \prod_{m<\mu(m)} \expect{\nu_{m}\nu_{\mu(m)}}_{\phi_\ell}\,.
\end{align}
where the sum is over all involutive derangements $\mu$ of $\mbf{Q}_K$ (i.e.~permutations that give rise to valid pairings, satisfying $\mu=\mu^{-1}$ and $\mu(m)\neq m$). Here $\pi_\mu$ is a permutation that takes the modes from normal order into an order in which $m$ and $\mu(m)$ are adjacent, and $m$ precedes (follows) $\mu(m)$ for $m<\mu(m)$ ($m>\mu(m)$). Each of the $\frac{(2M)!}{2^M M!}$ terms in this expansion has exactly $M$ factors.

Since code basis states $\phi_\ell$ are fermion vacuum states, and so are annihilated by $\gamma_{\mbf{k}_m}$ for all $\mbf{k}_m$, we can compute each expectation value explicitly as
\begin{align}	\label{E:FourierTransf_hk}
	\expect{\nu_i \nu_{j}}_{\phi_\ell}
	&=  \frac{1}{N^2}\sum_{\mbf{k}} \hat{h}_{\mbf{k}}(s_i, s_j) e^{i\mbf{k}\cdot(\mbf{q}_i - \mbf{q}_j)}	\,,
\end{align}
with discrete Fourier coefficients
\begin{align}
	\hat{h}_{\mbf{k}}(s_i, s_j) = 
	v_{-\mbf{k}}^*u_{\mbf{k}} + s_j\abs{v_{\mbf{k}}}^2 + s_i\abs{u_{\mbf{k}}}^2 + s_is_ju_{\mbf{k}}^* v_{-\mbf{k}}	\,.
\end{align}
Using the expressions (\ref{E:Bogoliubov_u}) and (\ref{E:Bogoliubov_v}), we find 
\begin{align}	
\label{E:FourierComp_hk}
	\hat{h}_{\mbf{k}}(s_i, s_j) = 
	\left\{
		\begin{array}{ll}
			\pm 1,   &\quad s_i=s_j\\
			\frac{\pm\xi_\mbf{k} - \Delta_\mbf{k}}{E_\mbf{k}} , 									  &\quad s_i\neq s_j\\
		\end{array}
	\right.
\end{align}

This allows us to simplify \Eref{E:wickthm} slightly, since $\expect{\nu_i\nu_j}_{\phi_\ell}$ will vanish for $s_i=s_j$. There are then at most $M!$ non-vanishing terms in this expression (since this is how many pairings of $M$ elements with $M$ elements there are). We denote the sum over these non-vanishing pairings as $\sum_{\tilde{\mu}}$ to distinguish it from the naive sum in \Eref{E:wickthm}. 

The reformulation given by equations (\ref{E:wickthm}-\ref{E:FourierComp_hk}) will allow us to bound the difference of expectation values of link operators between any two codestates. Before considering the general case, we will first demonstrate the result for the simpler class of string-like link operators.

Given the general expression of an expectation value for any $K \in \langle \mathcal{K}(A) \rangle$, we bound the difference of expectation values of code basis states, beginning with the following lemma.
\begin{lemma}[String-like link operator indistinguishability]
\label{L:single_string_case}
Let $K \in \langle \mathcal{K}(A) \rangle$ be a string-like link operator with endpoints $\mathbf{Q}_K = \{ (\mbf{q}_1, s_1), (\mbf{q}_2, s_2)\}$. Let $\mbf{\Delta q} = \mbf{q}_1 - \mbf{q}_2$, and $\phi_\ell$, $\phi_{\ell'}$ be any pair of common eigenstates of $L_x$ and $L_y$ in $\mathcal{C}_{\mathrm{khm}}$ with coupling constants as in \Eref{eq:couplingconstants}.
Then for $\|\mbf{\Delta q}\|_\infty \leq \tfrac{N-1}{2}$
\begin{align}
	\abs{\expect{K}_{\phi_\ell} - \expect{K}_{\phi_{\ell'}}} \leq c\, \mathrm{e}^{-aN/2},
\end{align}
where $c=\frac{4 \sqrt{2}}{\sinh(a/2)}$, $a = \ln\bigl(\tfrac{J_z-J_y}{J_x}\bigr)$.
\end{lemma}
\begin{proof}
This proof uses the same idea as the proof of \Thref{Th:qsdegeneracy}, but instead of bounding the difference between the rectangular rule and an integral, we will bound the difference between the discrete Fourier transform (DFT) from its continuous version. Consider a function $f:\R\rightarrow\C$, and define the error for the $N$-point discrete Fourier transform by
\begin{align}
	\epsilon_0(N, q)
	\equiv \frac{1}{2\pi} \int^{2\pi}_0 f(x) e^{ixq} \mathrm{d}x - \frac{1}{N} \sum_{n=0}^{N-1} f\Bigl(\frac{2\pi n}{N}\Bigr) e^{i\frac{2\pi n}{N}q}.
\end{align}
Let us recall the conditions that we used in \Thref{Th:qsdegeneracy}: Suppose that $f$ is also analytic and $2\pi$-periodic. Then there exists a strip in the complex plane $D_0=\R\times(-a_0,a_0)\in\C$ with $a_0>0$ such that $f$ can be extended to a holomorphic and $2\pi$-periodic bounded function $f:D_0\rightarrow\C$. We will choose the largest possible strip, so that $a_0$ is defined as the supremum of the half-width over of all such strips for $f$. Let $Q_0$ be an upper bound for $\abs{f}$ on $D_0$; in particular we choose $Q_0=\sup_{D_0}\abs{f}$. Our error term for the DFT can be bounded by a theorem of Epstein~\cite{Epstein2005} and we find that
\begin{align}
	\abs{\epsilon_0(N, q)}\leq \frac{2Q_0e^{-a_0\left(\frac{N-1}{2}\right)}}{e^{a_0} - 1} \,.\label{E:erbnd_DFT_1D}
\end{align} 

Since our integration is two-dimensional, we define the error of the two-dimensional DFT as
\begin{align}
	\epsilon(N, \mbf{q}) 
   	\equiv 
    \frac{1}{4\pi^2} \int^{2\pi}_0 \int^{2\pi}_0 f(x,y) e^{i(xq_x + yq_y)}\mathrm{d}x\mathrm{d}y
	- \frac{1}{N^2} \sum_{n_x=0}^{N-1} \sum_{n_y=0}^{N-1} 
	f\Bigl(
		\frac{2\pi n_x}{N},\frac{2\pi n_y}{N}
	\Bigr) e^{i\frac{2\pi}{N}(n_x q_x + n_y q_y)}\,. 
\end{align}
Following the same procedure as in the proof of \Thref{Th:qsdegeneracy}, the error bound of two dimensional error based on the one-dimensional case is
\begin{align}
	\abs{\epsilon(N, \mbf{q})} \leq \frac{2 Q \mathrm{e}^{-a \left(\frac{N-1}{2}\right)}}{\mathrm{e}^{a} - 1} + \frac{2R \mathrm{e}^{-b\left(\frac{N-1}{2}\right)}}{\mathrm{e}^{b} - 1}	\,, 
\end{align}
where $a, b, Q, R$ have the same definitions given previously in the approximate degeneracy proof, but of course we will apply them to the function $\hat{h}_\mbf{k}(s_i, s_j)$ instead of the energy density.

Now let us consider the expectation values of a string-like link operator in the code basis states $\ell$ and $\ell'$. Since there are only two endpoints of a string-like link operator, \Eref{E:wickthm} consists of a single term, and \Eref{E:FourierTransf_hk} allows us to convert this calculation into the approximate evaluation of a continuous Fourier transform with error
\begin{align}
	\frac{1}{4\pi^2} \int_0^{2\pi} \int_0^{2\pi} \hat{h}(\mbf{x};s_1, s_2) e^{-i\mbf{x}\cdot(\mbf{\Delta q})}\mathrm{d}x\mathrm{d}y = \expect{K}_{\phi_\ell} + \epsilon_\ell(N, \mbf{\Delta q})	\,,
\end{align}
where $\mbf{x} = (x, y)$.
Then from the triangle inequality we have
\begin{align}
	\abs{\expect{K}_{\phi_\ell} - \expect{K}_{\phi_{\ell'}}} \leq \abs{\epsilon_\ell(N, \mbf{\Delta q})} + \abs{\epsilon_{\ell'}(N, \mbf{\Delta q})}\,,
\end{align}

Recall that $\hat{h}_\mbf{k}(s_i, s_j)$ in \Eref{E:FourierComp_hk} is $\pm 1$ for $s_i=s_j$ or a function only of basic trigonometric functions of $\mbf{k}$ for $s_i \neq s_j$. In the case $s_1=s_2$, since $\sum_\mbf{k} e^{i\mbf{k}\cdot\mbf{\Delta q}} = 0$, the expectation values both vanish for all $\mbf{\Delta q}$, and their difference is 0. 
In the other case $s_1 \neq s_2$ the integrand $\hat{h}_\mbf{k}(s_i, s_j)$ is $2\pi$-periodic, and it is analytic in the gapped phase. 
For a $2\pi$-periodic function $f(x)$, its Fourier transformation $\int_0^{2\pi}f(x+\delta)e^{i(x+\delta)q}\mathrm{d}x = \int_0^{2\pi}f(x)e^{ixq}\mathrm{d}x$ for arbitrary real constant $\delta$, so we can apply Epstein's theorem to bound the error independent of $\ell$, and we find
	\begin{align}
		\abs{\expect{K}_{\phi_\ell} - \expect{K}_{\phi_{\ell'}}} &\leq 2\left( \frac{2\tilde{Q} \mathrm{e}^{-a\left(\frac{N-1}{2}\right)}}{\mathrm{e}^{a} - 1} + \frac{2\tilde{R} \mathrm{e}^{-b\left(\frac{N-1}{2}\right)}}{\mathrm{e}^{b} - 1} \right) \\
		&\leq 2\left( \frac{2(\tilde{Q}+\tilde{R})}{\mathrm{e}^{\min\{a, b\}} - 1} \right) \mathrm{e}^{-\min\{a, b\}\left(\frac{N-1}{2}\right)}\,,\label{E:stringboundnoconstants}
	\end{align}
where, similarly to the proof of \Thref{Th:qsdegeneracy}, we define
	\begin{align}
		\tilde{Q} = \max_{y} \tilde{Q}(y)\,, \quad a = \min_{y} a(y)\,, \quad \tilde{R} = \max_{x} \tilde{R}(x) \,, \quad b = \min_{x} b(x) \,,
	\end{align}
and where the implicit function dependence is on $\hat{h}_\mbf{k}(s_i, s_j)$. Note that we have chosen to maximize $\tilde{Q}(y)$ and $\tilde{R}(y)$ now because this specific expression does not obviously enjoy the same monotonicity properties that allowed us to use a minimum in the proof of \Thref{Th:qsdegeneracy}. The constants $a$ and $b$ turn out to be identical to those in the approximate degeneracy bound, since the point of nonanalyticity is still determined by the energy. A similar calculation as in the proof of \Lref{L:expConstants} then shows that $\tilde{R}, \tilde{Q} \le \sqrt{2}$. We omit the details of these calculations, and the lemma then follows straightforwardly from \Eref{E:stringboundnoconstants} with these values for the constants.
\end{proof}

\begin{lemma}[General link operator indistinguishability]
\label{L:many_string_case}
Let $K \in \langle \mathcal{K}(A) \rangle$ be a link operator with endpoints $\mbf{Q}_K = \left\{(\mbf{q}_m,s_m)\right\}_{m=1}^{2M}$, $\mbf{\Delta q}_{mn} = \mbf{q}_m - \mbf{q}_n$, and $\phi_\ell$, $\phi_{\ell'}$ be any pair of common eigenstates of $L_x$ and $L_y$ in $\mathcal{C}_{\mathrm{khm}}$ with coupling constants as in \Eref{eq:couplingconstants}.
Then for $\max_{m,n}\|\mbf{\Delta q}_{mn}\|_\infty < \tfrac{N-1}{2}$
\begin{align}
	\abs{\expect{K}_{\phi_\ell} - \expect{K}_{\phi_{\ell'}}} \leq c\, \mathrm{e}^{-aN/2 + M\ln M}\,,
\end{align}
where $c = \tfrac{4 \sqrt{2}}{\sinh(a/2)}$ and $a = \ln\bigl(\tfrac{J_z-J_y}{J_x}\bigr)$
\end{lemma}

\begin{proof}
Given the general form of $K$ in \Eref{E:wickthm}, we first use the triangle inequality to obtain
\begin{align}
	  \abs{\expect{K}_{\phi_\ell}- \expect{K}_{\phi_{\ell'}}} 
	  \le \sum_{\tilde{\mu}} \abs{\prod_{m<\tilde{\mu}(m)} \expect{\nu_{m}\nu_{\tilde{\mu}(m)}}_{\phi_\ell} - \prod_{m'<\tilde{\mu}(m')} \expect{\nu_{m'}\nu_{\tilde{\mu}(m')}}_{\phi_{\ell'}} }\,.
\end{align}
Then we use the following identity of a telescoping sum and again apply the triangle inequality, 
\begin{align}
	\prod_{m=1}^M A_m - \prod_{m=1}^M B_m &= \sum_{m=1}^M\left(\prod_{k=1}^{m-1}A_k\right)(A_m - B_m)\left(\prod_{l=m+1}^{M}B_l\right) \notag \\
	&\le \sum_{m=1}^M \prod_{k=1}^{m-1}\abs{A_k} \cdot \abs{A_m - B_m} \cdot \prod_{l=m+1}^{M}\abs{B_l}\,,
\end{align}
where by convention $\prod_{k=1}^{0} = \prod_{l=M+1}^{M} \equiv 1$. 
Assuming that $\abs{A_k}, \abs{B_k} \le 1$, we can write an even simpler inequality,
\begin{align}
	\prod_{m=1}^M A_m - \prod_{m=1}^M B_m \le \sum_{m=1}^M \abs{A_m - B_m} \le M \max_m \abs{A_m - B_m} \,.
\end{align}

Now recall that $\expect{\nu_i\nu_j}_{\phi_\ell}$ is the DFT of $\hat{h}_\mbf{k}$, and $\bigl|\hat{h}_\mbf{k}\bigr| = 1$ in \Eref{E:FourierComp_hk}, so we have $\bigl|\expect{\nu_i\nu_j}_{\phi_\ell}\bigr| \leq 1$.  Moreover, from \Lref{L:single_string_case} we have a uniform bound for the difference of a single pair from the string-like case (given the size restriction $\|\mbf{\Delta q}\|_\infty < \tfrac{N-1}{2}$), namely 
\begin{align}
	\Bigl| \expect{\nu_{m}\nu_{\mu(m)}}_{\phi_\ell} - \expect{\nu_{m}\nu_{\mu(m)}}_{\phi_{\ell'}} \Bigr| \le c\, \mathrm{e}^{-a N/2} \,.
\end{align}
Putting these ingredients together, we find that
\begin{align}
	  \abs{\expect{K}_{\phi_\ell}- \expect{K}_{\phi_{\ell'}}} \le & \sum_{\tilde{\mu}} M c\, \mathrm{e}^{-a N/2} = M! \cdot M c\, \mathrm{e}^{-a N/2} \le c\, \mathrm{e}^{-a N/2 + M\ln M}\,,
\end{align}
where we have used the fact that $M\cdot M! \leq M^M = \mathrm{e}^{M\ln M}$.
\end{proof}

\begin{theorem}[Arbitrary local operator indistinguishability]
\label{T:arbitrary_operator}
Let $A$ be a simply connected region with volume $\abs{A} \le \frac{a N/4}{\ln(aN/4+1)}$. Then with coupling constants as in \Eref{eq:couplingconstants}, the code space $\mathcal{C}_{\mathrm{khm}}$ is $\Delta$-approximately indistinguishable on all such regions $A$, i.e.\ 
\begin{align}
	\sup_{\psi_1,\psi_2\in \mathcal{C}_{\mathrm{khm}}}\abs{\expect{O_\sub{A}}_{\psi_1}-\expect{O_\sub{A}}_{\psi_2}}\leq \Delta(N)\cdot\|O_\sub{A}\|  \,,
\end{align}
for all operators $O_\sub{A}$ whose support is contained in $A$, and with decay function given by
\begin{align}
\label{eq:deltascaling}
	\Delta(N) = c\, \mathrm{e}^{-a N/4}\,.
\end{align} 
Here the constants can be chosen as $c = \tfrac{4 \sqrt{2}}{\sinh(a/2)}$ and $a = \ln\bigl(\tfrac{J_z-J_y}{J_x}\bigr)$.
\end{theorem}

For any fixed values of the coupling constants $J_{\alpha}$, this theorem shows that the honeycomb codespace satisfies $\Delta$-approximate local indistinguishability at any length scale $N^* \le \mathcal{O}\bigl(\sqrt{N/\ln N}\bigr)$, corresponding to any $\beta < 1/2$ in \Dref{D:Apprxlocalindist}, and so is approximately locally indistinguishable. We remark that the proof below does not depend on the geometry of the regions considered, only on their volume and topology (simply connected).

\begin{proof}
A general operator supported on $A$ will consist of linear combinations of Pauli operators that either commute or anticommute with the plaquette operators. 
Operators that anticommute are always trivial in expectation by \Lref{L:LocalIndAntiCommWq}, so we can expand any nontrivial part of a general operator in terms of $K_i\in \langle \mathcal{K}(A) \rangle$ only by \Lref{L:linkGenSet}. 
Write a general operator supported on the $K_i$ as $O_\sub{A} = \sum_i r_i K_i$. 

Next note that the operators $L_x$ and $L_y$ commute with the entire link group $\mathcal{K}(A)$. 
Therefore any observable supported only on elements of the link group breaks into a direct sum labeled by the eigenvalues of the logical operators. 
In particular, any off-diagonal element of the form $\langle \phi_{\ell}| O_\sub{A} |\phi_{\ell'}\rangle$ must vanish for code basis states $\phi_{\ell} \not = \phi_{\ell'}$ (since they are simultaneous eigenstates of $L_x$ and $L_y$). 

Now consider an arbitrary (potentially mixed) state $\rho$ supported only on $\mathcal{C}_\mathrm{khm}$. 
The expectation value of an operator $O_\sub{A}$ supported on $\mathcal{K}(A)$ is given by $\expect{O_\sub{A}}_{\rho} = \Tr(O_\sub{A}\rho)$. 
From the previous paragraph, we know that $O_\sub{A}$ can be chosen to be diagonal in the basis of states $\{\phi_\ell\}$. 
Therefore by convexity and the variational characterization of eigenvalues, we have
\begin{align}
	\min_{\phi_\ell} \expect{O_\sub{A}}_{\phi_\ell} \leq \expect{O_\sub{A}}_{\rho\in \mathcal{C}_\mathrm{khm}} \leq \max_{\phi_\ell} \expect{O_\sub{A}}_{\phi_\ell}\,.
\end{align}
This implies that the difference between expectation values of any $\rho_1$ and $\rho_2$ in $\mathcal{C}_\mathrm{khm}$ is bounded by
\begin{align}
	\abs{\expect{O_\sub{A}}_{\rho_1}- \expect{O_\sub{A}}_{\rho_2}} \leq \max_{\phi_\ell,\phi_{\ell'}}\abs{\expect{O_\sub{A}}_{\phi_\ell}- \expect{O_\sub{A}}_{\phi_{\ell'}}}	\,.
\end{align}
It then suffices to prove the bound only for these eigenstates.

Using the expansion into Pauli operators inside $A$, we apply H\"{o}lder's inequality to obtain
\begin{align}
	\abs{\expect{O_\sub{A}}_{\phi_\ell} - \expect{O_\sub{A}}_{\phi_{\ell'}}} 
	= \abs{\sum_i r_i \bigl( \expect{K_i}_{\phi_{\ell}} - \expect{K_i}_{\phi_{\ell'}}\bigr)}
	\leq  \|r\|_1 \max_i \abs{\expect{K_i}_{\phi_{\ell}} - \expect{K_i}_{\phi_{\ell'}}} \,.
\end{align}
Since the $K_i$ are all unitary trace-orthogonal Paulis acting on a $2^{2N^2}$-dimensional space, we can bound the 2-norm of the coefficient vector $r$ as follows, 
\begin{align}
	\|O_\sub{A}\|_{\mathrm{F}} =\Tr\bigl(O_\sub{A}^\dagger O_\sub{A}\bigr)^{1/2} = 2^{N^2} \|r\|_2 \le 2^{N^2} \|O_\sub{A}\| \,,
\end{align}
where $\|O_\sub{A}\|_{\mathrm{F}}$ is Frobenius norm. 
Because there are at most $2^{3\abs{A}/2}$ linearly independent $K_i$ by \Lref{L:linkGenSet}, we therefore have the bound on the 1-norm
\begin{align}
	\|r\|_1 \le 2^{3\abs{A}/4} \|r\|_2 \le 2^{3\abs{A}/4}  \|O_\sub{A}\| \,.
\end{align}

From the bound in \Lref{L:many_string_case}, and the fact that a region with $\abs{A}$ spins has at most $2M \le \abs{A}$ string endpoints, we obtain
\begin{align}
	\abs{\expect{O_\sub{A}}_{\phi_\ell} - \expect{O_\sub{A}}_{\phi_{\ell'}}} \leq  2^{3 \abs{A}/4}  \|O_\sub{A}\| \frac{4 \sqrt{2}}{\sinh(a/2)} \mathrm{e}^{-a N/2 + \tfrac{\abs{A}}{2} \ln \tfrac{\abs{A}}{2}} \,,
\end{align}
We have $\abs{A} \ge 6$ since we only consider simply connected plaquettes and the smallest size region must contain at least one plaquette. Then we have the inequality
$2^{3 \abs{A}/4}  \le \mathrm{e}^{ \tfrac{\abs{A}}{2} \ln \tfrac{\abs{A}}{2}}$ and the bound simplifies,
\begin{align}
	\abs{\expect{O_\sub{A}}_{\phi_\ell} - \expect{O_\sub{A}}_{\phi_{\ell'}}} \leq  \|O_\sub{A}\| \frac{4 \sqrt{2}}{\sinh(a/2)}\mathrm{e}^{-a N/2 + \abs{A} \ln \abs{A}} \,.
\end{align}
To achieve our desired decay function with decay rate $a/4$, we require that the region size $\abs{A}$ satisfies $\abs{A} \ln \abs{A} \le \frac{a N}{4}$. 
It is straightforward to show that choosing
\begin{align}
	\abs{A} \le \frac{a N/4}{\ln(aN/4+1)} \,
\end{align}
gives the desired result for all $a > 0$.
\end{proof}

For some applications, a slightly different form of this result is desirable that emphasizes the projection into the code space rather than differences of expectation values. 
We state here an immediate corollary of \Thref{T:arbitrary_operator} that is more convenient in these cases. 
\begin{corollary}\label{c:topoorder}
Under the same conditions as \Thref{T:arbitrary_operator}, and with $\Pi$ the projector onto the code space $\mathcal{C}_{\mathrm{khm}}$, we have
\begin{align}\label{e:cortopoorder}
	\Bigl\| \Pi O_\sub{A} \Pi - \tfrac{\Tr(\Pi O_\sub{A})}{\Tr \Pi} \Pi \Bigr\| \le \tfrac{3}{4} \Delta(N) \|O_\sub{A}\| \,.
\end{align}
\end{corollary}
\begin{proof}
Since $\Pi O_\sub{A} \Pi$ is codiagonal with $\Pi$, the proof follows immediately from \Thref{T:arbitrary_operator} and the triangle inequality, noting only that $\Tr \Pi = 4$.
\end{proof}
		
	\subsection{Consequences for error correction}\label{s:distance}

In \Sref{s:asymptotics}, local indistinguishability was derived as a consequence of the code properties of the honeycomb model. These two notions are closely connected, and here we will illustrate this by using the indistinguishability results obtained in the previous section to derive consequences for the error correction properties of the honeycomb code. 

However, here it will be convenient to use a slightly different formulation of approximate codes to that used previously. In \Sref{s:asymptotics}, we used a notion of approximate codes due to Flammia et.~al.~\cite{Flammia2016} suited to the topological arguments of that section. This formalism restricts to codes with some notion of geometric locality, and treats erasure errors preferably. In contrast, we will now consider a slightly more abstract formulation of approximate codes due to Beny \& Oreshkov~\cite{Beny2010}. In this formalism, approximate correctability is defined as follows:

\begin{definition}\label{d:epscorrect}
A code defined by an encoding map $\mathcal{E}$ is called \emph{$\epsilon$-correctable under a noise channel $\mathcal{N}$} if there exists a recovery map $\mathcal{R}$ such that $d(\mathcal{RNE},I)\leq \epsilon$, where
\[
d(\mathcal{M},\mathcal{N}) = \max_{\rho} \sqrt{1-f((\mathcal{N}\otimes I)(\psi),(\mathcal{M}\otimes I)(\psi))}
\]
is the Bures distance based on the entanglement fidelity, with $\psi$ a purification of $\rho$ and $f(\rho,\sigma) = \Tr\sqrt{\sqrt{\sigma}\rho\sqrt{\sigma}}$ the fidelity.
\end{definition}

Ref.~\cite{Beny2010} also provides a useful characterization of codes that are approximately correctable: 
\begin{theorem}[Approximate code correctability~{\cite[Cor.~2]{Beny2010}}]
A code space $\mathcal{C}$ is $\epsilon$-correctable under a noise channel with elements $\{F_i\}$ iff
\begin{align}
\Pi_{\mathcal{C}}F_i^{\dagger}F_j\Pi_{\mathcal{C}} = g_{ij}\Pi_{\mathcal{C}}+\Pi_{\mathcal{C}}T_{ij}\Pi_{\mathcal{C}}
\end{align}
for $\Pi_{\mathcal{C}}$ the projector to $\mathcal{C}$, $T_{ij}$ arbitrary matrices, $g_{ij}$ the components of a density operator, and with $d(\Gamma+\mathcal{T},\Gamma)\leq \epsilon$, where $\Gamma(\rho)=\Tr(\rho)\sum_{ij}g_{ij}\ketbra{i}{j}$ and $\mathcal{T}(\rho) = \sum_{ij}\Tr(\rho T_{ij})\ketbra{i}{j}$.
\end{theorem}

Using Beny \& Oreshkov's conditions together with the result of \Cref{c:topoorder} yields the following bound on the correctability of the honeycomb code. 
Recall that as in the previous section, we restrict our attention to simply connected regions composed of whole plaquettes. 
This easily extends to arbitrary simply connected regions by expanding the region slightly to complete any partial plaquettes.
\begin{theorem}[Local approximate correctability of the honeycomb code]\label{T:localcorrect}
The honeycomb code space $\mathcal{C}_{\mathrm{khm}}$ is $\epsilon(N)$-correctable under a noise channel with elements supported in any simply connected region $A$ composed of complete plaquettes with volume $|A|\leq \min\{\frac{aN}{8\ln 4},\frac{aN/4}{\ln(aN/4+1)}\}$, where 
\begin{align}
	\epsilon(N) \le  \sqrt{\tfrac{3c}{8}}\e^{-aN/16} \,,
\end{align} 
and where $c = \tfrac{4 \sqrt{2}}{\sinh(a/2)}$ and $a = \ln\bigl(\tfrac{J_z-J_y}{J_x}\bigr)$.
\end{theorem}

\begin{proof}
Consider an orthonormal basis for operators supported in $A$ (of size $4^{|A|}$), and take any two basis elements $F^A_i$ and $F^A_j$. Then note that $O_{ij} = \left(F^A_i\right)^{\dagger}F^A_j$ is also supported in $A$, and so satisfies \Eref{e:cortopoorder}. This gives that
\begin{align}
	\Pi \left(F^A_i\right)^{\dagger}F^A_j \Pi = \Pi O_{ij}\Pi &= \tfrac{\Tr(\Pi O_{ij})}{\Tr \Pi} \Pi + \Pi T_{ij} \Pi
\end{align}
for some $T_{ij} = \Pi T_{ij} \Pi$ with $\|T_{ij}\|\leq \tfrac{3}{4} \Delta(N) \|O_{ij}\|\leq \tfrac{3}{4} \Delta(N)$. Setting $g_{ij}=\tfrac{\Tr(\Pi O_{ij})}{\Tr \Pi}$, it can easily be verified that $\Gamma$ is trace-one and Hermitian as required.

Next we relate the fidelity to the trace distance using the Fuchs-van de Graaf inequality~\cite{Fuchs1999}, which says that for any pair of states $\rho$ and $\sigma$, we have
\begin{align}
	1-f(\rho,\sigma) \le \frac{1}{2} \|\rho -\sigma \|_1 \,.
\end{align}
It follows that it suffices to bound
\begin{align}
	d(\Gamma + \mathcal{T},\Gamma)^2 \le \max_\rho \frac{1}{2} \left\|\left((\Gamma+\mathcal{T})\otimes I\right)(\psi) -\left(\Gamma\otimes I\right)(\psi)\right\|_1 &= \max_\rho \frac{1}{2}\left\|\left(\mathcal{T}\otimes I\right)(\psi)\right\|_1 \,,
\end{align}
where $\psi$ is a purification of $\rho$. For any matrix $M$, we have $\|M\|_1 \le \dim (M) \max_{i,j} |M_{i,j}|$, and then using the matrix H\"older inequality it follows that for all $\psi$ 
\begin{align}
	\left\|\left(\mathcal{T}\otimes I\right)(\psi)\right\|_1 &\le \dim(\psi) \max_{i,j,x,y} \Tr(T_{ij}\otimes |y\rangle\!\langle x|\, \psi)\\
	&\le \dim(\rho)^2 \, \|\psi\|_1 \, \|T_{ij}\otimes |y\rangle\!\langle x|\| \\
	&\le 4^{|A|}\, \|T_{ij}\|\,  \| |y\rangle\!\langle x|\| \\
	&\le 4^{|A|} \tfrac{3}{4}\Delta(N)  \,.
\end{align}
Here the second inequality also uses the upper bound on the purification dimension, and subsequent lines use submultiplicativity and \Eref{e:cortopoorder}. 

Putting this together with \Eref{eq:deltascaling}, we find that 
\begin{align}
	d(\Gamma+\mathcal{T},\Gamma) &\leq \sqrt{\tfrac{3c}{8} 4^{|A|} \e^{-aN/4}}\\
	&= \sqrt{\tfrac{3c}{8}\e^{-aN/4+|A|\ln 4}}
\end{align}
Since $|A|\leq \frac{aN}{8\ln(4)}$ by assumption, the result follows directly. \end{proof}

This demonstrates that, as with the approximate degeneracy and approximate local indistinguishability of the honeycomb model, its approximate code properties can also be demonstrated more precisely than was shown using general topological arguments of \Sref{s:asymptotics}. However, the present result is slightly weaker in the sense that it only holds for simply-connected regions $A$. We anticipate that this requirement could be relaxed by modifying the arguments in \Sref{s:approxlocalindistinguish} if desired, but that this would be relatively involved.

\section{Incoherent noise}\label{s:incoherent}

\subsection{Simulation of thermalization processes}\label{s:incoherentmodel}
	
		A general framework for simulating many-body quantum systems interacting with a thermalizing environment is to model the dynamics by a continuous-time Markov process, where we assume that each spin couples to an independent environment (for a general discussion on thermalization in these kinds of systems, see e.g.~Ref.~\cite{Brown2016}). A common form for such dynamics is given by the spin-boson model, which considers errors $F$ that decrease the energy of the system by $\Delta_F$. It assigns rates to such processes as
		\begin{align}	\label{E:gammarate}
			\gamma_F&= \frac{\Delta_F}{1-\e^{-\beta\Delta_F}}
		\end{align}
		for inverse temperature $\beta$. These dynamics have desirable properties such as being motivated by a microscopic model, having the Gibbs state as a fixed state, and satisfying the detailed balance condition. To simulate thermalization using this model, a set of (typically ergodic) allowed errors is chosen, and processes are chosen at random according to their relative rates for a Poisson distributed number of time steps around the desired simulation time. For more details about these standard simulation techniques, see~\cite{Chesi2010a, Brown2016}.
		
		In order for this approach to be well motivated, it is important that an error $F$ will map energy eigenstates to energy eigenstates, so that $\Delta_F$ is well-defined. This is the case for Pauli errors $F$ and stabilizer Hamiltonians, but in general, Hamiltonians that admit local errors $F$ with this property are very special. In particular, this assumption would not be satisfied if we tried to simulate single-site Pauli errors and take the energies as given by the Hamiltonian (\ref{E:KHM}). However, we will work in the perturbative limit, where the effective Hamiltonian (\ref{e:pcutH}) is a good approximation to the true dynamics, up to a close-to-trivial unitary transformation. When restricted to the lowest energy eigenspace of the unperturbed Hamiltonian, this effective Hamiltonian reduces to (\ref{e:loweffectiveham}), a stabilizer Hamiltonian, where the spin-boson simulation methods can easily be applied for local Pauli errors. As such, we will choose errors $F$ to be arbitrary single-qubit Pauli operators. The main subtlety is that these ``single-site'' errors we apply in the effective model actually correspond to quasi-local operators in the physical model, but this does not seem physically unreasonable, and in any case strict locality is restored in the perturbative limit.
		
		By restricting to the lowest energy eigenspace of the unperturbed Hamiltonian, we would be significantly restricting the possible errors that may occur. In some of the regimes we consider, the system will (almost) never leave this eigenspace over the timescales of interest. However, it is also desirable to provide some account of the dynamics of our system outside these special regimes. In order to do this, we note that in a system where the energy spectrum is suitably close to that of a stabilizer Hamiltonian, we might nevertheless expect to approximate the thermalization process on small lattices over relatively short timescales and low energy scales by similar Monte-Carlo methods. 
		
		In particular, if an error $F$ takes Hamiltonian eigenstates to states whose spectral support is mostly contained within a small energy band, we can use the energy expectation value as an approximation to calculate $\Delta_F$. On general grounds, we expect that in the perturbative limit and at low energies, local Pauli operators $F$ will jump between narrow spectral bands in our system. A simulation where $\Delta_F$ is approximated by the difference of energy expectation values will then be reasonable so long as the total number of jumps is also small (i.e.~small timescales and lattice sizes), though we leave a quantitative justification of the accuracy of such a simulation for future work.
		
	\subsection{Error-correction protocols}
	
		In standard topological error correction protocols, the system is imagined to be initialized in a ground state of the Hamiltonian, before undergoing some noise channel, often a model for thermal noise. Local measurements are made on the system to determine an error syndrome, which is then fed into a decoding algorithm that will determine an appropriate recovery operation to return the system to the code space. Our protocol is identical to this setup, except for two small subtleties. The first is that the pseudo-ground space of our model is used as a code space in lieu of a truly degenerate ground space. The second is that we imagine measuring local operators in the effective model of \Eref{e:pcutH}. These operators correspond to quasi-local operators in the original Hamiltonian (\ref{E:KHM}), as opposed to strictly local measurements.
		
		Our decoding algorithm is based on the standard Perfect Matching Algorithm (PMA) decoders for the toric code~\cite{Edmonds65paths,Kolmogorov2009}. These can be applied with only minor modification in our setting. Recall that in the limit $J_x,J_y\ll J_z$, the KHM approaches the toric code Hamiltonian, where each toric code qubit corresponds to the 2-dimensional $K^z_\mbf{q}=+1$ eigenspace at a given dimer of the honeycomb model. Before being able to apply a standard toric code decoding algorithm, our decoder must return the system to the common $K^z_\mbf{q}=+1$ eigenspace at each edge.
			
		Explicitly, the error correction protocol first takes measurements of all $K^z_{\mbf{q}}$ and $W_{\mbf{q}}$ operators. Following these measurements, any $K^z_{\mbf{q}}$ operator that gave a $-1$ measurement outcome is returned to its $+1$ eigenspace by the application of an $X$ operator on either of the qubits of the corresponding edge (for concreteness, we choose the upper qubit). This will of course modify the measured $W$ syndrome so that it is no longer an accurate representation of the state, though the updated $W$ syndrome can easily be calculated by the Gottesman-Knill method~\cite{gottesman1998heisenberg} without needing to remeasure the system. Within the $+1$ eigenspace of all $K^z_\mbf{q}$ operators, the $W_\mbf{q}$ operators correspond to the stabilizers of the toric code, and so the updated $W_\mbf{q}$ syndrome is passed to the PMA toric code decoder which returns an appropriate final recovery operation.  Since the details of a PMA decoder for the toric code are well studied, we will not describe this algorithm, and refer the interested reader to e.g.~Ref.~\cite{Dennis2002} and references therein. Following this procedure, we are guaranteed to return to a state in the common $+1$ eigenspace of all $K^z_\mbf{q}$ and $W_\mbf{q}$ operators. Notice that this is precisely the ground space of the effective Hamiltonian (\ref{e:pcutH}). We regard the error correction as having been successful if the combined noise and recovery operations commute with the logical operators of the code. While this is far from an optimal decoding strategy, its effectiveness is sufficient for this exploratory study, and its close relation to the analogous toric code protocol makes it suitable for direct comparison of results.
		
		One final remark is that in these simulations, we neglect the coherent sources of error studied in \Sref{s:coherent}, such as dephasing due to inexact degeneracy of the ground space. Since we have proved that these effects are exponentially suppressed in the system size, neglecting these additional errors is justified for our desired precision.

	\subsection{Numerical results}

		We conduct simulations over a broad range of temperatures, focusing on several distinct regimes. The first regime we test is the ultra-high-temperature limit, in which the rates $\gamma_F\approx\frac{1}{\beta}$ are independent of the error $F$ and so we may effectively disregard the Hamiltonian. This serves as a benchmark of our model, corresponding to the well studied case of i.i.d.~depolarizing noise on each spin.
		
		Following this benchmark, we consider three temperature regimes: low, intermediate and high. The low temperature regime for our system corresponds roughly to $\beta\gtrsim 10^5$, where we find interesting finite-size effects beginning to become significant. The intermediate temperature regime corresponds roughly to $\beta \sim 10^4$, and is chosen such that $\beta\Delta_F\sim 1$ for errors $F$ that create two vortices. In the low- and intermediate-temperature regimes, for the lattice sizes studied and over timescales comparable to the memory lifetime of the system, the probability of encountering a broken dimer is negligible. The high temperature regime has $\beta\sim 10^1$, which is chosen such that $\beta\Delta_F\sim 1$ for errors $F$ that create a broken dimer. At these temperatures, the density of $-1$ $K^z_\mbf{q}$ eigenvalues becomes measurable during our simulations, and so the mean-energy approximation is required.
		
		All numerical results presented in this work have parameters $J_z=\frac{1}{2}$ and $J_x=J_y=0.1$, and energies are taken by the effective Hamiltonian (\ref{e:pcutH}) calculated to $6\th$ order. For convenience, we have also taken different boundary conditions to those used in the previous sections. Instead of the lattice of \Fref{f:honeycomb}, we will now use the rotated boundary conditions shown in \Fref{f:honeycombrotate}, and denote the corresponding linear lattice size by $\tilde{N}$.
		
	\begin{figure}
		\centering
		\includegraphics{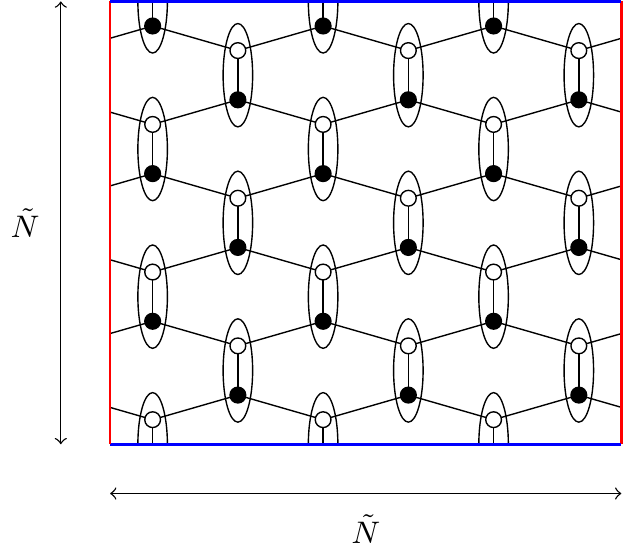}
		\caption{In our numerical simulations, we use boundary conditions as shown. These are rotated by $45^\circ$ compared to those introduced in \Fref{f:honeycomb}. We denote the linear length (measured as the number of dimers) along the new boundaries as $\tilde{N}$ to distinguish it from the $N$ used previously.}
		\label{f:honeycombrotate}
	\end{figure}

	\subsubsection{The ultra-high temperature benchmark}
		
		In the infinite temperature limit, the Hamiltonian plays no role in the thermalization dynamics, and so we can approximate these dynamics by an i.i.d.~depolarizing channel for each qubit. We can also neglect coherent evolution in this regime, since it is assumed to act on a much longer timescale than the memory lifetime, which will be extremely short. As such, this can be treated as a benchmark for the honeycomb model and readily compared to error correction of the standard toric code under i.i.d.~depolarizing noise. 
		
	\begin{figure}[H]
		\centering
			\includegraphics[width=0.5\textwidth]{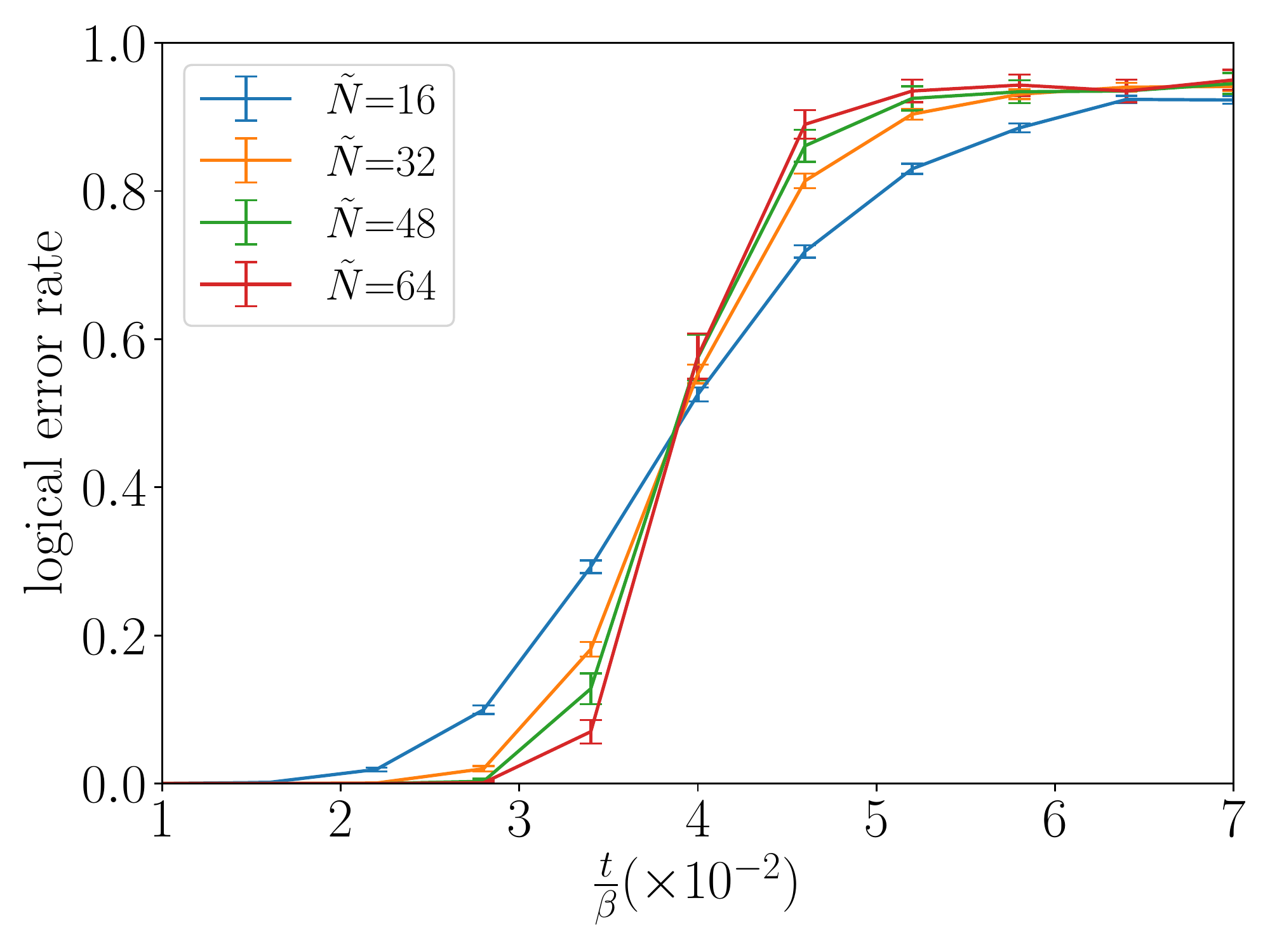}
		\caption{Logical error rate against time for various system sizes under a depolarizing channel. Error bars are 95\% confidence intervals. We find clear evidence of threshold behavior at $t_c\approx 0.04\beta$. This linear scaling of the lifetime is as expected for low $\beta$.}
		\label{f:depolresult}
	\end{figure}

		The results of this simulation are shown in \Fref{f:depolresult}, where the threshold value is given as a memory lifetime instead of the typical depolarizing channel strength, for more accurate comparison with later results. A critical memory lifetime (defined as the time below which increasing the lattice size reduces the error rate) can be found at around $\frac{t_c}{\beta}\approx 0.04$. Noise and error correction for the honeycomb model in the ultra-high temperature regime are extremely similar to the standard toric code, since the existence of the honeycomb Hamiltonian plays no role except to define the initial encoded states. The only subtlety is that the depolarizing channel applied to the honeycomb model does not correspond to the depolarizing channel applied to the toric code, as toric code qubit operators are encoded in two qubits of the honeycomb model. As such, since we sample single-site Pauli operators independently, we effectively consider a noise source for the toric code that biases towards phase-flip errors over bit-flip errors.
		
		On small timescales (equivalently small error rates), the simulation time is approximately related to the depolarizing channel strength as $p_{\mathrm{dep}}\approx \frac{3}{\beta}t$ (since there are three possible Pauli error events per qubit, each with rate $\frac{1}{\beta}$). This allows us to compare our memory lifetime to the standard toric code under depolarizing noise using a PMA decoder, which has a threshold of around $p_{\mathrm{dep}}^*\approx 0.158$~\cite{Dennis2002}, corresponding to a memory time of $\frac{t_c}{\beta}\approx 0.05$. Note also in this comparison that since each of the physical qubits of the toric code corresponds to two physical qubits in the honeycomb model, the effective error rate per encoded toric code qubit is approximately double that of the honeycomb qubits.

	\subsubsection{Low temperatures}\label{s:lowtempresult}

		For low temperatures, $\beta\gtrsim 10^5$, we find that the appearance of broken dimers over the memory lifetime of the system (for the lattice sizes we consider) is so infrequent that they were never observed. The error rate as a function of time is shown for various lattice sizes in \Fref{f:lowresult}.
	
	\begin{figure}
		\centering
		\subfloat[$\beta=60000$]{\includegraphics[width=0.4\textwidth]{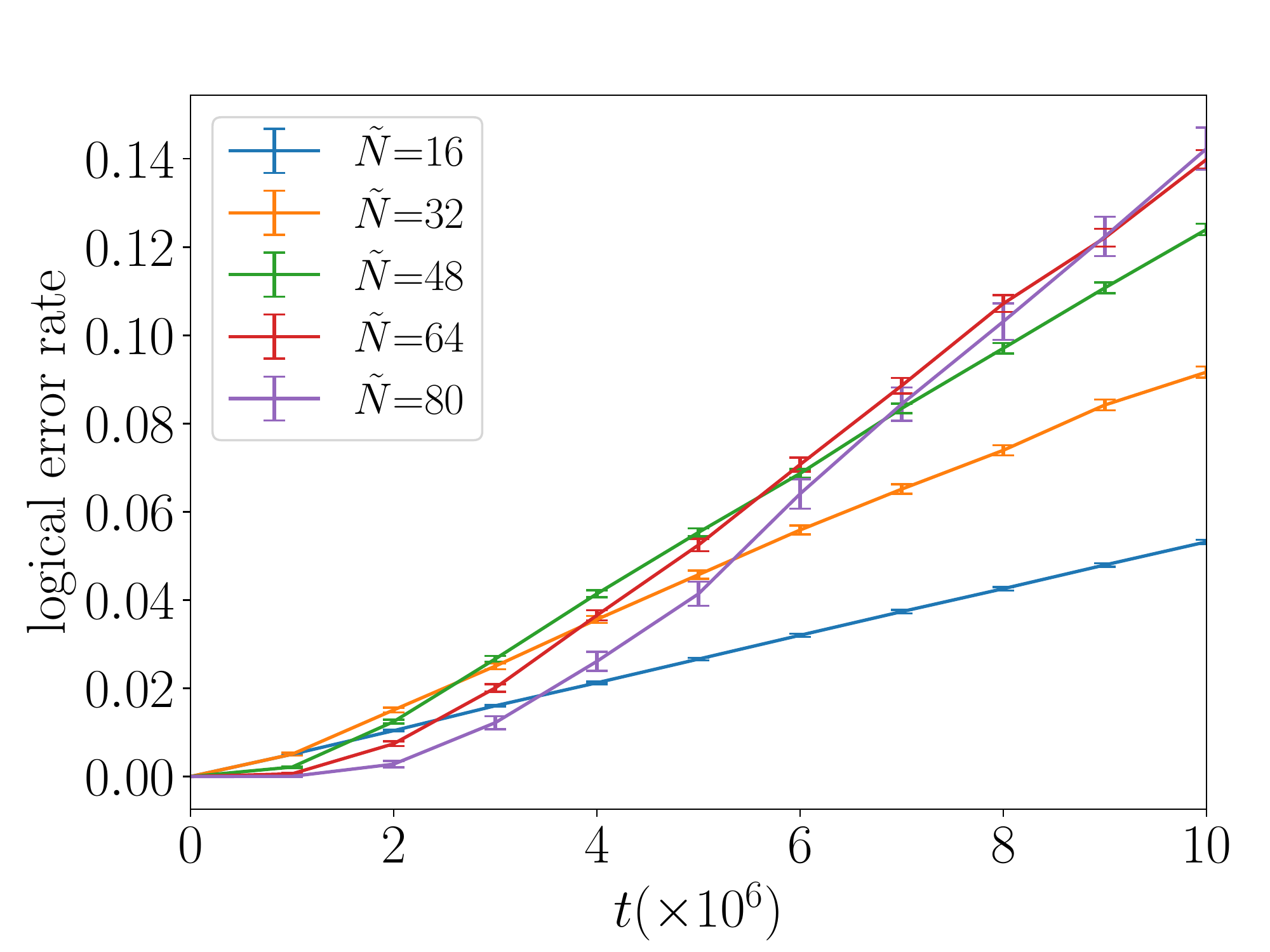}}
		\hspace{0.25cm}	
		\subfloat[Pseudo-threshold for $\beta=60000$]{\includegraphics[width=0.4\textwidth]{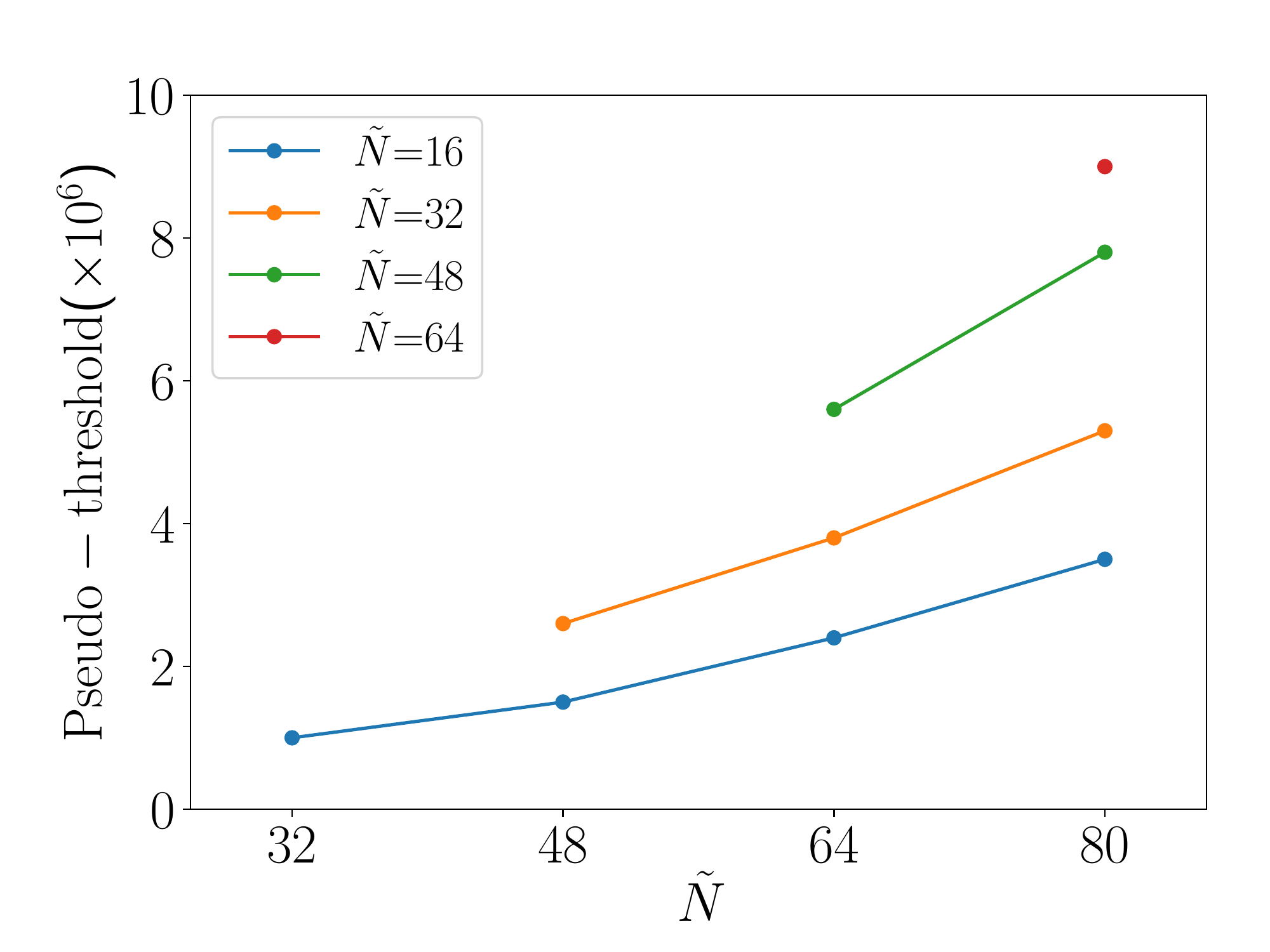}}
		\caption{(a) Logical error rate against time and (b) pseudo-threshold for various system sizes in low temperature region. The pseudo-thresholds are given by the crossing points of the two corresponding curves on the logical error rate plot. Each point belongs to a series that determines one lattice size, and the value on the $x$-axis specifies a second lattice size. The value of the pseudo-threshold for this point is then the time at the intersection of the logical error rate curves for these two lattice sizes. These data do not display threshold behavior, and in fact show a clear trend of pseudo-threshold increasing with system size. This is a finite size effect, where the diffusion of a single pair of vortices is the dominant logical failure mode.}
		\label{f:lowresult}
	\end{figure}

		We do not find evidence of threshold behavior in this regime. In fact, it appears that the pseudo-threshold, representing the time below which increasing the size of the system reduces the logical error rate, is itself increasing with system size. This can be understood as a finite-size effect in this low temperature regime. It is known that at low temperatures and on finite size lattices, there is a regime where the random walks of a single pair of toric code excitations (vortices in the honeycomb model) are the dominant source of error~\cite{Freeman2014, Brown2016}. This effect leads to increasing suppression of the logical error rate as the lattice size decreases, as seen in \Fref{f:lowresult}.
		
		At such low temperatures, since it is unlikely that any error would lead to a broken dimer, and since we assume a local coupling of each spin to the environment, this effectively restricts the errors considered in the honeycomb model to Pauli $Z$ errors (all others occur with vanishing probability). It is easy to see that these errors can only move vortices horizontally along 1-dimensional rows, as opposed to the typical situation for the toric code where the excitations may be moved in any direction. As such, we can intuit the dynamics of excitations by considering them to be undergoing a 1-dimensional random walk, instead of the 2-dimensional walks seen in Refs.~\cite{Freeman2014, Brown2016}.
		
		As in the conventional toric code, this low temperature regime is expected to persist while the number of anyons is low, i.e.~$N^2\e^{-\beta\delta E}\lesssim 1$ for $\delta E$ the gap, giving a critical $\beta^*\sim \frac{\ln N}{\delta E}$. Within this regime, the memory lifetime has been argued phenomenologically to be inversely proportional to the probability of a topologically non-trivial walk, $\Pi_{\mathrm{nt}}$~\cite{Freeman2014}. In the 2-dimensional case, numerical evidence~\cite{Brown2016} supports this scaling with $\Pi_{\mathrm{nt}}^{2D}\sim \frac{1}{\ln N}$, while for our 1-dimensional walk behavior, we expect $\Pi_{\mathrm{nt}}^{1D}\sim \frac{1}{N}$~\cite{Freeman2014}. This would increase the memory lifetime by a factor of $\frac{N}{\ln N}$ in this regime compared to the behavior of the standard toric code, strengthening the effect of this error-suppression mechanism on larger lattices. This interesting finite-size effect in the perturbative regime of the honeycomb model may warrant further study in order to determine under what circumstances it could be exploited to produce improved memory lifetimes or logical error rates. If this were desirable, the decoding algorithm could presumably also be optimized to take this extra structure into account.

	This effective 1-dimensional walk behavior is an artefact of our having restricted to independent errors on each qubit. Typically this is justified by the expectation that physical errors will act quasi-locally, and that there will be little qualitative distinction between strictly local and quasi-local error models. However, it is clear that in the present case, if we were to allow suitable 2-local errors, we could engineer behavior equivalent to the standard toric code under depolarizing noise. Thus a more careful consideration of the physically relevant noise model for a given implementation must be made when considering error correction in the low-temperature regime of the honeycomb model.
	
	\subsubsection{Intermediate temperatures}\label{s:intermediatetempresult}

	\begin{figure}
		\centering
		\subfloat[$\beta=15000$]{\includegraphics[width=0.4\textwidth]{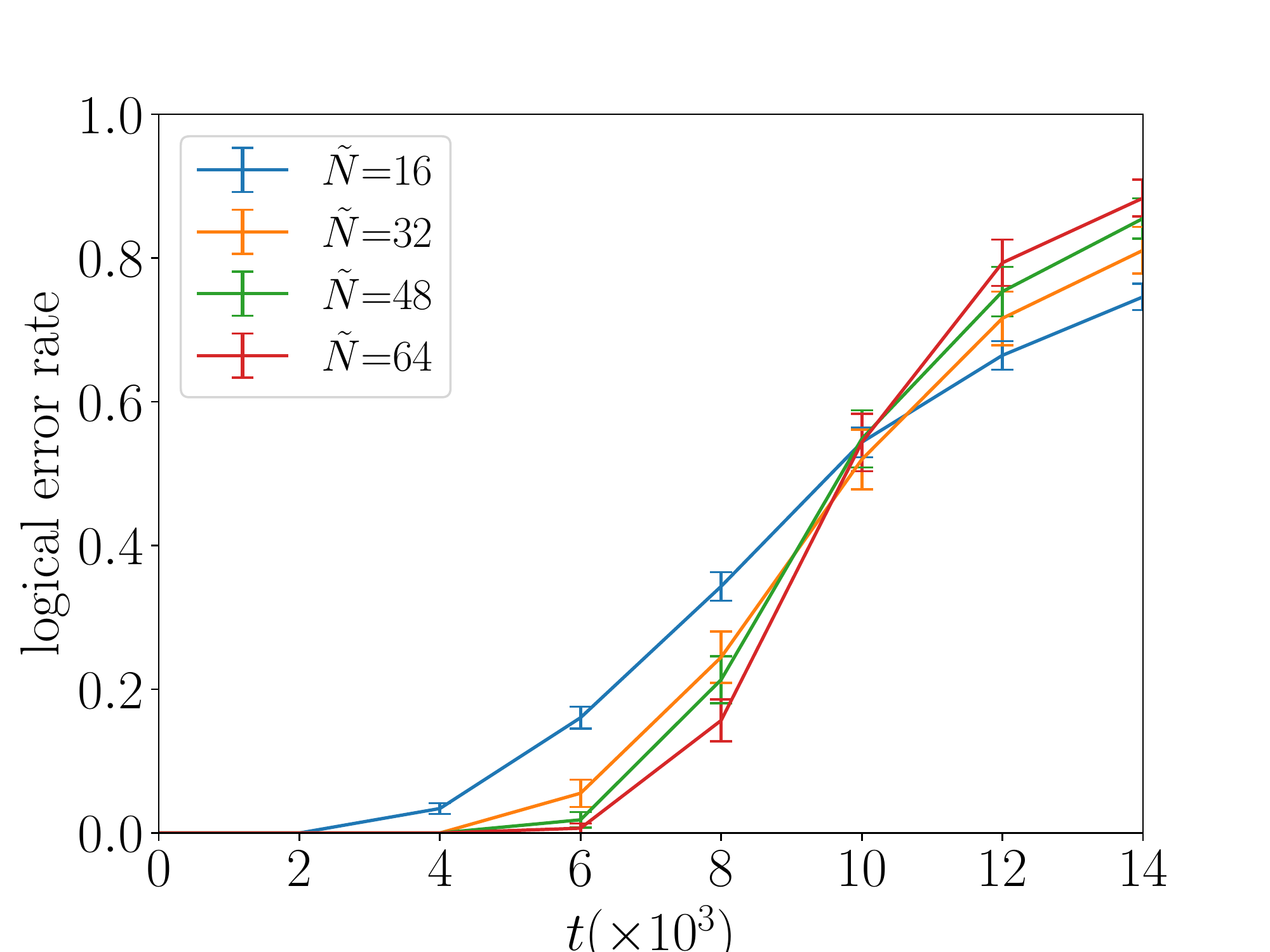}}\hspace{0.25cm}	
		\subfloat[$\beta=25000$]{\includegraphics[width=0.4\textwidth]{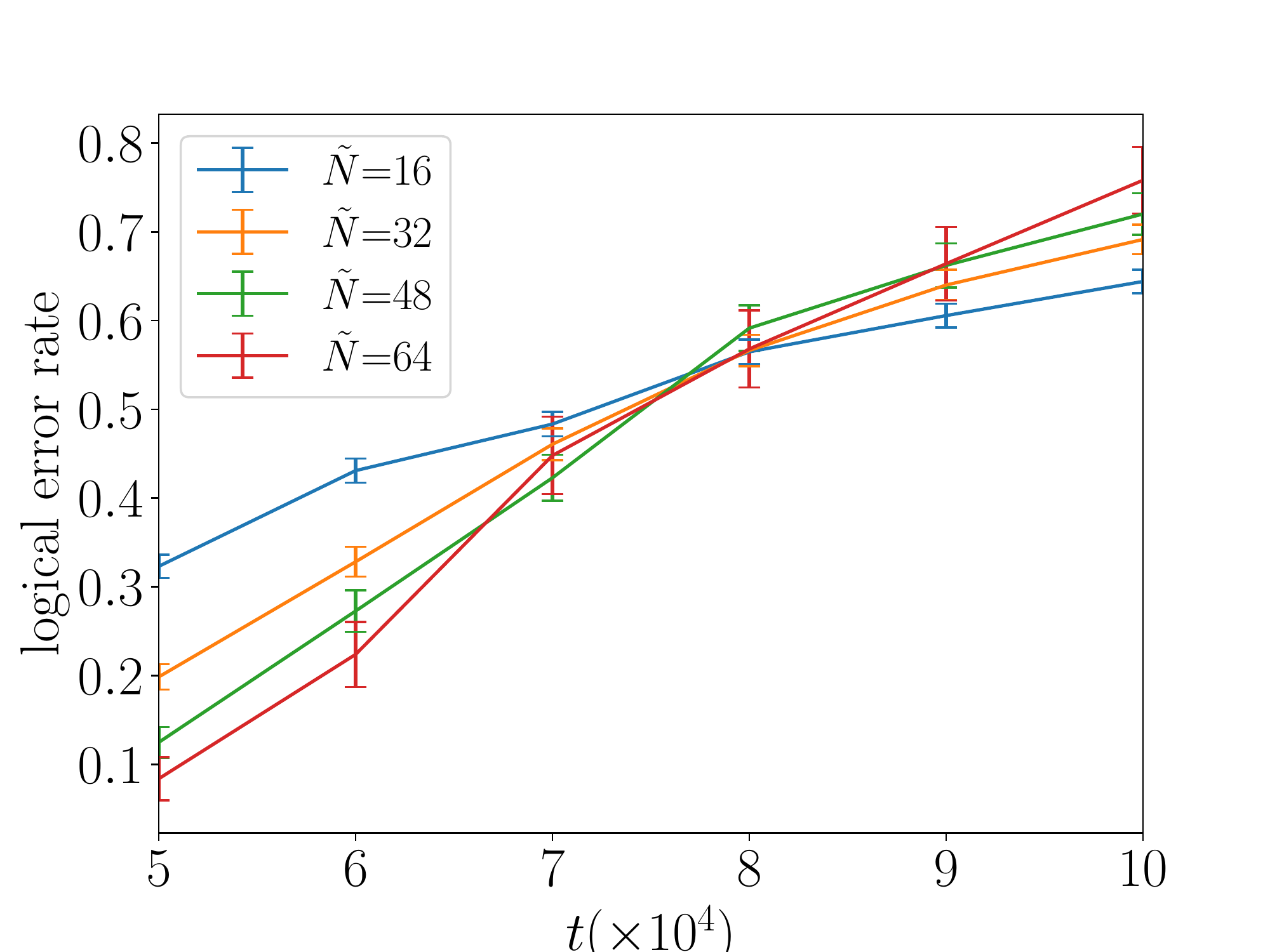}}\\
		\subfloat[$\beta=35000$]{\includegraphics[width=0.4\textwidth]{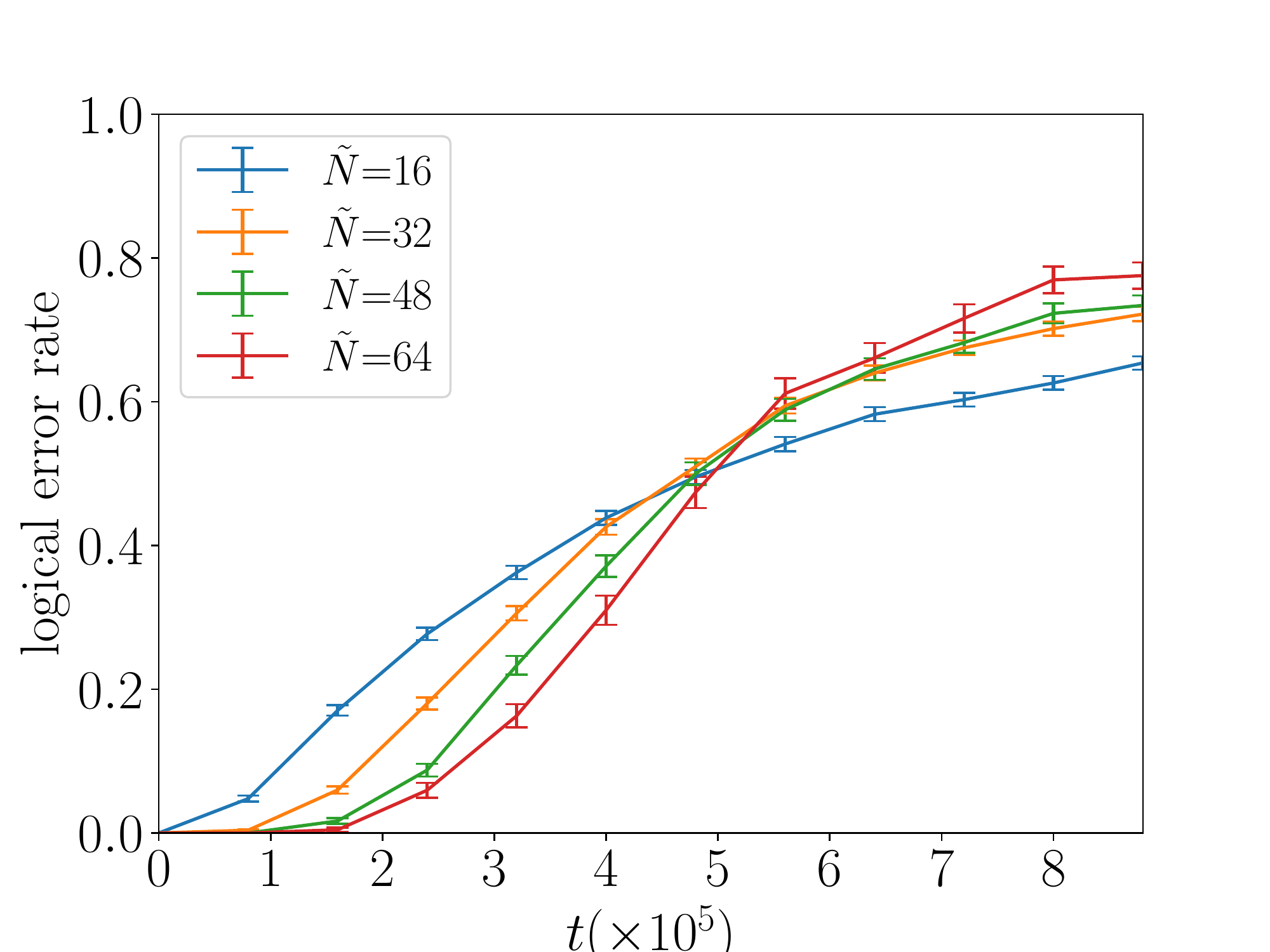}}\hspace{0.25cm}	
		\subfloat[$\beta=40000$]{\includegraphics[width=0.4\textwidth]{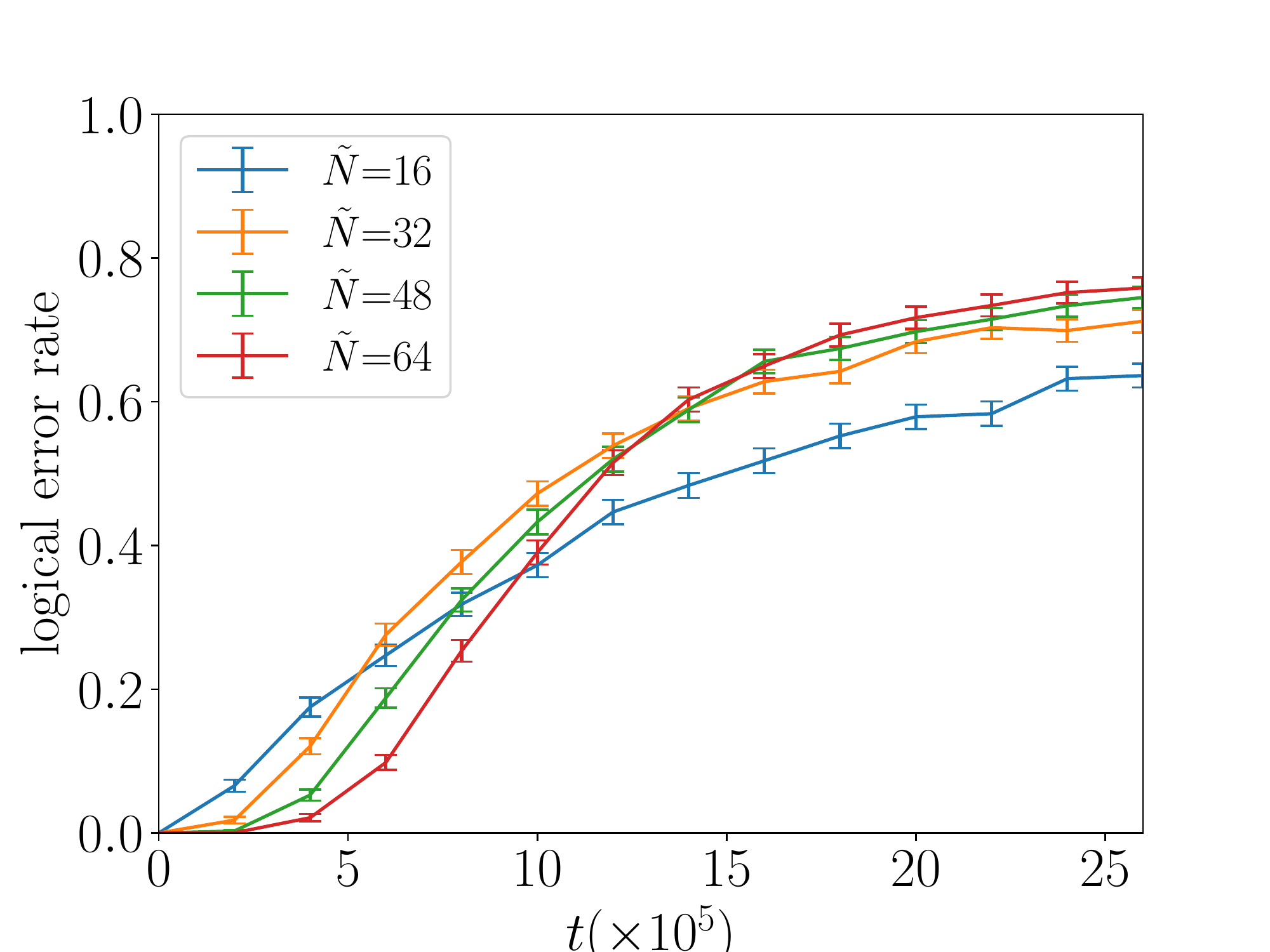}\label{f:intermediateresult_6thorder_d}}
		\caption{Logical error rate against memory lifetime for various system sizes and intermediate temperatures. We see clear evidence of threshold behavior in each plot, though at $\beta=40000$, the onset of the low-temperature finite-size regime can be seen for the smallest lattice size, $N=16$.}
		\label{f:intermediateresult_6thorder}
	\end{figure}

	At intermediate temperatures ($\beta\sim 10^4$), we are still in a regime where broken dimers are sufficiently unlikely that they are unobservable for the system sizes and timescales we simulate. However, it differs from the low-temperature regime in that the failure modes are no longer dominated by randomly walking anyons. Rather, they are dominated by the more typical situation where clusters of errors percolate beyond the point at which they can be reliably distinguished.

	The results of our simulations are shown in \Fref{f:intermediateresult_6thorder}. The memory lifetimes (or equivalently error thresholds) are estimated as the common crossing point of the curves for different lattice sizes, and plotted in \Fref{f:intermediateresult_6thorder_2}. At lower temperature and small lattice sizes (most notably \Fref{f:intermediateresult_6thorder_d} at $N=16$), we see deviation from this common crossing point, as the finite-size effects seen in the low-temperature regime (\Fref{f:lowresult}) begin to become significant. Nonetheless, at larger lattice sizes the expected threshold behavior is observed. The memory lifetimes are found to be exponential in inverse temperature, as is the case for the standard toric code~\cite{Brown2016}.
	
	\begin{figure}
		\centering
			\includegraphics[width=0.55\textwidth]{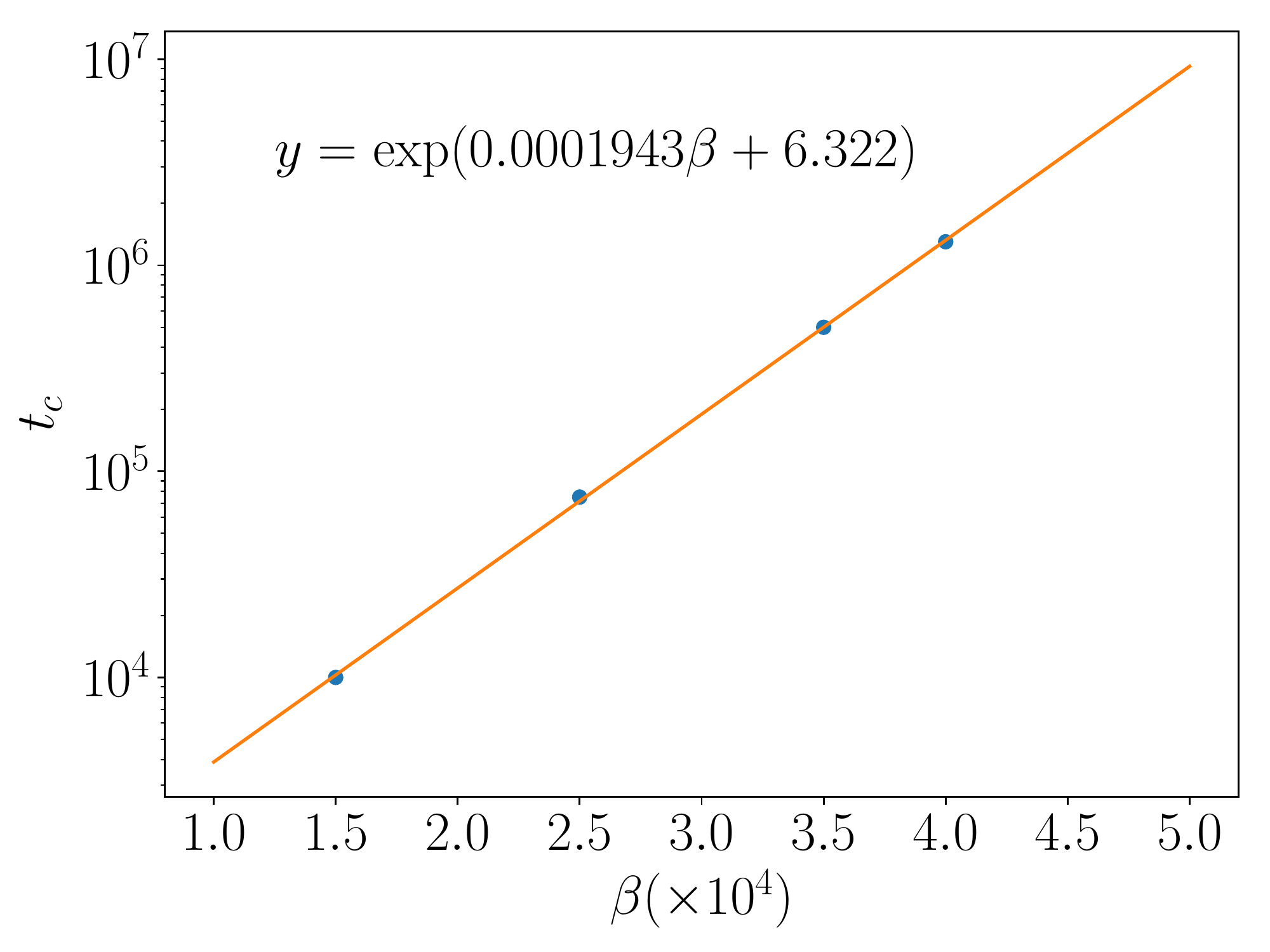}
		\caption{Memory lifetime against temperature in the intermediate temperature region. The lifetimes are observed to increase exponentially in inverse temperature, as is expected for large $\beta$. This is consistent with the analogous behavior in the toric code.}
		\label{f:intermediateresult_6thorder_2}
	\end{figure}

	\subsubsection{High temperatures}

	\begin{figure}
		\centering
			\subfloat[]{\raisebox{1.1ex}{\includegraphics[width=0.33\textwidth]{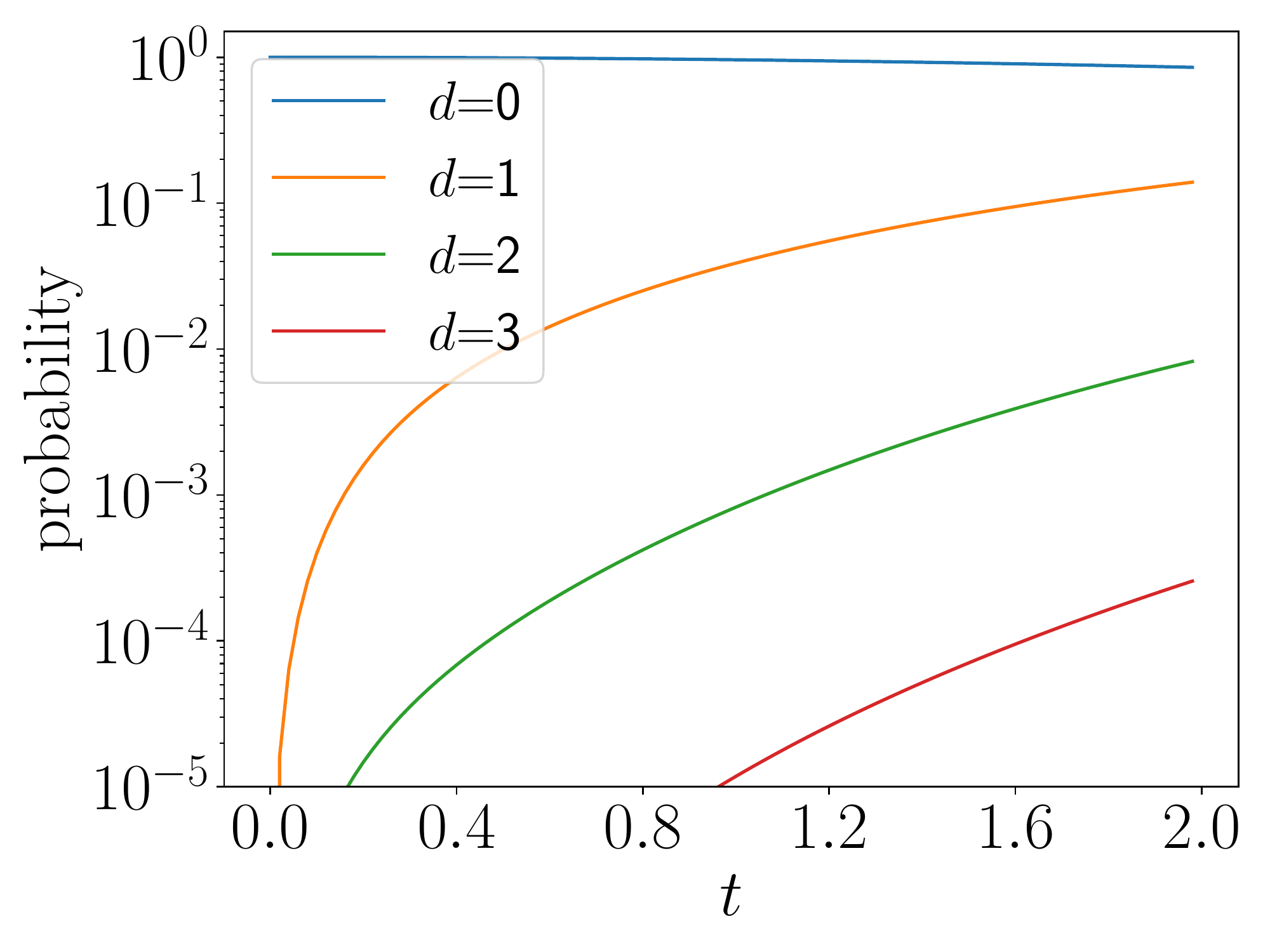}}}
			\subfloat[]{\raisebox{1.1ex}{\includegraphics[width=0.33\textwidth]{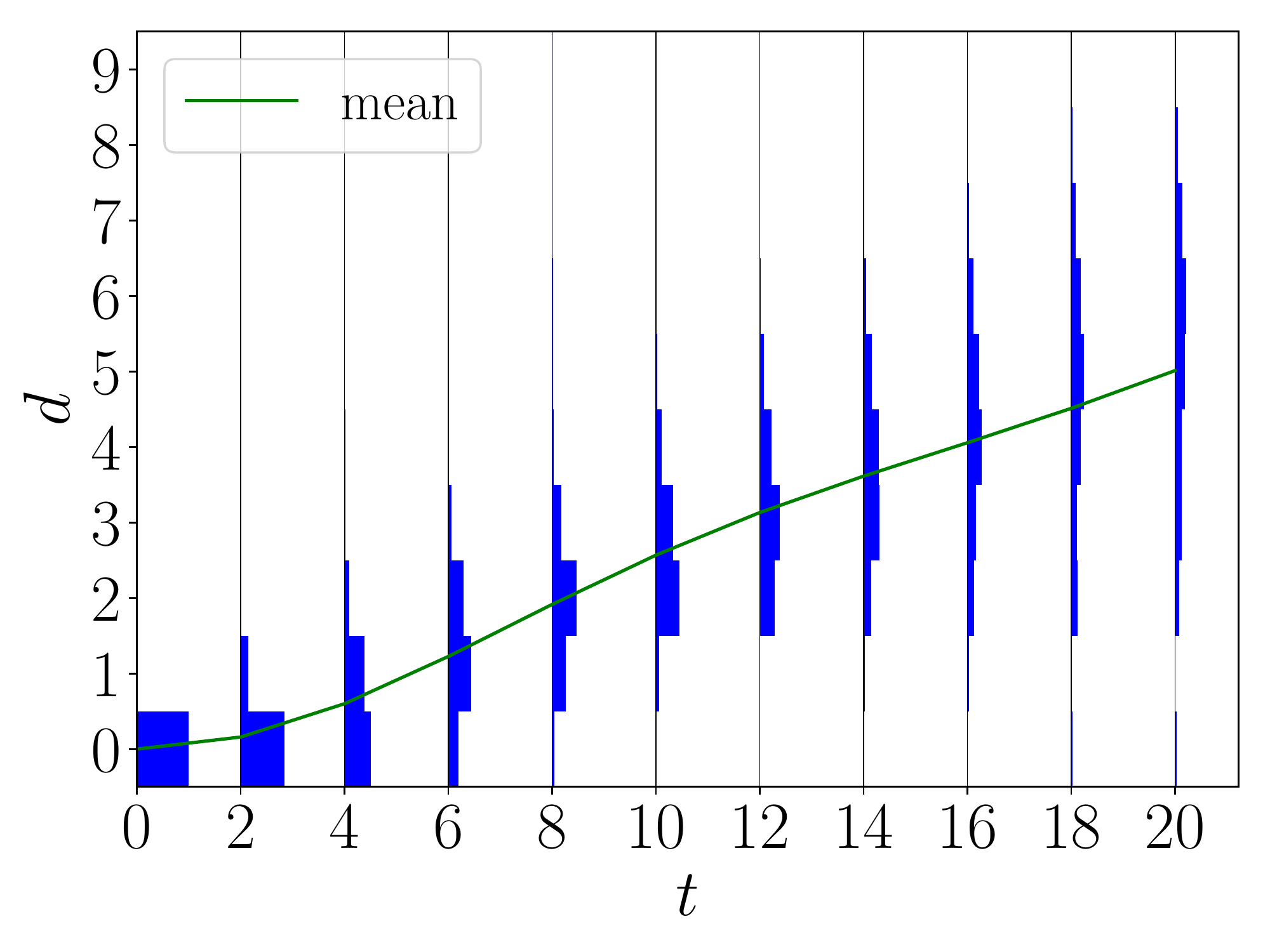}}}
			\subfloat[]{\includegraphics[width=0.32\textwidth]{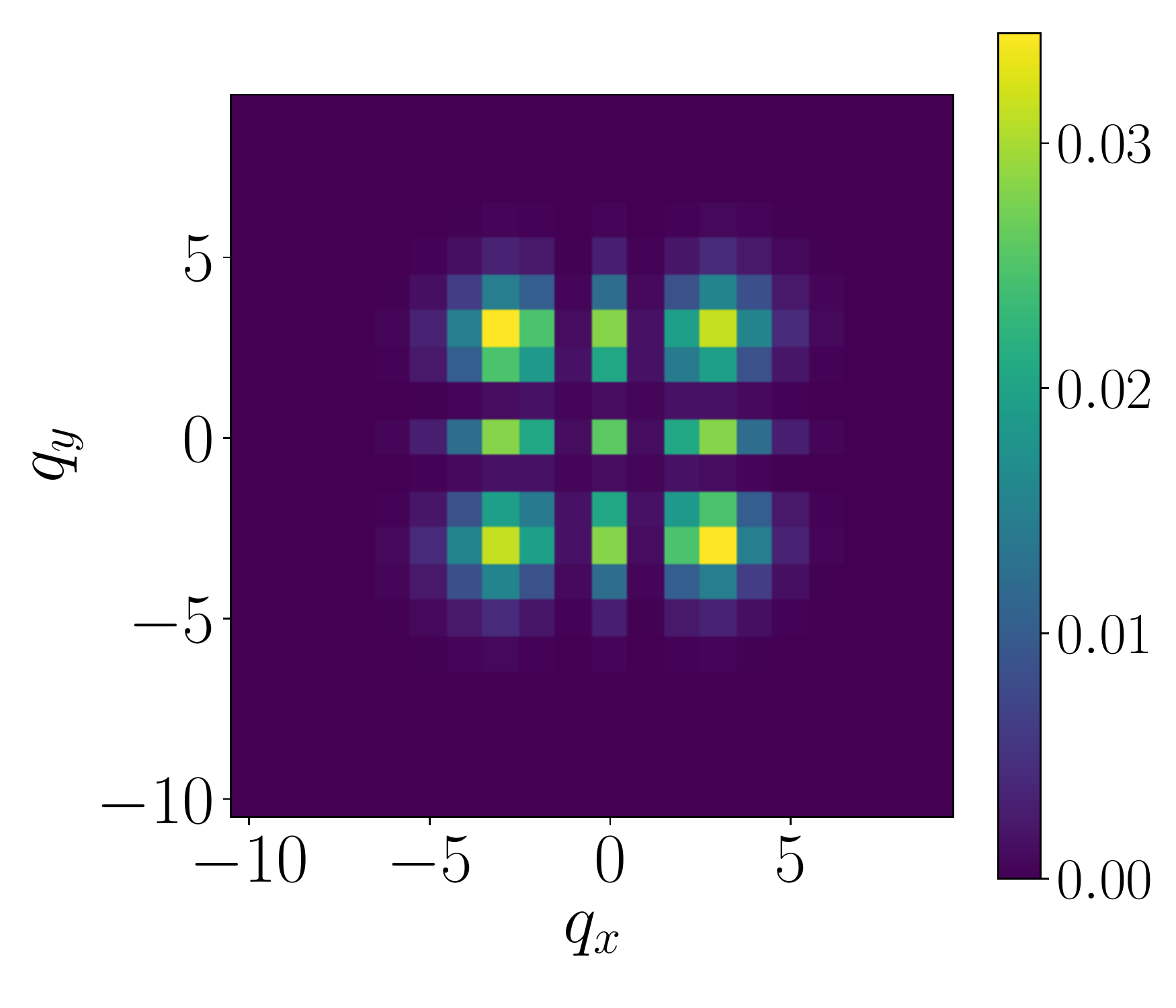}}
		\caption{Diffusion of a single broken dimer under Hamiltonian evolution. Fig.~(a) shows the probabilities of the dimer having moved different distances $d$ (in Manhattan distance) over short timescales. We see that on these timescales, the likelihood of travelling further than one or two sites is negligible. Fig.~(b) considers longer timescales, and plots a histogram of the dimer positions at each timestep, as well as the mean distance travelled. Fig.~c) shows the probability distribution of dimer positions after time $t=20$, corresponding to the final timestep of (b).}
		\label{f:broken_dimer_diffusion}
		\end{figure}
		
	\begin{figure}
		\centering
		\subfloat[$\beta=1$]{\includegraphics[width=0.4\textwidth]{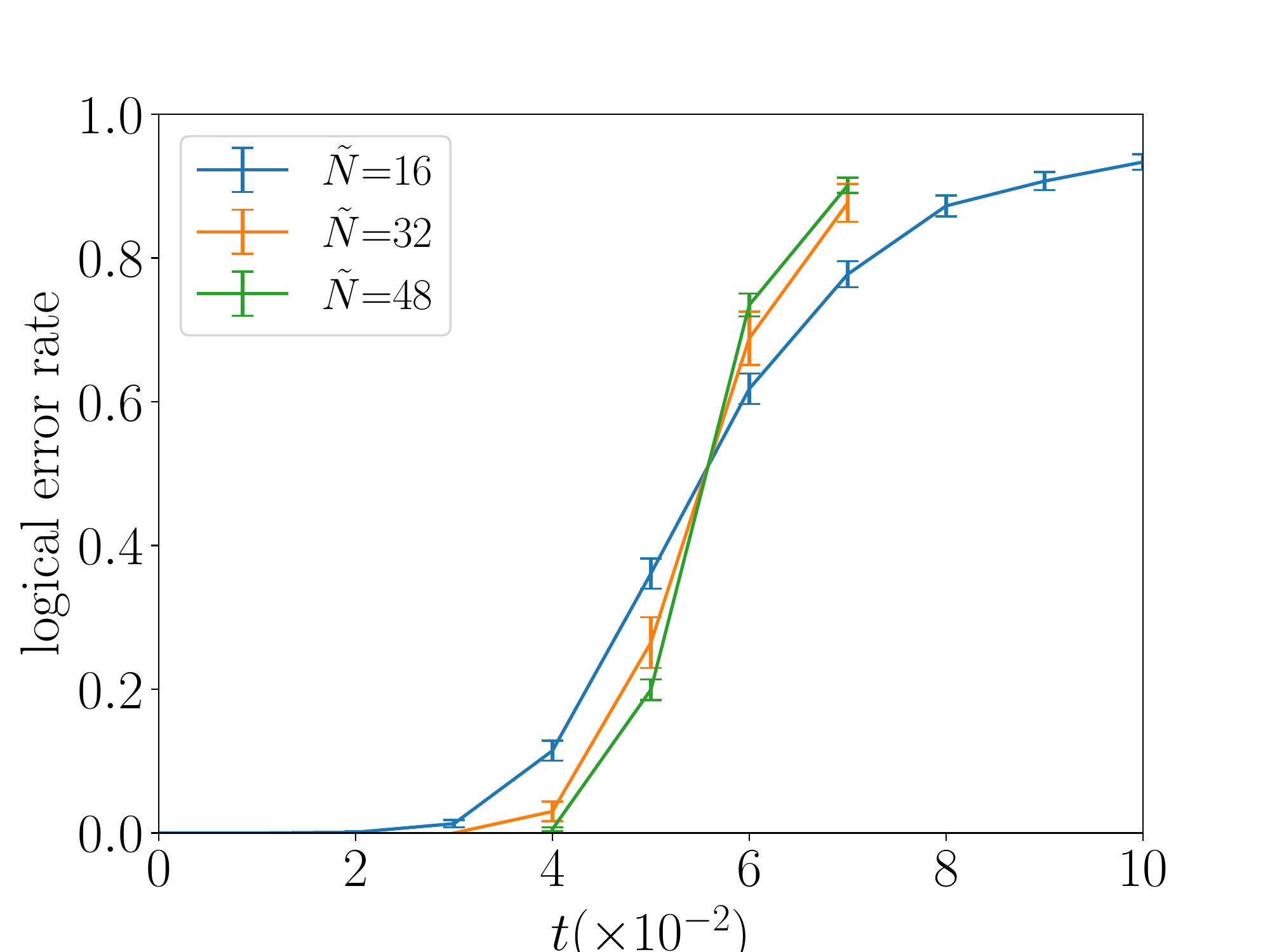}}\hspace{0.25cm}	
		\subfloat[$\beta=7$]{\includegraphics[width=0.4\textwidth]{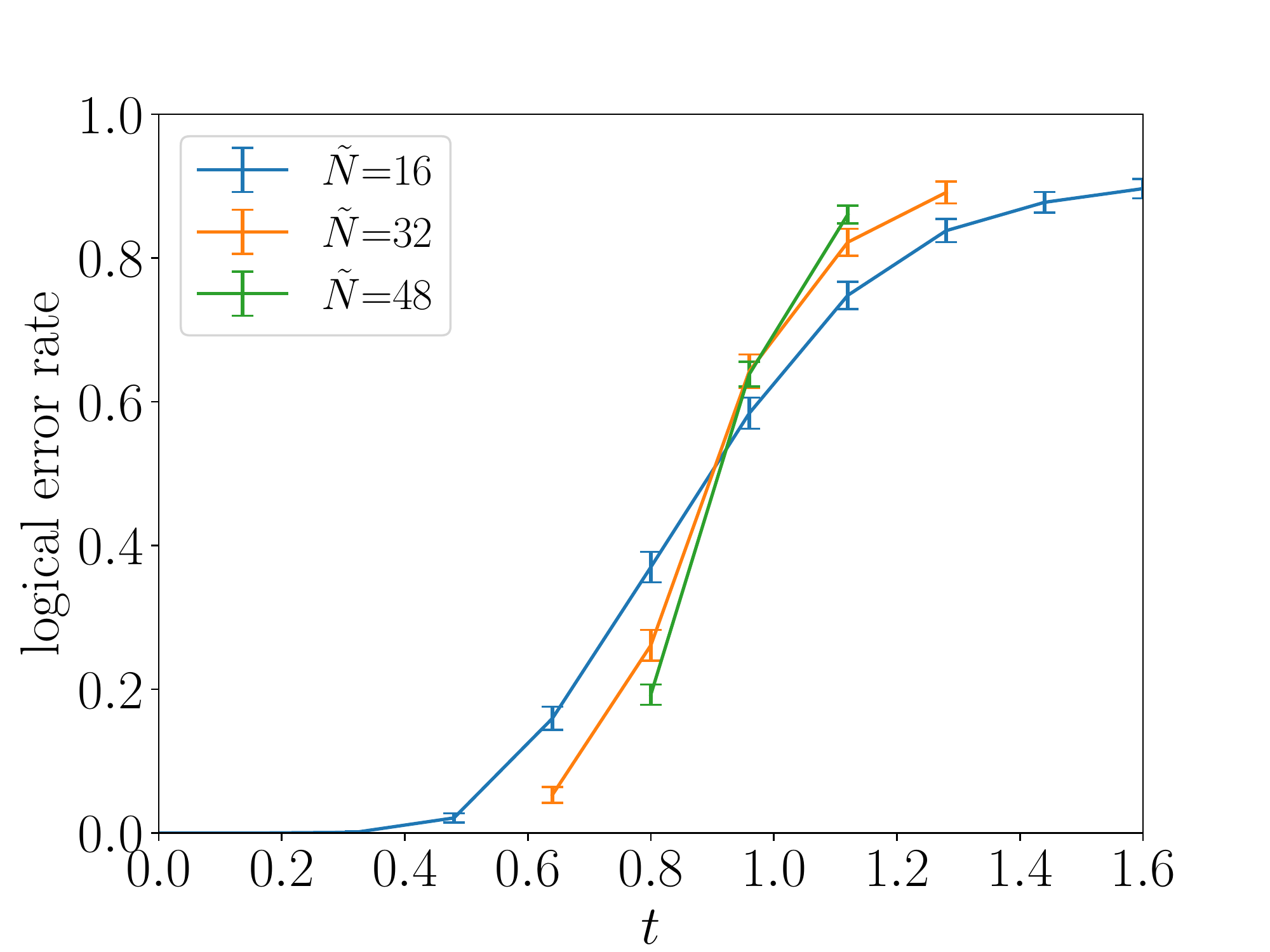}}\\
		\subfloat[$\beta=20$]{\includegraphics[width=0.4\textwidth]{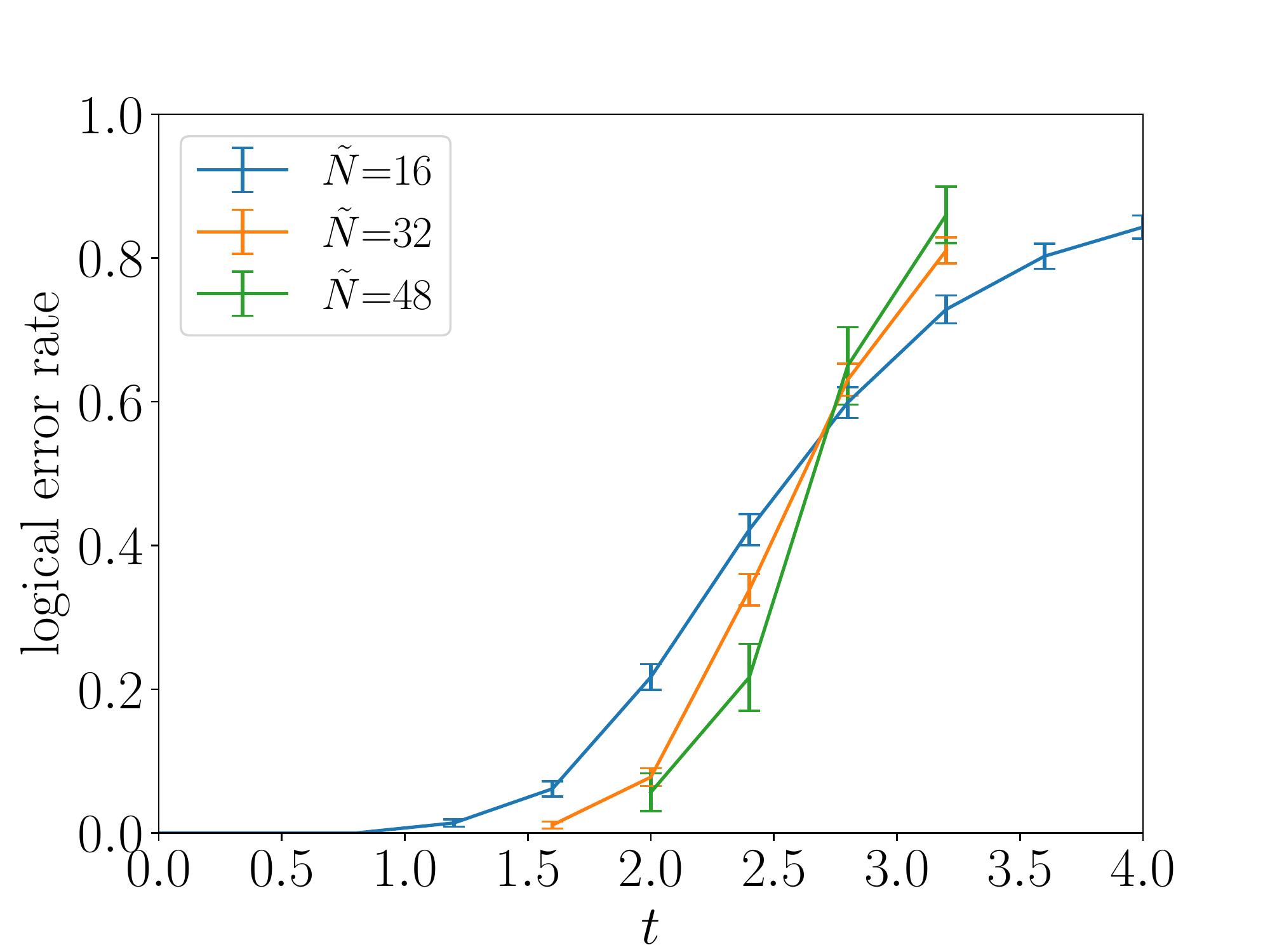}}\hspace{0.25cm}	
		\subfloat[$\beta=40$]{\includegraphics[width=0.4\textwidth]{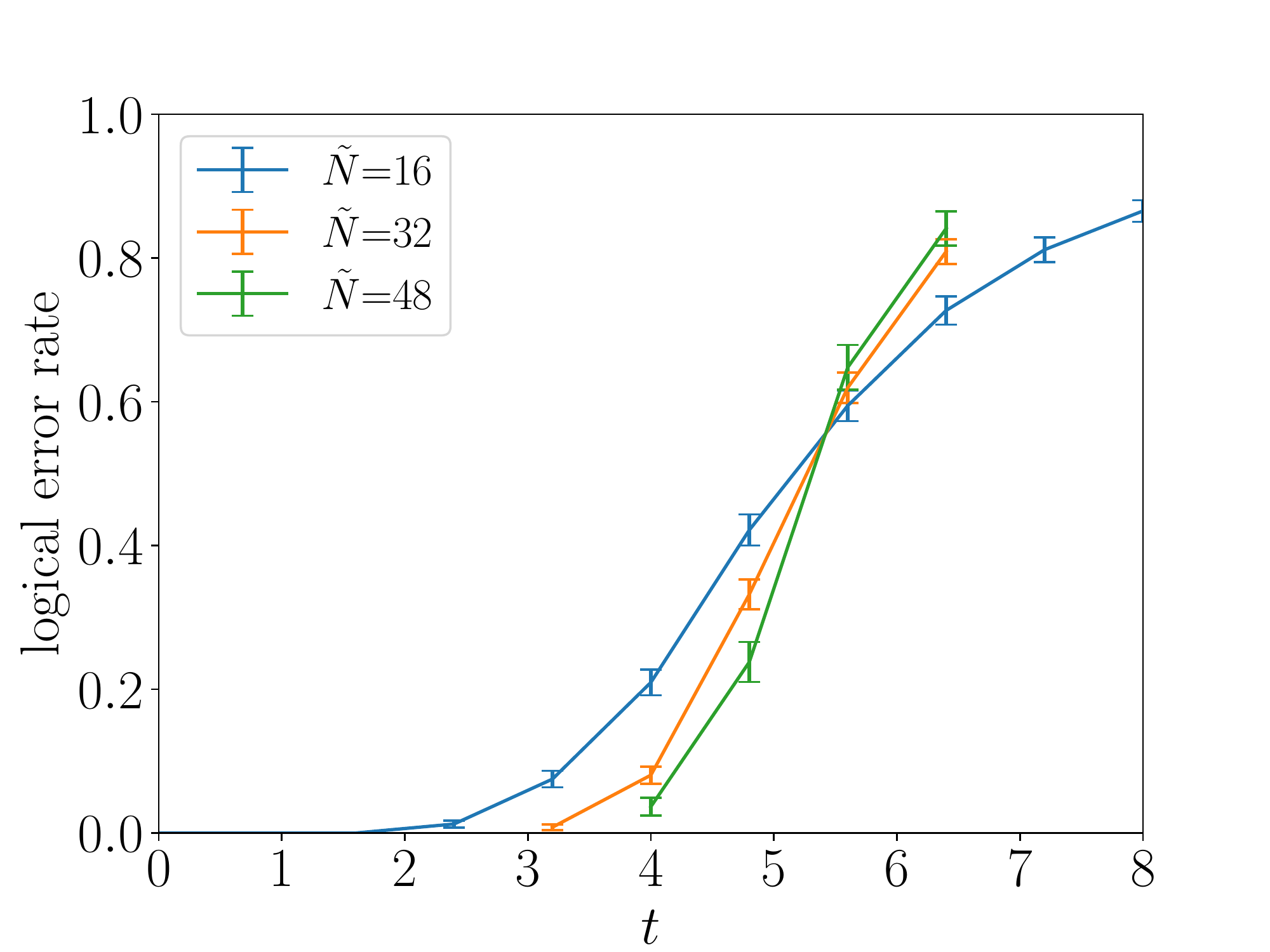}}
		\caption{Logical error rate against memory lifetime for various system sizes and high temperatures. Again, threshold behavior is observed at each of these temperatures, though the computational expense of these simulations precluded the larger lattice sizes considered at lower temperature.}
		\label{f:highresult}
	\end{figure}	

	The low- and intermediate-temperature regimes considered in the previous sections were characterized by the fact that broken dimers were unlikely to appear on timescales shorter than the memory lifetime. This is clearly unrealistic as we approach the thermodynamic limit, so in order to characterize the impact of leaving this regime, we consider two sets of effects---coherent diffusion of broken dimers, and a heuristic simulation of thermalization processes.

		\paragraph{Dimer diffusion}
	We consider a state with a single broken dimer ($-1$ eigenstate of $K^z_{\mbf{q}}$) and no vortices. We then simulate the time-evolution of this state under the effective Hamiltonian on dimer degrees of freedom (taken to $2\nd$ order), before calculating the distribution of broken dimer positions, and hence the average distance travelled. This will allow us to determine whether the diffusion of broken dimers is significant on timescales relevant for error-correction. The results are shown in \Fref{f:broken_dimer_diffusion}.
		
		\paragraph{Mean-energy analysis}
			Finally, we perform a heuristic simulation of thermalizing dynamics at high temperatures, in which we approximate the quantum jumps of the monte-carlo method described in \Sref{s:incoherentmodel} by single-qubit Pauli errors with transition probabilities determined by mean energy differences in the effective Hamiltonian. This approach is expected to be accurate in the perturbative regime for small lattices at short timescales. The results are shown in \Fref{f:highresult}.
	
		As in the intermediate temperature regime, we see clear evidence of threshold behavior in the memory lifetime, though the computational cost of these simulations is higher and so we are restricted to probing slightly smaller lattices. The memory lifetimes are plotted in \Fref{f:highresult_2}, where we see behavior consistent with linear lifetime in inverse temperature, in contrast to the intermediate temperature regime. This is a result of the small $\beta$ limit of \Eref{E:gammarate} giving the error rate for vortex-creation errors as $\gamma\sim\frac{1}{\beta}$. Note however that the dimer-breaking error rates are not yet in this high-temperature linear-scaling regime. This leads to a contrast between the lifetimes in the ultra-high temperature regime (\Fref{f:depolresult}) where all errors are in the linear-scaling regime, and $\frac{t_c}{\beta}\approx 0.04$, and the temperatures considered in \Fref{f:highresult}, where only Pauli $Z$ errors are in the linear-scaling regime, and correspondingly the lifetimes are increased to $\frac{t_c}{\beta}\approx 0.14$.
	
	\begin{figure}
		\centering
			\includegraphics[width=0.55\textwidth]{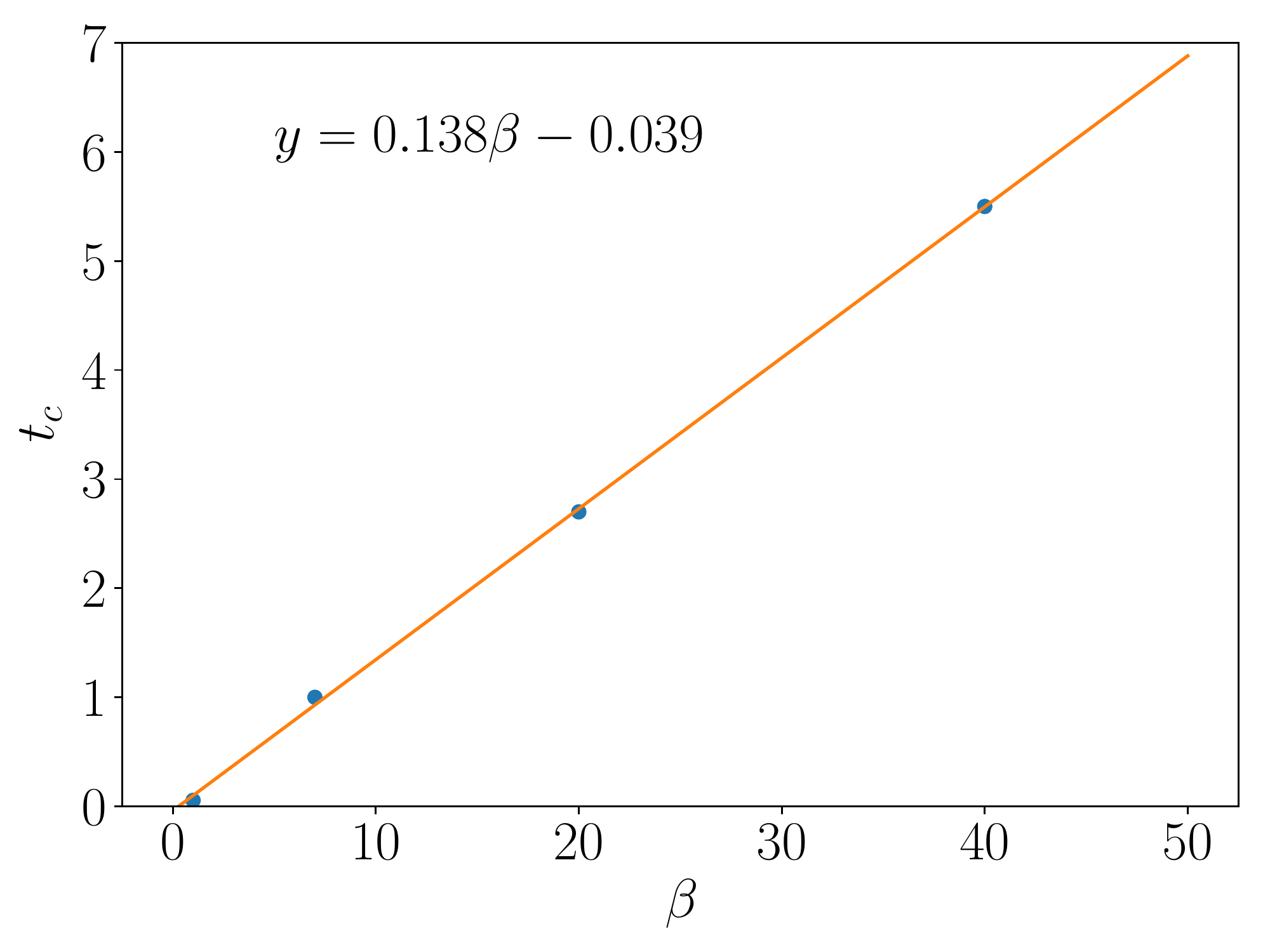}
		\caption{Memory lifetime against temperature at high temperatures. We observe a linear scaling, as expected for low $\beta$ and as found in the ultra-high temperature simulations of \Fref{f:depolresult}, but in contrast to the high $\beta$ exponential scaling of the intermediate temperature regime seen in \Fref{f:intermediateresult_6thorder}.}
		\label{f:highresult_2}
	\end{figure}
	
	We also plot the observed density of broken dimers over relevant timescales in \Fref{f:brokendimerdensity}. We find dimer density equilibrating exponentially fast at high temperatures, as expected. The density of dimers is always below $10\%$ even in our highest temperature simulations, and for $\beta\geq 20$, we observe no broken dimers for lattice sizes and timescales studied.
	
	It should also be noted that the diffusion behavior studied in \Fref{f:broken_dimer_diffusion} is unlikely to significantly affect the error correction properties of this model, since at the higher temperatures where the dimer density is non-negligible (e.g.~\Fref{f:densitybeta1}), the memory lifetime is so short that the diffusion distance is negligible, and at lower temperatures where the memory lifetime is long enough to give average diffusion distance on the order of $0.1$ lattice spacings (e.g.~\Fref{f:densitybeta2}), the density of broken dimers is so low that their diffusion is unlikely to cause them to interact or play a major role in logical failures.
	
	\begin{figure}
		\centering
		\subfloat[$\beta=1$]{\includegraphics[width=0.4\textwidth]{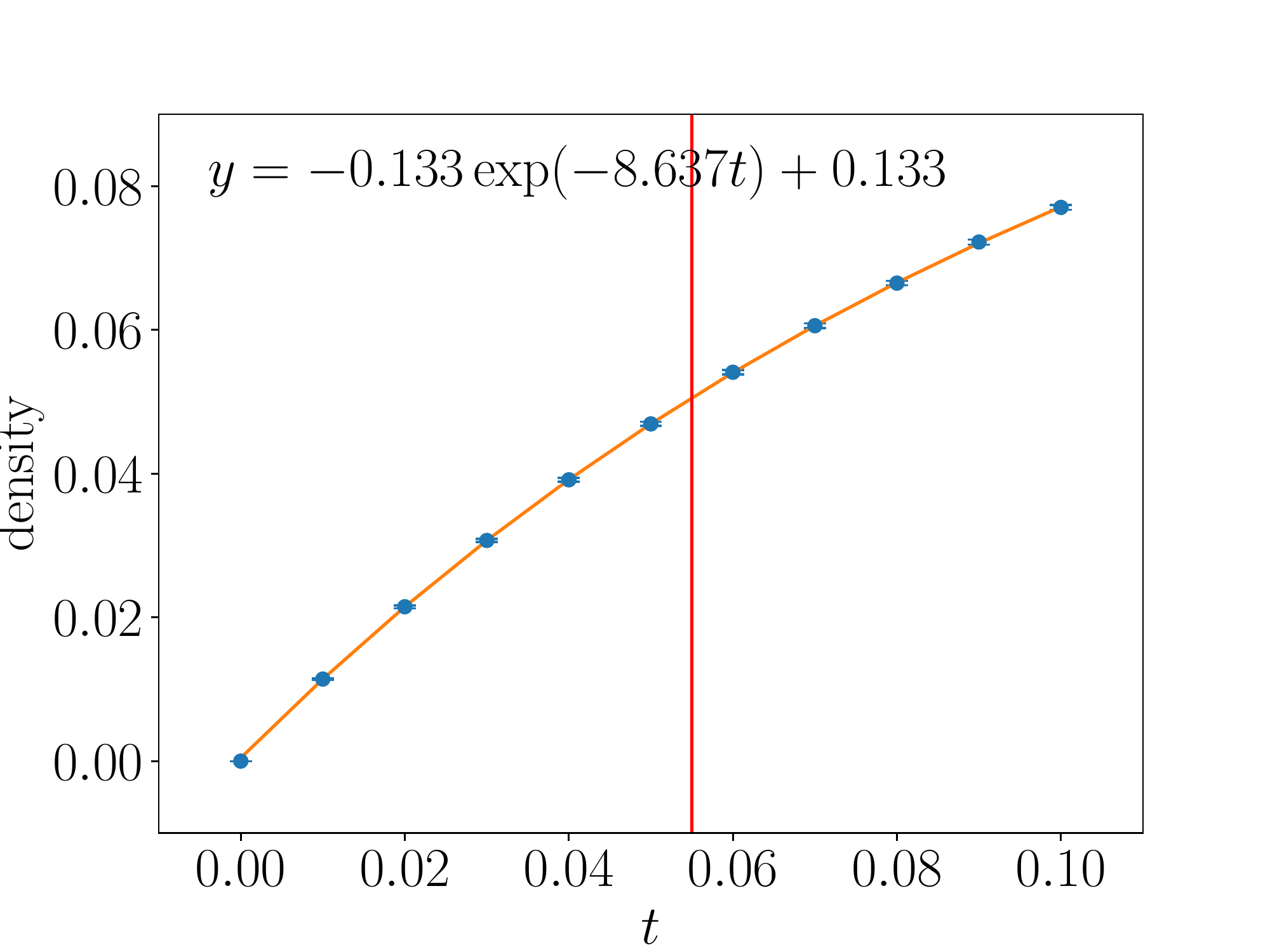}\label{f:densitybeta1}}\hspace{0.25cm}	
		\subfloat[$\beta=7$]{\includegraphics[width=0.4\textwidth]{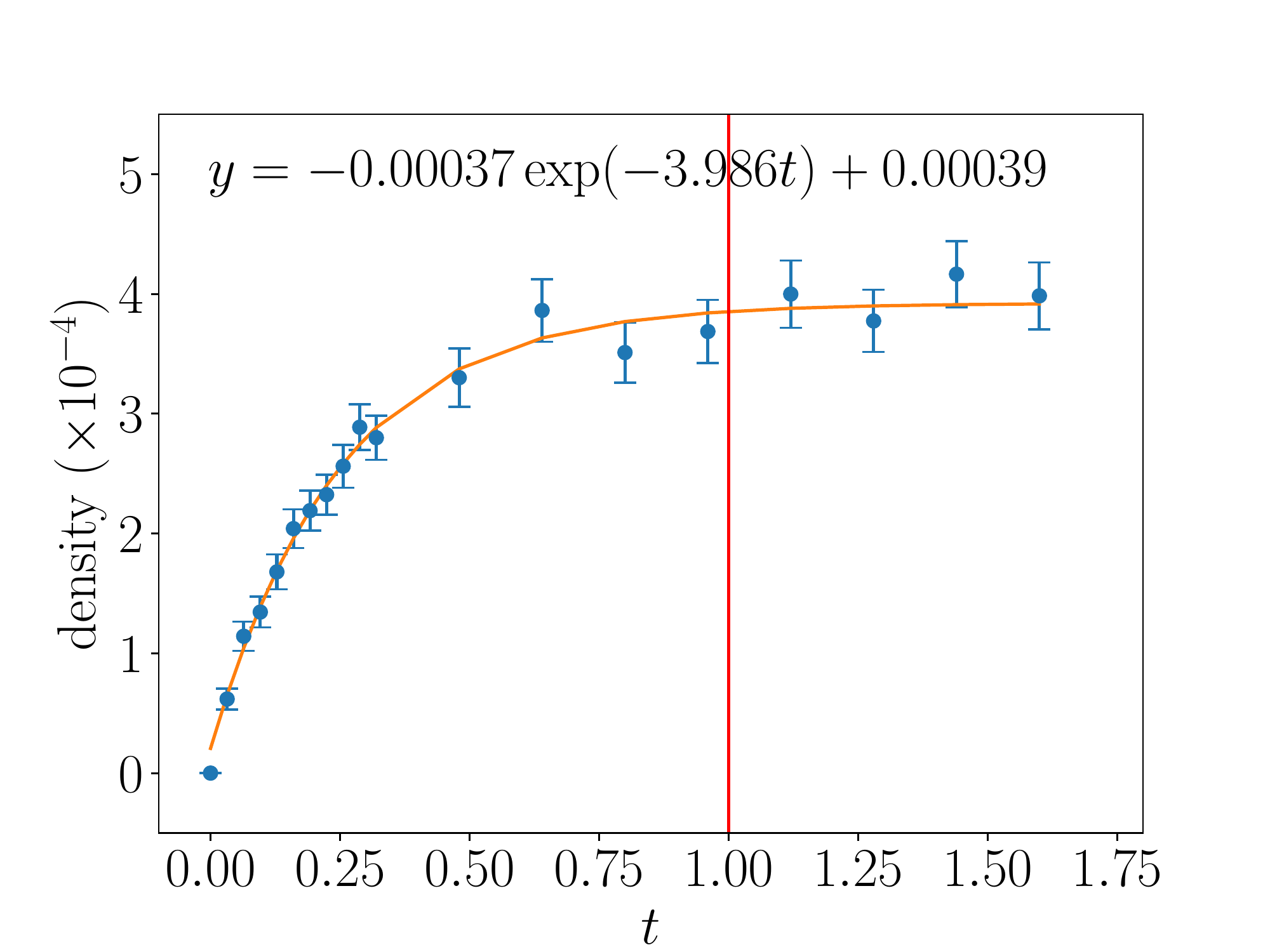}\label{f:densitybeta2}}
		\caption{Density of broken dimers over time. The memory lifetimes observed in \Fref{f:highresult} are marked in red for comparison. For $\beta\geq 20$ we observe no broken dimers in our simulations. At the higher temperatures seen in (a), the memory lifetime is so short that the dimer density does not equilibrate over the simulation time. In contrast, at the lower temperature seen in (b), the memory lifetime is long enough for the dimer density to equilibrate, but at these temperatures the equilibrium broken density is extremely low. In both cases, the density on timescales relevant for error correction is small enough that we can safely neglect the coherent diffusion considered in \Fref{f:broken_dimer_diffusion}.}
		\label{f:brokendimerdensity}
	\end{figure}
	
	To summarize, we do not find evidence that the appearance of broken dimers at high temperature has any qualitative effect on the honeycomb code on relevant timescales, and we observe behavior similar to that of the standard toric code.

\section{Discussion}\label{s:discussion}
	
Our analysis and numerical simulations highlight behaviors in the honeycomb model that differ in several ways from error correction in the standard toric code. Despite this, our results suggest that these disparities will not qualitatively affect the properties of the honeycomb model as a quantum memory.

Of course, as shown broadly in \Sref{s:asymptotics}, the key error correction properties of degeneracy and local indistinguishability hold only approximately for the honeycomb model. Analytically, we have provided an explicit exponential bound on the energy splitting of the codespace, which is important for bounding the dephasing time of the memory, as well as on its local indistinguishability and approximate correctability, which are crucial for determining the parameters of the code. By taking advantage of the special structure of the honeycomb model, these bounds are sharper that what appears to be possible using the known toolkit for general topological phases, and have the added bonus of explicit constants. 

Numerically, we have found an exponential lifetime in inverse temperature, as holds for the toric code. 
We also found evidence that the appearance of non-topological broken dimer excitations has little qualitative effect on the error correction process. 
In regimes where the memory lifetime would be long enough for dimer diffusion to be significant, the temperature is so low that dimers are exceedingly unlikely to arise. While this is quite fortuitous, it is unclear to what extent this particular phenomenon would extend out of the perturbative regime, where our simulation methods are unable to probe.
In any case, we expect that error correction in general will be possible everywhere in the gapped phase, though of course the memory lifetime is expected to decrease as the gap decreases. How to reasonably implement the detection and manipulation of anyons has been considered in various experimental contexts, most thoroughly for optical lattices~\cite{Zhang2007, Pachos2007, Jiang2008, Aguado2008creation, Dusuel2008}.

As in the standard toric code, we identified several distinct regimes of behavior, the most interesting of which is the low-temperature regime, where finite size effects dominate the error modes. 
While this kind of regime exists in the standard toric code, we argued that the structure of the honeycomb model leads to an improved lifetime scaling by a polynomial factor under local thermal noise. 
Again, while fortuitous, this effect relies on the perturbative structure of the honeycomb mode, and so it is unclear how far outside the perturbative regime it can be expected to persist. This may warrant further study to determine whether this effect can be exploited in near-term experimental implementations.

While these departures from the behavior of the toric code are important to understand when designing good error correction protocols, none of these novel behaviors are particularly shocking, nor would they seem to be particular obstacles to successful error correction in the honeycomb model.

Studies of approximate topological quantum error correcting codes and their behavior are in their infancy, as are designing tools to probe these systems, and developing protocols to ideally compensate for the deviations from exact codes. 
There are several issues that must be considered when analyzing such systems, one of which is the motivating experimental setup.
		
	In typical studies of topological error correction in idealized models, the error correction protocols are designed to return precisely to the ground space of the Hamiltonian, as we have considered above. A clear motivation for such a protocol is that the system could be initialized by cooling to a state exactly within the ground space, and the error correction procedure is fine-tuned to return to a state within this space. In a more realistic system, any imprecision in the implementation or knowledge of the Hamiltonian, or any restriction on the ability to perform non-local unitary operations perfectly will prevent a return to the ground space (i.e.~the overlap between the recovery space and the Hamiltonian ground space will generically decay exponentially quickly in the system size). One might imagine that after performing a recovery operation that returns to a state ``close'' to the ground space, the experimentalist can again cool to an exact ground space state, but for 2D topologically ordered systems it is not clear that this can be done without performing a logical error.
	
	As such, when considering topological error correction in realistic systems, there are several kinds of error correction protocols that may be of interest. The first is as described above, where the system is initialized into the ground space, and the error correction algorithm returns to this space. The second involves again initializing the memory into the ground space, whereas the error correction protocol attempts to return not to this initial space, but to another code space. A reasonable candidate for this space in some circumstances is the ground space of an RG fixed point of the same phase. The third class of error correction protocols is more suited to multi-round error correction, and does not initialize in the ground space of the Hamiltonian. Instead it initializes in the same space in which it intends to return (such as the RG fixed-point ground space).
	
	It may be of interest to consider the properties of protocols falling within these other error correction strategies, depending on the physical application in mind. Furthermore, if the intended setup has a code space other than the ground space, then it is not only of interest to understand the degeneracy splitting and local indistinguishability of the ground space, but also the dephasing rates and code properties of this other space. A fundamental question to consider is then the role of the Hamiltonian in these kinds of topological quantum information storage protocols. Although the Hamiltonian would normally provide some level of passive error suppression through the presence of a mass gap, this is no longer obviously the case if we do not always encode information in the ground space of the model. Nonetheless, this is presumably the practically relevant scenario for an experimentalist attempting to use a real topologically ordered system for information storage in the laboratory.
			
An obvious extension of this work is to consider the honeycomb model as a quantum memory in the full fault-tolerant setting, where faulty measurements are allowed. It would also be very interesting to simulate the model outside the perturbative regime and/or in the non-abelian phase of the model. Though it is not obvious how to perform such simulations efficiently, recent progress in tensor network descriptions of the honeycomb model~\cite{Schmoll2017} and tensor network simulations of noise and error correction~\cite{darmawan2016tensor} may be a fruitful avenue for such research.

We believe that many of the results in this paper may give leading clues to behavior in more general approximate topological codes. The structure of the honeycomb model has provided us a starting point to investigate more realistic quantum memories than the exact toric code, and developing methods to analyze more general systems is an important step in understanding topological quantum error correction in realistic devices.


\begin{acknowledgments}
We thank Jens Eisert, Min-Hsiu Hsieh, Tobias Osborne, David Poulin, and Ted Yoder for valuable discussions.
This research was supported in part by Perimeter Institute for Theoretical Physics. Research at Perimeter Institute is supported by the Government of Canada through the Department of Innovation, Science and Economic Development Canada and by the Province of Ontario through the Ministry of Research, Innovation and Science.
CGB was supported by the ERC grant QFTCMPS and by the cluster of excellence EXC 201 Quantum Engineering and Space-Time Research. 
STF was supported by the Australian Research Council via EQuS project number CE11001013 and by an Australian Research Council Future Fellowship FT130101744.

This is an author-created, un-copyedited version of an article published in the Journal of Statistical Mechanics: Theory and Experiment. IOP Publishing Ltd is not responsible for any errors or omissions in this version of the manuscript or any version derived from it. The Version of Record is available online at doi.org/10.1088/1742-5468/aa7ee2.
\end{acknowledgments}

\bibliography{honeycomb}

\end{document}